\documentclass{article}

\usepackage{authblk}

\usepackage{amsmath,amssymb,amsthm}
\usepackage{booktabs}
\usepackage{multirow}
\usepackage{makecell}
\usepackage{enumerate}
\usepackage{enumitem}
\usepackage{graphicx}
\graphicspath{{figs/}}
\usepackage{url}
\usepackage{color,subcaption}
\usepackage[top=1in, bottom=1in, left=1.5in, right=1.5in]{geometry}
\usepackage[ruled, vlined]{algorithm2e}
\usepackage{caption}
\usepackage{tikz}
\usepackage[counterclockwise]{rotating}
\usepackage[T1]{fontenc}
\usepackage{microtype}

\usepackage{thmtools}
\usepackage{thm-restate}
\usepackage{nameref}
\usepackage{stackengine}
\usepackage{hyperref}

\usepackage{booktabs}

\usepackage[round]{natbib}

\usepackage[most]{tcolorbox}

\newtheorem{theorem}{Theorem}
\newtheorem{proposition}{Proposition}
\newtheorem{lemma}{Lemma}

\theoremstyle{definition}
\newtheorem{definition}{Definition}[section]
\newtheorem{example}{Example}[section]

\newcommand{\indep}{\perp \!\!\! \perp}

\newcommand{\EE}{\mathbb{E}}

\makeatletter
\def\letterdef#1#2#3{\def\letterdef@##1{\expandafter\def\csname #1\endcsname{#2}}%
  \letterdef@@#3{?\@car{}}\@nil}
\def\letterdef@@#1{\@gobble#1\letterdef@{#1}\letterdef@@}
\makeatother

\letterdef{c#1}{\mathcal{#1}}{ABCDEFGHIJKLMNOPQRSTUVWXYZ}

\newcommand{\FDP}{\textnormal{FDP}}
\newcommand{\FDR}{\textnormal{FDR}}
\newcommand{\BH}{\textnormal{BH}}
\newcommand{\SU}{\textnormal{SU}}

\newcommand{\Bonf}{\textnormal{Bonf}}
\newcommand{\IndBH}{\textnormal{IndBH}}
\newcommand{\MT}{\textnormal{MT}}
\newcommand{\BY}{\textnormal{BY}}

\DeclareMathOperator{\Ind}{Ind}

\newcommand{\Nio}{{N_i^\circ}}

\newcommand{\Njo}{{N_j^\circ}}

\newcommand{\NAo}{{N_A^\circ}}
\newcommand{\NKo}{{N_K^\circ}}

\newcommand{\bone}{\boldsymbol{1}}

\newcommand{\oneNio}{\bone^{\Nio}}

\newcommand{\betaai}{\beta_{\alpha, i}}

\newcommand{\piktozero}[1]{p^{(i, \##1) \gets 0}}

\definecolor{light1}{HTML}{Fee0d2}
\definecolor{medium1}{HTML}{Fc9272}
\definecolor{dark1}{HTML}{de2d26}

\definecolor{light2}{HTML}{deebF7}
\definecolor{medium2}{HTML}{9ecae1}
\definecolor{dark2}{HTML}{3182bd}

\definecolor{overlay12}{HTML}{9c7465}

\definecolor{light0}{HTML}{F2F0F7}
\definecolor{mlight0}{HTML}{9e9ac8}
\definecolor{medium0}{HTML}{6a51a3}
\definecolor{hisat0}{HTML}{4a00F3}
\definecolor{dark0}{HTML}{4a1486}

\usetikzlibrary{math}
\usetikzlibrary{fadings}
\usetikzlibrary{positioning}

\usetikzlibrary{graphs}
\usetikzlibrary{patterns,patterns.meta}
\usetikzlibrary{calc}
\usetikzlibrary{decorations.pathreplacing}

\tikzset{
  NEhatch/.style={
    pattern={Lines[angle=45,            
                  distance=4.5pt,        
                  line width=0.5pt]},    
    pattern color=gray!70               
  }
}

\tikzset{
  NWhatch/.style={
    pattern={Lines[angle=-45,            
                  distance=4.5pt,         
                  line width=0.5pt]},     
    pattern color=gray!70              
  }
}

\tikzset{
  mydots/.style={
    pattern={Dots[angle=-30,            
                  distance=2.5pt]},     
    pattern color=gray!70              
  }
}

\tikzset{
  doublehatch/.style={
    NEhatch,                      
    preaction={fill, mydots}    
  }
}

\begin{tikzfadingfrompicture}[name=myfade]
  \clip (0,0) rectangle (2,2);
  \shade [left color=transparent!0,
          right color=transparent!0]
         (0,0) rectangle (1.4,1.4);
  \shade [left color=transparent!0,
          right color=transparent!100]
         (1.4,0) rectangle (1.85,1.4);
  \shade [left color=transparent!100,
          right color=transparent!100]
         (1.85,0) rectangle (2,1.4);
  \shade [bottom color=transparent!0,
  top color=transparent!100]
  (0,1.4) rectangle (2,1.85);
  \shade [bottom color=transparent!100,
  top color=transparent!100]
  (0,1.85) rectangle (2,2);
\end{tikzfadingfrompicture}

\title{Controlling the false discovery rate under a non-parametric graphical dependence model}
\author{Drew T. Nguyen and William Fithian}
\date{\today}

\begin{document}

\maketitle

\begin{abstract}
    We propose sufficient conditions and computationally efficient procedures for false discovery rate control in multiple testing when the $p$-values are related by a known \emph{dependency graph}---meaning that we assume independence of $p$-values that are not within each other's neighborhoods, but otherwise leave the dependence unspecified. Our methods' rejection sets coincide with that of the Benjamini--Hochberg (BH) procedure whenever there are no edges between BH rejections, and we find in simulations and a genomics data example that their power approaches that of the BH procedure when there are few such edges, as is commonly the case. Because our methods ignore all hypotheses not in the BH rejection set, they are computationally efficient whenever that set is small. Our fastest method, the IndBH procedure, typically finishes within seconds even in simulations with up to one million hypotheses.    
\end{abstract}


\section{Introduction}

The false discovery rate (FDR) is a popular error metric in large-scale multiple testing. Given independent, or arbitrarily dependent, $p$-values corresponding to null hypotheses, it can be controlled at a level $\alpha \in (0,1)$ by the Benjamini--Hochberg (BH) procedure, though arbitrary dependence requires an adjustment to the $\alpha$ level known as the Benjamini--Yekutieli (BY) correction. Though there are now many ways to control FDR, the BH procedure remains the most frequently used, due to its simplicity, speed, and availability in software. The impact of \citet{benjaminiControllingFalseDiscovery1995} is felt widely, but especially in ``-omics'' sciences such as genomics, transcriptomics, and proteomics. 

In all these applications, it is rare to apply the BY correction for arbitrary dependence. This owes at least partly to a belief that the correction, originally derived by \citet{benjaminiControlFalseDiscovery2001}, makes BH over-conservative in practical settings, and so is not worth the significant loss in power. Indeed, the full BY correction is only \emph{known} to be necessary under specially crafted settings involving strong, unusual correlations. Simulation studies typically show that, under more realistic conditions, uncorrected BH does control FDR---for example, the work of \citet{kimEffectsDependenceHighdimensional2008}, which is inspired by the analysis of microarrays. 

Unfortunately, finite-sample theoretical results are unavailable for many real-world dependence settings. Previous analyses have focused on a restrictive set of dependence conditions, such as positive regression dependence or parametric dependence. So the claim that BH is safe to use in practice is not, yet, rigorously founded---a state of affairs also remarked upon by \citet{suFDRLinkingTheorem2018} and \citet{chi2022multiple}, both of whom call it ``unsettling''. After all, the use of uncorrected BH throughout the scientific literature is, implicitly, a bet on this claim of safety. 

Our contribution to the literature is a way to control FDR when the analyst cannot assume full independence among all the $p$-values, but is willing to assume some of the $p$-values are independent of some others. We formalize this with the notion of a \emph{dependency graph} for the $p$-values. In practice, analysts can often safely assume a considerable degree of independence, constituting a sparse dependency graph, while any leftover dependence does not require the analyst to specify a parametric form. Given this graph, we propose a new procedure called Independent Set BH, or IndBH. It controls FDR, reduces to BH in the special case of full independence, and with further computational effort, its power can be improved. 

\begin{figure}[tbp]
    \centering
\includegraphics[width = \textwidth]{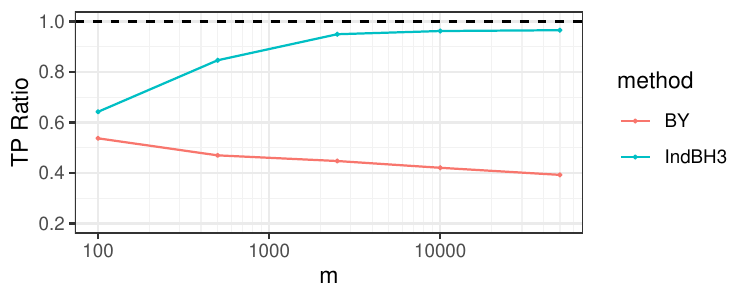}\caption{The expected true positive ratio, relative to BH, of our $\IndBH^{(3)}$ method against the $\BY$ correction ($\alpha = 0.1$). We test the hypotheses $H_i: \mu_i = 0$ against two-sided alternatives with $X \sim \cN_m(\mu, \Sigma)$, with non-null proportion $10 \%$ and non-null means $\mu_i = 3$ distributed uniformly at random, where  $\Sigma$ has equicorrelated blocks ($\rho = 0.5$). For every $m$, the blocks are of fixed size $100$, so the dependency graph becomes sparse for large $m$.}
\label{fig:powercomp}
\end{figure}

Figure \ref{fig:powercomp} illustrates the performance of one such improvement, the $\IndBH^{(3)}$ procedure, under a dependency graph model with equicorrelated blocks of test statistics. When the graph is sparse, our method has nearly the same power as the BH procedure, whereas the BY correction of $L_m := \sum_{j = 1}^m (1/j)$ is overly conservative. (This may seem unfair because BY is not graph-aware, but even a graph aware correction is likely an overcorrection.) Computationally, all of our methods at level $\alpha$ can be implemented by running them at the adjusted level $\alpha R/m$ on just the $p$-values for hypotheses in the BH rejection set. On real data, this set is often small, so this means they can scale to typical data problems with a very large number of hypotheses.

In Section \ref{sec:fdrctrl}, we present our nonparametric dependency graph model and give sufficient conditions for procedures to guarantee FDR control that extend the conditions of \citet{blanchardTwoSimpleSufficient2008b}. Section \ref{sec:method} details our main examples of FDR controlling procedures, including several that are implemented in the R package \texttt{depgraphFDR}. Section \ref{sec:computation} outlines our computational strategy. Section \ref{sec:experiments} illustrates our method in simulations and real data, and Section \ref{sec:discussion} concludes. 

\subsection{Prior work.}

We study multiple testing procedures with finite-sample FDR control using a set of $m$ $p$-values, assuming a nonparametric restriction on their dependency structure. \citet{benjaminiControlFalseDiscovery2001} initiated this line of work by studying the BH procedure, with their results subsequently generalized by \citet{sarkarResultsFalseDiscovery2002a} and \citet{blanchardTwoSimpleSufficient2008b}, 
among others. With few exceptions, however, these results typically focus only on three dependence conditions: independence, positive regression dependence (PRD), and arbitrary dependence, as originally introduced by \citet{benjaminiControlFalseDiscovery2001}. Both independence and PRD are restrictive, and arbitrary dependence is overly pessimistic for most applications.

We study a setting which differs from these three conditions. Several other works have also recently done so; \citet{guoAdaptiveControlsFWER2020} study block dependence, \citet{chi2022multiple} study negative dependence, and the approach of \citet{fithianConditionalCalibrationFalse2020b} is applicable to specific parametric and nonparametric settings. \citet{sarkarControllingFalseDiscovery2023a} showed that a variant of BH controls FDR on two-sided $p$-values for multivariate Gaussian means, when the covariance matrix is known up to a constant. 

Our approach is connected to concepts described by \citet{blanchardTwoSimpleSufficient2008b} and \citet{fithianConditionalCalibrationFalse2020b}, as we discuss further in Section \ref{sec:discussion}.

Asymptotic approaches provide a complementary way to study the FDR of the BH procedure under weak dependence. One example is \citet{farcomeniResultsControlFalse2007a}, who show that BH controls FDR as the number of hypotheses $m \to \infty$, if the dependence decays quickly enough in $m$. However, it can be unclear whether the dependence is weak enough, relative to $m$, for the asymptotic regime to apply to the data problem at hand. On the other hand, our procedures can take any known dependency graph as a hyper-parameter, control FDR for all finite $m$, and reduce to BH when the dependency graph has no edges. 


\section{FDR control in dependency graphs}\label{sec:fdrctrl}

\subsection{Multiple testing and the FDR}
\label{sec:multitestsetup}

The multiple testing problem is to control a form of type I error while rejecting a subset $S \subset \{1, \dots, m\}$ of null hypotheses $H_1, \ldots, H_m$, based on data $X \sim P$. Generically, we take a \emph{multiple testing procedure} $\textup{MT}(\alpha)$ to be represented by some rejection function, denoted $\cR^\textup{MT}_\alpha: X \mapsto S$, where $\alpha \in [0,1]$ will come to denote the level of the procedure.

Let $\cP$ be a family of distributions such that $P \in \cP$, and each null hypothesis $H_i \subsetneq \cP$ to be a subfamily. Denote the true null hypotheses $\cH_0(P) = \{i: P \in H_i\}$ to be the subfamilies that actually contain $P$, and let $m_0 = |\cH_0|$ be the number of true nulls. 

In the present work, we assume that $X$, the data used for testing, comes in the form of $p$-values $p = (p_1, \dots, p_m) \in [0,1]^m$, such that $p_i$ is stochastically larger than $\text{Unif}[0,1]$, or \emph{super-uniform}, whenever $i \in \cH_0$. The procedure $\cR$ is then a set-valued function on $[0,1]^m$. 

A notion of type I error, the \emph{false discovery rate} (FDR) of a procedure $\textup{MT}(\alpha)$, is defined by \citet{benjaminiControllingFalseDiscovery1995} as
\[
\FDR_P(\cR^\MT_\alpha) = \EE_{p \sim P}\big[\FDP(\cR^\MT_\alpha(p))\big] \quad \text{where} \quad \FDP(\cR^\MT_\alpha(p)) = \frac{|\cR^\MT_\alpha(p) \cap \cH_0|}{|\cR^\MT_\alpha(p)|}
\]
with the convention that $0/0 = 1$ inside the expectation. (We omit the dependence on $P$ if it is clear from context.) 

In multiple testing, a standard goal is to design a procedure with high power which satisfies $\sup_{P \in \cP} \FDR_P(\cR^\textup{MT}_\alpha) \leq \alpha$ for every fixed level $\alpha \in [0,1]$ or, more simply put, that $\FDR(\cR^\textup{MT}_\alpha) \leq \alpha$ whenever the distribution of $p$ satisfies certain assumptions encoded by $\cP$, such as mutual independence or positive regression dependence. We then say that $\operatorname{MT}(\alpha)$ \emph{controls} $\FDR$ at level $\alpha$ under these assumptions.

\subsection{Testing in the dependency graph model}

In this work, we formalize ``local dependence'' via the dependency graph model, as in previous work on
central limit theorems \citep{chenNormalApproximationLocal2004} and
Hoeffding-type inequalities \citep{jansonLargeDeviationsSums2004} under local dependence.

Recall that the $p$-value vector $p$ has length $m$, and let $\mathbb{D}$ be an undirected graph with nodes $\{1, \dots, m\}$. For each $i$, let $N_{\mathbb{D},i} \subset \{1, \dots, m\}$ denote the {\em neighborhood} of node $i$ in $\mathbb{D}$, i.e. the set of nodes with edges to $i$. We suppress the dependence on $\mathbb{D}$ if it is clear from context. If $A \subset \{1, \dots, m\}$, let $A^\mathsf{c} = \{1, \dots, m \} \setminus A$. 

\begin{definition}\label{def:dependencygraph}
    $\mathbb{D}$ is a \emph{dependency graph} for the model $\cP$ if, whenever $P \in \cP$ and we draw $p \sim P$, $p_i$ is independent of $(p_j)_{j \in N_{\mathbb{D},i}^\mathsf{c}}$ for every $i$.
\end{definition}

When $\mathbb{D}$ is a dependency graph, we call $N_i$ the \emph{dependency neighborhood} of node $i$. Note that any dependency graph must trivially include all self-edges since $p_i$ cannot be independent of itself. We generally ignore these in the exposition, and in our figures, but the reader may assume they are always present; in particular, the dependency neighborhood $N_i$ always includes $i$ itself.

The dependency graph, used as an input to our methods, encodes known independence; a \emph{sparse} graph says that the existing dependence is weak. Note that if $\mathbb{D}'$ is a dependency graph for $\cP$, then any denser graph $\mathbb{D}$ is also a dependency graph for the same $\cP$, where denser means that the set of edges of $\mathbb{D}$ contains that of $\mathbb{D}'$. This means that our methods achieve FDR control even when the graphs are not the sparsest possible, though power improves from having a sparser graph. 

Later, we will also need the graph-theoretic notion of an \emph{independent set} in $\mathbb{D}$, meaning a set $I \subset \{1, \dots, m\}$ such every distinct $i, j \in I$ are
not connected by an edge in $\mathbb{D}$, and we say $I$ is \emph{maximal} if it is not strictly contained by any other independent set. 

This setting of dependency graphs must be distinguished from the more common assumption of undirected graphical models. In a graphical model, the graph $\mathbb{D}$ instead encodes \emph{conditional} independence, and in particular, the neighborhood $N_i$ satisfies $p_i \indep p_{N_i^\mathsf{c}}$ \emph{conditionally} on $p_{N_i \setminus \{i\}}$. Dependency graphs involve marginal independence and can be more relevant in generic multiple testing, when the goal is not explicitly prediction or model selection.

These definitions imply that, when $\mathbb{D}$ is a dependency graph, the $p$-values $(p_i)_{i \in I}$ are mutually independent whenever $I$ is an independent set, so the probabilistic and graph-theoretic uses of ``independent'' coincide. 

Before proceeding, we settle some notation to state our later results. 

\subsubsection{Notation}

First, we define an important modification of the neighborhoods. If $\mathbb{D}$ is a graph with neighborhoods $N_1, \dots, N_m$, let $\Nio := N_i \setminus \{i\}$ denote the ``punctured neighborhood'' of $i$. This contains
all nodes in $i$'s neighborhood except $i$ itself. We will often mask the $p$-values corresponding to $\Nio$ in the following sense: 
 Let $\bone^A: [0,1]^m \to [0,1]^m$ be the $p$-value transformation setting the indices in $A$ to $1$, and leaving the rest unchanged, that is:
\[
  (\bone^A p)_i := \begin{cases}
                        1 & i \in A \\
                        p_i & i \in A^\mathsf{c}
                  \end{cases}
\]
This represents masking the entries $A$ of the $p$-value vector $p$.  Such masking 
can eliminate dependencies between $p_i$ and the rest of $p$, since $p_i$ is independent of $\bone^{N_i} p$. Often, we will see $A = \Nio$, which masks all the $p$-values in the dependent neighborhood of a given $p_i$, except $i$ itself. 

We make use of the shorthands
\begin{gather*}
    p_{-i} := (p_1, \dots, p_{i-1}, p_{i+1}, \dots, p_m), \\
    S_i := (p_j)_{j \in N_i^\mathsf{c}}.
\end{gather*}
In particular, when $\mathbb{D}$ is a dependency graph, note that $p_i$ is independent of $S_i$. 

For $p$-value vectors $p$ and $q$, we use the symbol~$\preceq$ to mean element-wise comparison, meaning $p \preceq q$ if $p_i \leq q_i$ for each $i$. For a subset $B \subset [m]$, let $\mathbb{D}[B]$ be the subgraph of $\mathbb{D}$ induced by restricting to nodes that are also in $B$, and if $x$ is a vector in $\mathbb{R}^m$, let $x_B$ denote the sub-vector with indices in $B$. 

Let $2^{[m]}$ be the power set of $[m] := \{1, \dots, m\}$. Throughout, $\cR_\alpha: [0,1]^m \to 2^{[m]}$ represents a general multiple testing procedure that takes $p$-values as input and returns a rejection set. If it is a specific multiple testing procedure, say $\cR_\alpha^{\textup{MT}}$, we will also refer to it as $\operatorname{MT}(\alpha)$ For example, $\BH(\alpha)$ is the BH procedure at level $\alpha$, also written $\cR^\BH_\alpha$, and $\operatorname{Bonf}(\alpha)$ is the Bonferroni correction, also written $\cR^\text{Bonf}_\alpha$. Occasionally, we will also use as input a $p$-value vector of length different than $m$, which will be clear from context.

\begin{figure}[ht]
    \centering

    \begin{minipage}[t]{0.96\textwidth}
    \centering
    \begin{tikzpicture}[scale=1.5, every node/.style={circle, draw, minimum size=10mm, font=\Large}]
      \node (1) at (1.5,0.8) {1};
      \node[fill=gray!30] (2) at (1.5,-0.8) {2};
      \node[fill=gray!30] (3) at (3,0) {3};
      \node (4) at (4.5,0.8) {4};
      \node (5) at (4.5,-0.8) {5};

      \draw (3) -- (1) -- (2) -- (3);
      \draw (3) -- (4);
      \draw (3) -- (5);
      \draw[thick, dashed] ($(1)$) circle[radius=0.44cm];

      \node[above=-10pt of 1, draw=none, inner sep=0pt, font=\normalsize] {$p_1$'s node};
      \node[below=-10pt of 2, draw=none, inner sep=0pt, font=\normalsize] {$p_2$'s node};
      \node[below=-10pt of 3, draw=none, inner sep=0pt, font=\normalsize] {$p_3$'s node};
      \node[above=-10pt of 4, draw=none, inner sep=0pt, font=\normalsize] {$p_4$'s node};
      \node[below=-10pt of 5, draw=none, inner sep=0pt, font=\normalsize] {$p_5$'s node};  
    \end{tikzpicture}

    \end{minipage}
    \caption{Some notation depicted on a dependency graph with $m = 5$ $p$-values, which will serve as a running example in this paper. We circled the node $i = 1$ and shaded the nodes for $\Nio = \{2,3\}$. The statistic $S_i = (p_4, p_5)$ corresponds to nodes outside of these.} 
    \label{fig:notation}
\end{figure}
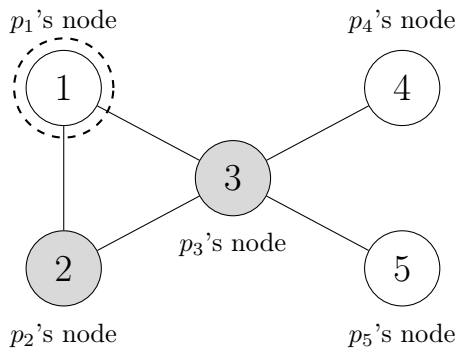

Figure \ref{fig:notation} depicts some of the important notation so far. 

\subsection{Simple sufficient conditions for FDR control}

With only three simple conditions, we can obtain FDR control under the dependency graph. We say that $\cR_\alpha$ is a \emph{$\mathbb{D} $-adapted procedure} if it satisfies these conditions for all $p$. 

\begin{enumerate}[label=(P\arabic*)]
    \item \label{item:SC} (Self-consistency.) If $i \in \cR_\alpha(p)$, then $p_i \leq \alpha |\cR_\alpha(p)|/m$. 
    \item \label{item:Mon} (Monotonicity.) If $p' \preceq p$, then $\cR_\alpha(p) \subseteq \cR_\alpha(p')$.
    \item \label{item:NB} (Neighbor-blindness.) $i \in \cR_\alpha(p) \iff i \in \cR_\alpha(\oneNio p)$. 
\end{enumerate}

Property \ref{item:SC} was first introduced in a more general form by 
\citet{blanchardTwoSimpleSufficient2008b}, and property \ref{item:Mon} 
has been studied for FDR control since at least \citet{tamhane2008weighted}. The condition 
\ref{item:NB} is new, giving a clean solution to FDR control in our setting.

\subsubsection{Main results}\label{sec:Dadaptiveproof}

\begin{theorem}\label{thm:control} Suppose $\cR_\alpha : [0,1]^m \to 2^{[m]}$ is a $\mathbb{D}$-adapted multiple testing procedure, and that $p_i$ is superuniform whenever $i \in \cH_0$. Then whenever $\mathbb{D}$ is a dependency graph for $p$, the procedure $\cR_\alpha$ controls the FDR at level $\alpha |\cH_0|/m$ on $p$. 
\end{theorem}
Even though the result is implied by Theorem \ref{thm:pprdcontrol}, 
we provide a self-contained proof. It demonstrates
how any $\mathbb{D}$-adapted procedure compares each $p_i$ to a threshold 
which depends only on 
$p$-values independent of $p_i$. 
\begin{proof}
Let $\cR = \cR_\alpha$. Because 
    \[
    \FDR = \sum_{i \in \cH_0} \EE\left[
    \frac{1\{i \in \cR(p)\}}{|\cR(p)|}
        \right],
    \]
it is sufficient to show that 
\begin{equation}\label{eq:toshow} 
        \EE\left[
    \frac{1\{i \in \cR(p)\}}{|\cR(p)|}
        \right]
        \leq \alpha/m
\end{equation}
for an arbitrary $i \in \cH_0$. Let $q = \oneNio p$, which depends on $p$ only through $p_i$ and $S_i$, so that $q_i = p_i$, 
and $q_{-i}$ is a function of $S_i$ alone.

We first characterize the rejection event in terms of local thresholds, using all three of
the conditions. Observe by \ref{item:Mon} that there exists a function $\beta_i: [0,1]^{m-1} \to \mathbb{R}$ satisfying
    \begin{equation}\label{eq:condtest}
        i \in \cR_\alpha(p) \Leftrightarrow p_i \leq \frac{\alpha \beta_i(p_{-i})}{m}, 
    \end{equation}
    where by \ref{item:SC},
    \begin{equation}\label{eq:rejlowerbound}
       \beta_i(p_{-i}) \leq |\cR_\alpha(p)|  \quad  \text{ whenever } i \in \cR_\alpha(p),
    \end{equation}
    and by \ref{item:NB}, 
    \begin{equation}\label{eq:nbthresh}
       \beta_i(p_{-i}) = \beta_i(q_{-i}).
    \end{equation}
In other words, $\beta_i(p_{-i})$ represents a lower bound on the number of rejections we eventually make, on the event where $H_i$ is rejected, and it depends only on the other $p$-values $p_{-i}$. In fact, \eqref{eq:nbthresh} implies that $\beta_i(p_{-i})$ depends only on $S_i$, the $p$-values {\em outside of $i$'s dependency neighborhood}, and is therefore independent of $p_i$.

As a result, we have that for $i \in \cH_0$,
\begin{align*}
     \EE\left[
    \frac{1\{i \in \cR(p)\}}{|\cR(p)|}
        \right]
    &\leq \EE\left[
    \frac{1 \{i \in \cR(p)\}}{|\cR(q)| } 
            \right] &&\text{\ref{item:Mon}}\\ 
    &\leq \EE\left[
    \frac{1 \{p_i \leq \alpha \beta_i(q_{-i}) \big/ m\}}{\beta_i(q_{-i}) } 
        \right] &&\text{Eqs. } \eqref{eq:condtest}-\eqref{eq:nbthresh} \\
    &= \EE\left[ \EE\left[
    \frac{1 \{p_i \leq \alpha \beta_i(q_{-i}) \big/ m\}}{\beta_i(q_{-i}) } 
         \mid S_i \right] \right] &&\text{(tower rule)} \\
    &\leq \EE\left[  
    \frac{\alpha \beta_i(q_{-i}) \big/ m }{\beta_i(q_{-i})}  \right]   &&\text{(super-uniformity)} \\
    &= \alpha/m.
\end{align*}
\end{proof}

Next, recall from the multiple testing literature that $p$ is \emph{positive regression dependent} (PRD) on a set $A \subset [m]$ if for any \emph{increasing} set $K \subset \mathbb{R}^m$, meaning one such that $p \in K,\, p \preceq p'$ implies $p' \in K$, the probability $\mathbb{P}(p \in K \mid p_i = t)$ is increasing in $t$ whenever $i \in A$. 

We say that $p$ is \emph{partially positive regression dependent} (PPRD) on $A$ with respect to the graph $\mathbb{D}$ if, for any increasing set $K \subset \mathbb{R}^{|N_i^\mathsf{c}|}$,  the probability $\mathbb{P}(p_{N_i^\mathsf{c}} \in K \mid p_i = t)$ is increasing in $t$ whenever $i \in A$, a strictly weaker assumption than PRD for any $\mathbb{D}$. Now we can strengthen the result of Theorem \ref{thm:control}. 

\begin{restatable}{theorem}{pprd}\label{thm:pprdcontrol} 
Suppose $\cR_\alpha : [0,1]^m \to 2^{[m]}$ is a $\mathbb{D}$-adapted multiple testing procedure, and that $p_i$ is superuniform whenever $i \in \cH_0$. Then whenever the input vector $p$ is PPRD on $\cH_0$ with respect to $\mathbb{D}$, the procedure $\cR_\alpha$ controls the FDR at level $\alpha |\cH_0|/m$. 
\end{restatable}
These two results generalize a standard pair of results from the literature, as discovered by \citet{benjaminiControlFalseDiscovery2001} and generalized by \citet{blanchardTwoSimpleSufficient2008b}: namely, that the BH procedure controls FDR under independence and positive regression dependence. 

The proof is deferred to Appendix \ref{app:proofs}, using a related 
leave-one-out technique based on the superuniformity lemma of \citet{ramdasUnifiedTreatmentMultiple2019a}.

\subsubsection{Remarks}
\label{sec:fdrctrlremarks}

\paragraph{Dependency graphs in practice.} Given a vector of $p$-values, a practitioner can often identify a dependency graph by simply identifying, for each $p_i$, which $p_j$'s are independent of $p_i$. For example, in spatiotemporal settings, independent $p$-values are distant in space or time. A particularly useful special case is \emph{block dependence}, when there exist blocks of $p$-values with mutual independence between but not necessarily within blocks, see e.g. assumption 2(a) of \citet{ignatiadisCovariatePoweredCrossweighted2021}, which is equivalent to a dependency graph $\mathbb{D}$ whose connected components are all cliques. Such blocks arise in genome-wide testing, corresponding to $p$-values from different chromosomes, or microplate assays, corresponding to $p$-values from different plates. 

\paragraph{A BY correction for the graph?} A seemingly natural approach would be to find the worst-case FDR of the $\BH(\alpha)$ procedure under a dependency graph, and use this to correct the $\alpha$-level of the BH procedure. But we find that the bounds we are able to derive are, like BY, overly conservative, while our preferred methods typically deliver finite-sample guarantees with good power. A tight estimate would nevertheless be of independent interest. Define $\mathcal{P}_\mathbb{D}$ be the family of distributions $P$ for which $\mathbb{D}$ is a dependency graph, and let the worst-case FDR be
\begin{equation}\label{eq:worstcasefdr}
    \gamma_\mathbb{D}(\alpha) := \sup \big\{\FDR_P(\cR^\BH_\alpha) : P \in \mathcal{P}_\mathbb{D}\big\}.
\end{equation}
We do not compute $\gamma_\mathbb{D}(\alpha)$, but in Appendix \ref{app:BYD} can be found upper and lower bounds on it, accompanied by some discussion. 

\paragraph{Local thresholds.} Equation \eqref{eq:condtest} says that any monotone procedure can be viewed as a thresholding procedure on every $p$-value separately. Together with \eqref{eq:rejlowerbound} and \eqref{eq:nbthresh}, we see that $\mathbb{D}$-adapted procedures simply ``reject $H_i$ whenever $p_i \leq c_i(S_i)$'', for a certain function $c_i$. We call $c_i(S_i) = \alpha \beta_i(p_{-i})/m$ the procedure's $i$th local threshold, where $\beta_i(p_{-i})$ only depends on $p_i$ through $S_i = p_{N_i^\mathsf{c}}$ by \eqref{eq:nbthresh}. We also call $\beta_i(p_{-i})$ the procedure's $i$th \emph{rejection lower bound}~(r.l.b.), as it satisfies equation \eqref{eq:rejlowerbound}. This local threshold representation leads to an effective computational strategy for our procedures (see Section \ref{sec:computation}).

\paragraph{Self-consistency gap.} The bound \eqref{eq:rejlowerbound} can be quite loose, especially if the r.l.b. $\beta_i(p_{-i})$ tends to be far from $|\cR_\alpha(p)|$. In that case, the underlying procedure is conservative. The gap in \eqref{eq:rejlowerbound} can be interpreted as a ``self-consistency gap''---the looseness with which the self-consistency criterion is satisfied. This gap is some sense optimizable, allowing us to increase the local threshold and mitigate the conservatism (see Section \ref{sec:gapchase}).

\subsection{Parametric examples}

Though assumptions on dependency graphs can often be made without a specific
parametric model for the $p$-values, a few parametric models nevertheless 
fall naturally into our framework, as we now describe. 
\begin{example}[Multivariate $z$-statistics]
    If $Z \sim \mathcal{N}_m(\mu, \Sigma)$ is a multivariate normal random vector
of length $m$, then if $\Sigma$ is known, then the graph $\mathbb{D}$ 
given by $N_i = \{j: \Sigma_{ij} \neq 0\}$ is a dependency graph for $Z$, and therefore also for one- or two-sided $p$-values based on the coordinates of $Z$. Additionally, 
$Z$ is PPRD with respect to the graph with neighborhoods $N_i = \{j: \Sigma_{ij} < 0\}$, and so are one-sided $p$-values based on its coordinates.  
\end{example}

\begin{example}[Multivariate $t$-statistics]
    Now suppose that $Z \sim \mathcal{N}_m(\mu, \sigma^2 \Psi)$, where $\Psi$ is
    known and $\sigma^2$ is not, but we can use the independent random variable 
    $\nu {\hat\sigma}^2 \sim \sigma^2 \chi^2_{\nu}$ to estimate it. Then 
    for $T_i = Z_i / \sqrt{\Psi_{i,i} \hat\sigma^2}$, we justify the following result in Appendix \ref{app:proofs}. 
    
    \begin{restatable}{proposition}{pprdtstat}\label{prop:pprdtstat}
    Let $\mathbb{D}$ be the graph with neighborhoods 
    $N_i = \{j : \Psi_{ij} \neq 0\}$ for $i = 1, \dots, m$. Then the test statistics $T^2 = (T_1^2, \dots, T_m^2)$, 
    are PPRD on $[m]$ with respect to $\mathbb{D}$.
    \end{restatable}
    Consequently, two-sided $p$-values based on the coordinates of $T^2$ are PPRD with respect to $\mathbb{D}$.
\end{example}

\begin{example}[Linear models]
    The Gaussian linear model for a fixed design matrix $\mathbf{X} \in \mathbb{R}^{n \times m}$ is a special case of the previous example. We observe a response $\mathbf{y} \sim \mathcal{N}_n(\mathbf{X} \beta, \sigma^2 I_m)$ where $\beta \in \mathbb{R}^m, \sigma^2 > 0$ are 
    both unknown, and supposing that $\mathbf{X}$ satisfies $n > m$ 
    with full column rank, we compute the ordinary least squares estimate and residual
    sum of squares
    \[
        \hat \beta =  
            (\mathbf{X}^\top \mathbf{X})^{-1} \mathbf{X}^\top \mathbf{y} 
            \sim \mathcal{N}_m(\beta, \sigma^2 (\mathbf{X}^\top \mathbf{X})^{-1}), 
        \quad \text{and} \quad 
        (n-m) \hat \sigma^2 = \texttt{RSS} = 
            \|\mathbf{y} - \mathbf{X} \hat \beta \|^2 \sim \sigma^2 \chi_{n - m}^2.
    \]
    Then using the previous example, we know that for 
    $T_i = \hat \beta_i / \sqrt{(\mathbf{X}^\top\mathbf{X})^{-1}_{i,i} \hat\sigma^2}$, 
    the test statistic $T^2$ is PPRD on $[m]$ with respect to the graph defined
    by $N_i = \{j: (\mathbf{X}^\top\mathbf{X})^{-1}_{i,j} = 0\}$. This may be reasonable in an approximate sense; when the rows of $\mathbf{X}$ are iid with finite second moment, the entries of $(\mathbf{X}^\top\mathbf{X})^{-1}$ go to zero as $n$ grows. 
\end{example}


\section{Graph-adapted procedures}\label{sec:method}

Theorem \ref{thm:control} guarantee FDR control for procedures adapted to the underlying dependency graph. This section presents a family of graph-adapted procedures that trade off statistical power with computational efficiency. Note that, while we focus our exposition on independence relations, any graph-adapted procedure also controls FDR under partial positive dependence by Theorem \ref{thm:pprdcontrol}.

Before presenting our first procedure, we recall some preliminaries. Recall that the BH procedure at level $\alpha$, which we refer to as $\BH(\alpha)$, returns a rejection set according to the following rule:
Let $p_{(1)} \leq \cdots \leq p_{(m)}$ be the $p$-value order statistics. Then  
\begin{equation}\label{eq:BHdef}
    \cR^{\BH}_\alpha(p) := \{i : p_i \leq \alpha r^*/m \}, \quad \text{where} \quad r^* := \max\{r : p_{(r)} \leq \alpha r/m \}.
\end{equation}
The Bonferroni correction at level $\alpha$, abbreviated $\Bonf(\alpha)$, is simply
\[ 
    \cR^{\Bonf}_\alpha(p) := \{i : p_i \leq \alpha /m \}.
\]
These are in fact $\mathbb{D}$-adapted procedures, but for special choices of the graph $\mathbb{D}$: $\BH(\alpha)$ is adapted (only) for the \emph{empty graph} where $N_i = \{i\}$ for all $i$, and $\Bonf(\alpha)$ is adapted to any graph, but in particular the \emph{complete graph} where $N_i = \{1, \dots, m\}$ for all $i$. Indeed, each of the procedures below will reduce to BH when $\mathbb{D}$ is the empty graph, and to Bonferroni when $\mathbb{D}$ is the complete graph.

Between these two extremes, we will introduce several graph-adapted procedures 
for a general graph $\mathbb{D}$.
A first idea would be 
to recall, from the theory of multiple testing, how we can express BH via the 
local thresholds
\[
    \cR^\BH_\alpha(p) = \{i : p_i \leq \alpha |\cR^\BH_\alpha(p^{i \gets 0})|/ m\}.
\]
where $p^{i \gets 0}$ is equal to $p$ except with $p_i$ replaced by 0. Then, given a dependency graph $\mathbb{D}$, we could adjust the BH procedure's local threshold for $p_i$ by masking the $p$-values in its neighborhood, leading to the method:
\[
    \cR^{\textup{naiv}_\mathbb{D}}_\alpha(p) = \{i : p_i \leq \alpha |\cR^\BH_\alpha(\bone^{\Nio} p^{i \gets 0})|/ m\}.
\]
In calculating $p_i$'s threshold, recall that the operator $\bone^{\Nio}$ masks all the $p$-values in $p_i$'s neighborhood $N_i$, except for $p_i$ itself, which makes the local threshold independent of $p_i$---a property that BH has under full independence.

This procedure interpolates between $\BH$ and Bonferroni, as we wished, but unfortunately, it is not self-consistent according to \ref{item:SC}. Even worse, it does not control FDR, as the following proposition shows (see Appendix \ref{app:naive} for the proof and further intuition).
\begin{restatable}{proposition}{naivefails}\label{prop:naivefails}
    There exists a distribution $P$ on $m = 3$ 
    $p$-values with dependency graph $\mathbb{D}$ such that $\FDR_P(\cR^{\textup{naiv}_\mathbb{D}}) > \alpha$. 
\end{restatable}
This does not mean that our sufficient conditions are necessary, and in fact they are not (see Appendix \ref{app:notnecessary}). However, it does 
demonstrate that simply modifying BH may not lead to FDR control---new ideas will be necessary. For us, this will be the graph-theoretic notion of independent sets. 

We now move to our first procedure satisfying the sufficient conditions of Section \ref{sec:fdrctrl}.

\subsection{IndBH}\label{sec:indbh}
Recall that an independent set in $\mathbb{D}$ is a subset $I \subset [m]$ of nodes, where no two distinct nodes are connected. Let $\Ind(\mathbb{D})$ be the collection of independent sets in $\mathbb{D}$. 

The Independent Set BH procedure (at level $\alpha$, for graph $\mathbb{D}$), or $\IndBH_\mathbb{D}(\alpha)$, simply re-runs BH on a masked $p$-value vector for every independent set in $\mathbb{D}$ and unions the result. Formally, it is defined by
\begin{equation}\label{eq:indbhdef}
    \cR_\alpha^{\IndBH_\mathbb{D}}(p) := \bigcup_{I \in \Ind(\mathbb{D})} \mathcal{R}_\alpha^\BH(\bone^{I^\mathsf{c}} p),
\end{equation}
and we suppress $\mathbb{D}$ in the notation if it is clear from context. This controls FDR, as we show next. 
\begin{proposition}\label{prop:IndBH1}
    The $\IndBH_\mathbb{D}(\alpha)$ procedure is $\mathbb{D}$-adapted, and hence controls FDR at level $\alpha$ when $\mathbb{D}$ is a dependency graph for $p$. 
\end{proposition}
\begin{proof}
First observe that, for any independent set $I \subseteq \mathbb{D}$,  the procedure $\mathcal{R}_\alpha^\BH(\bone^{I^\mathsf{c}} p)$ is $\mathbb{D}$-adapted. Thus, it suffices to show that a union of $\mathbb{D}$-adapted procedures is $\mathbb{D}$-adapted. To this
    end, let $\cR^{(1)}_\alpha$ and $\cR^{(2)}_\alpha$ be $\mathbb{D}$-adapted, and let $\cR^\textup{un}_\alpha(p) = \cR^{(1)}_\alpha(p) \cup  \cR^{(2)}_\alpha(p)$. We will show $\cR^\textup{un}_\alpha$ is $\mathbb{D}$-adapted.
    
    Monotonicity and neighborblindness of $\cR^\textup{un}_\alpha(p)$ can be seen directly, while self-consistency holds because 
    \[
    i \in \cR^\textup{un}_\alpha(p) \Rightarrow p_i \leq \frac{\alpha}{m} \max\{|\cR_\alpha^{(1)}(p)|, |\cR_\alpha^{(2)}(p)|\} \Rightarrow p_i \leq \frac{\alpha}{m} |\cR^\textup{un}_\alpha(p)|. 
    \]
    where the first implication is by self-consistency of both $\cR_\alpha^{(1)}$ and $\cR_\alpha^{(2)}$. Hence $\cR^\textup{un}_\alpha$ is $\mathbb{D}$-adapted. 
\end{proof}
Note that the collection
of all independent sets $\IndBH(\mathbb{D})$ can be replaced with the collection of maximal independent
sets only without changing the procedure. Computing maximal independent sets 
is NP-hard, but as mentioned in the Introduction, we can ignore all $p$-values not in the BH rejection set, so on real data, this problem is typically tractable, with further computational tricks described in Section \ref{sec:computation}.

Let us now understand this method intuitively. When the graph $\mathbb{D}$ is sparse, 
the (maximal) independent sets can be large, making this procedure
powerful. When $\mathbb{D}$ is completely sparse, i.e. the empty graph, then $[m]$ is the only maximal independent set, recovering $\BH(\alpha)$, and when $\mathbb{D}$ is completely dense, i.e. the complete graph, the only independent sets are the singletons $(\{i\})_{i = 1}^m$, recovering $\Bonf(\alpha)$. 

When working by hand, a slightly different way to 
think about the rejection set can be useful.
Let us say a non-empty independent set $C$ is a 
\emph{certificate set} for IndBH if $p_j \leq \alpha |C|/m$ for all $j \in C$. This 
holds if and only if $C \subset \mathcal{R}_\alpha^\BH(\bone^{I^\mathsf{c}} p)$ for some $I \in \Ind(\mathbb{D})$, 
so that equivalently, $\IndBH(\alpha)$ rejects $H_i$ whenever there exists 
a certificate set $C$ containing $i$. For small graphs, we can go node by node
and look for such a $C$ to compute rejections. 

As an example, consider the $p$-values and dependency graph in Figure \ref{fig:indbh-certificate} with $m = 5$ p-values. When running $\IndBH(\alpha)$ with $\alpha = 0.05$, there are three certificate sets, namely $\{1, 4\}$, $\{2, 4\}$, and $\{3\}$. These ``certify'' the rejection of $H_1, \dots, H_4$, but $H_5$ remains unrejected (even though it would be rejected by $\BH(\alpha)$). If $p_3$ were increased beyond $0.01$, $H_3$ would
not be rejected either, as $\{3\}$ would no longer be a certificate set.

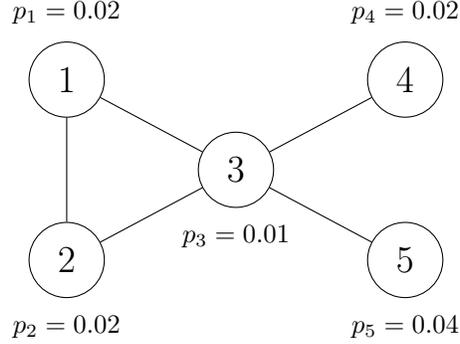
\begin{figure}[ht]
    \centering

    \begin{minipage}[t]{0.96\textwidth}
    \centering
    \begin{tikzpicture}[scale=1.5, every node/.style={circle, draw, minimum size=10mm, font=\Large}]
      \node (1) at (1.5,0.8) {1};
      \node (2) at (1.5,-0.8) {2};
      \node (3) at (3,0) {3};
      \node (4) at (4.5,0.8) {4};
      \node (5) at (4.5,-0.8) {5};

      \draw (3) -- (1) -- (2) -- (3);
      \draw (3) -- (4);
      \draw (3) -- (5);

      \node[above=-10pt of 1, draw=none, inner sep=0pt, font=\normalsize] {$p_1 = 0.02$};
      \node[below=-10pt of 2, draw=none, inner sep=0pt, font=\normalsize] {$p_2 = 0.02$};
      \node[below=-10pt of 3, draw=none, inner sep=0pt, font=\normalsize] {$p_3 = 0.01$};
      \node[above=-10pt of 4, draw=none, inner sep=0pt, font=\normalsize] {$p_4 = 0.02$};
      \node[below=-10pt of 5, draw=none, inner sep=0pt, font=\normalsize] {$p_5 = 0.04$};  
    \end{tikzpicture}
    \end{minipage}
    \caption{An example realization of $m = 5$ $p$-values on the graph from Figure \ref{fig:notation}. If $\IndBH(0.05)$ were run with this graph and $p$-values, then it would reject $H_1, \dots H_4$, while $\BH(0.05)$ would reject everything.}
    \label{fig:indbh-certificate}
\end{figure}

Could there be some way to reject $H_5$?
Computing the certificate sets $\{1, 4\}$ and $\{2,4\}$ did not 
depend on any $p$-values from $p_5$'s neighborhood, so perhaps we can leverage
these sets together, rather than separately, to reject $H_5$ while still retaining neighbor-blindness. 
We discuss this next.

\subsection{Improving IndBH}\label{sec:improving}
We now describe ways to improve $\IndBH$ with additional computation. 
Most important is Section \ref{sec:indbhk}, which describes an approach
we can run in practice. Section \ref{sec:stepup} characterizes the 
optimal $\mathbb{D}$-adapted procedure, which we do not run in 
practice. All proofs are deferred to Section \ref{sec:gapchase}. 

\subsubsection{$\text{IndBH}^{(k)}$}\label{sec:indbhk}
For $k \geq 1$, we define the $\text{IndBH}^{(k)}(\alpha)$ procedure by the rejection set
\begin{equation}\label{eq:indbhkdef}
     \cR_\alpha^{\IndBH^{(1)}}(p) := \cR_\alpha^{\IndBH}(p),\  
    \cR_\alpha^{\IndBH^{(k+1)}}(p) := \left\{i : p_i \leq \frac{\alpha |\{i\} \cup  \cR^{\IndBH^{(k)}}_\alpha(\bone^{\Nio} p)|}{m} \right\}.
\end{equation}
 $\IndBH^{(k+1)}$ amounts to re-running $\IndBH^{(k)}$ on a masked $p$-value vector (to enforce neighborblindness), and plugging the results into to a new threshold for $p_i$. It controls FDR, and always improves over $\IndBH^{(k)}$ by the following proposition.
\begin{proposition}\label{prop:IndBHk}
    The $\IndBH^{(k)}_\mathbb{D}(\alpha)$ procedure is $\mathbb{D}$-adapted for $k \geq 2$, and hence controls FDR at level $\alpha$ when $\mathbb{D}$ is a dependency graph for $p$. Also, we have 
    $\cR_\alpha^{\IndBH^{(k-1)}}(p) \subset \cR_\alpha^{\IndBH^{(k)}}(p)$. 
\end{proposition}
\begin{proof}
    A direct consequence of Theorem \ref{thm:gapchase} from Section \ref{sec:gapchase}. 
\end{proof}
At the bottom of the recursion, $\IndBH^{(k)}$ could make exponentially many calls to $\IndBH$ as $k$ increases, in addition to the complexity already inherent in $\IndBH$ as $m$ increases. This limits our ability to run $\IndBH^{(k)}$ for large $k$. In our simulations, we found  
that running $\IndBH^{(3)}$ targets a nice tradeoff between power and computational effort.

In Figure \ref{fig:indbh2-certificate}, we update Figure \ref{fig:indbh-certificate} to 
show how $H_5$ can be rejected by $\IndBH^{(2)}$, despite being not rejected 
by $\IndBH$. 

\begin{figure}[ht]
    \centering

    \begin{minipage}[t]{0.96\textwidth}
    \centering
    \begin{tikzpicture}[scale=1.5, every node/.style={circle, draw, minimum size=10mm, font=\Large}]
      \node (1) at (1.5,0.8) {1};
      \node (2) at (1.5,-0.8) {2};
      \node[fill=gray!30] (3) at (3,0) {3};
      \node (4) at (4.5,0.8) {4};
      \node (5) at (4.5,-0.8) {5};

      \draw (3) -- (1) -- (2) -- (3);
      \draw (3) -- (4);
      \draw (3) -- (5);
      \draw[thick, dashed] ($(5)$) circle[radius=0.44cm];

      \node[above=-10pt of 1, draw=none, inner sep=0pt, font=\normalsize] {$p_1 = 0.02$};
      \node[below=-10pt of 2, draw=none, inner sep=0pt, font=\normalsize] {$p_2 = 0.02$};
      \node[below=-10pt of 3, draw=none, inner sep=0pt, font=\normalsize] {$p_3 = 0.01$};
      \node[below=+2pt of 3, draw=none, inner sep=0pt, font=\normalsize] {(masked)};      
      \node[above=-10pt of 4, draw=none, inner sep=0pt, font=\normalsize] {$p_4 = 0.02$};
      \node[below=-10pt of 5, draw=none, inner sep=0pt, font=\normalsize] {$p_5 = 0.04$};  
    \end{tikzpicture}
    \end{minipage}
    \caption{The same graph and $p$-values from Figure \ref{fig:indbh-certificate}, but circling $i = 5$, with its punctured neighborhood $N_5^\circ = \{3\}$ shaded.
    Even after masking $p_3$ (i.e. replacing it with $1$), we still have the certificate sets $\{1,4\}$ and $\{2,4\}$, so $\cR^\IndBH_\alpha(\bone^{N_5^\circ}p) = \{1,2,4\}$  at $\alpha = 0.05$. As a result, the iterated procedure $\IndBH^{(2)}(\alpha)$ 
    rejects $H_5$, because $p_5 \leq \alpha|\{1,2,4,5\}|/m = 0.04$.} 
    \label{fig:indbh2-certificate}
\end{figure}
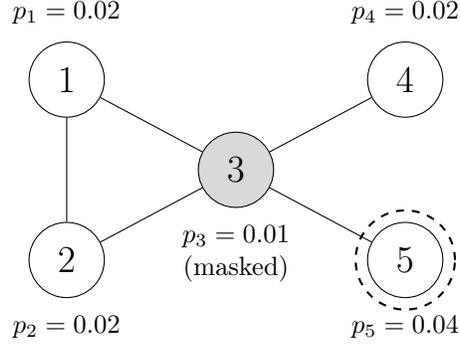

\subsubsection{The optimal procedure}\label{sec:stepup}

It turns out that $\IndBH^{(k)}$ is \emph{not} the optimal $\mathbb{D}$-adapted procedure for 
any $k$, even when taking $k \to \infty$, as we show in Appendix \ref{app:indbh-notoptimal}. (Here, ``optimal'' means ``most liberal''). Instead, consider the iterates
  \begin{equation}\label{eq:SUdef}
    \cR^{(1)}_\alpha = \cR^\BH_\alpha, \quad 
    \cR^{(k+1)}_\alpha(p) := \left\{i: p_i 
    \leq \frac{\alpha |\{i\} \cup \cR^{(k)}_\alpha(\bone^{\Nio} p)|}{m}  \right\}. 
  \end{equation}
  This is the \emph{same iteration} used for $\IndBH^{(k)}$, just with a different starting point. Now, define the \emph{$\mathbb{D}$-adapted step-up procedure} $\SU_\mathbb{D}(\alpha)$ to be the fixed point of this iteration, that is,
\[
    \cR^{\SU_\mathbb{D}}_\alpha(p) := \cR^{(k^*)}_\alpha(p), \quad \text{whenever $k^*$ satisfies $\cR^{(k^*)}_\alpha(p) = \cR^{(k^*+1)}_\alpha(p)$.}
\] 
We do not use the name ``step-up'' to imply ``stepping up'' from the least significant $p$-value to the most significant as, for example, \citet{romanoStepupProceduresControl2006} use the term. Our intent is to align with how \citet{blanchardTwoSimpleSufficient2008b} use it: they mean the most liberal of all self-consistent procedures, and we mean the most liberal of all $\mathbb{D}$-adapted procedures.

The next 
proposition, whose proof we again defer to Section \ref{sec:gapchase}, 
says that such a $k^*$ always exists and that $\SU_\mathbb{D}$ is optimal. 

\begin{proposition}\label{prop:SUopt}
    There exists some $k^* \geq 1$ such that $\cR^{(k^*)}_\alpha(p) = \cR^{(k^*+1)}_\alpha(p)$. The $\SU_\mathbb{D}(\alpha)$ procedure is then $\mathbb{D}$-adapted and optimal in the sense that $\cR^{\SU_\mathbb{D}}_\alpha(p) \supset \cR_\alpha^{\mathbb{D}}(p)$ where $\cR_\alpha^{\mathbb{D}}(p)$ is any other $\mathbb{D}$-adapted procedure. 
\end{proposition}
\begin{proof}
    A direct consequence of Theorem \ref{thm:gapchase} from Section \ref{sec:gapchase}. 
\end{proof}

Letting $\cR^{(k)}$ be defined as in the iterations \eqref{eq:SUdef}, and letting 
$\cR^{\IndBH_\mathbb{D}^{(\infty)}}_\alpha(p)$ be the fixed point of the iterations \eqref{eq:indbhkdef}, 
we hence have the following relations between the procedures defined thus far:
\begin{gather*}
    \cR^{\IndBH_\mathbb{D}}_\alpha(p) \subseteq
    \cR^{\IndBH_\mathbb{D}^{(2)}}_\alpha(p) \subseteq 
    \dots \subseteq
    \cR^{\IndBH_\mathbb{D}^{(\infty)}}_\alpha(p) \subseteq 
    \cR^{\SU_\mathbb{D}}_\alpha(p), \text{ and} \\
    \cR^{\SU_\mathbb{D}}_\alpha(p) \subseteq 
    \dots \subseteq 
    \cR^{(2)}_\alpha(p) \subseteq 
    \cR^{(1)}_\alpha(p) = \cR^\BH_\alpha(p).
\end{gather*}
The first line of procedures are all FDR controlling under $\mathbb{D}$. Note that the inclusion 
$\cR^{\IndBH_\mathbb{D}^{(\infty)}}_\alpha(p) \subseteq 
\cR^{\SU_\mathbb{D}}_\alpha(p)$ can be strict.

We do not implement $\SU_\mathbb{D}$, as there can always be cases 
where $k^*$ is too large to be reasonably computed, and stopping the iteration early does not yield a $\mathbb{D}$-adapted procedure. 
But for any $k \geq 2$ one can always 
cap the number of iterations to be at most $k$ and still obtain FDR control from by randomly pruning the set $\cR^{(k)}(p)$, an idea which comes from \citet{fithianConditionalCalibrationFalse2020b}. In light of our empirical observation that the non-randomized procedure $\IndBH^{(3)}(\alpha)$ seems to perform nearly as well as $\BH(\alpha)$ in practical settings, we do not recommend  using a randomized procedure; we refer the reader to Appendix \ref{app:randomprune} for details on the randomized approach. 

\subsection{The gap chasing update}\label{sec:gapchase}

In this section, we do not propose new methods. Instead, we 
prove and discuss the results on the shared iteration in \eqref{eq:indbhkdef}
and \eqref{eq:SUdef}, which we call the ``gap chasing'' update. 

\begin{theorem}[Gap chasing]\label{thm:gapchase}
    For some initial monotone multiple testing procedure $\cR^{(1)}_\alpha$, consider the 
    iterates
    \begin{equation}\label{eq:gapchasingiterate}
        \cR^{(k+1)}_\alpha(p) := \left\{i: p_i 
        \leq \frac{\alpha |\{i\} \cup \cR^{(k)}_\alpha(\bone^{\Nio} p)|}{m}  \right\}. 
    \end{equation}
    for some dependency graph $\mathbb{D}$ and its neighborhoods $N_i$, and 
    let $\cR^{\mathbb{D}}_\alpha$ be any $\mathbb{D}$-adapted procedure.

    \begin{enumerate}[label=(\alph*)]
        \item \label{item:gapchaseD}   When $\cR^{(1)}_\alpha = \cR^{\mathbb{D}}_\alpha$ then for $k \geq 1$, 
        $\cR^{(k)}_\alpha(p)$ is $\mathbb{D}$-adapted, and 
        $\cR_\alpha^{(k)}(p) \subset \cR_\alpha^{(k+1)}(p)$. 
        \item \label{item:gapchaseBH}     When $\cR^{(1)}_\alpha = \cR^\BH_\alpha$, then for $k \geq 1$,
        $\cR_\alpha^{(k)}(p) \supset \cR^\mathbb{D}_\alpha$, and $\cR_\alpha^{(k)}(p) \supset \cR_\alpha^{(k+1)}(p)$. Also, there exists $k^* \geq 2$ such that $\cR_\alpha^{(k^*)}(p) = \cR_\alpha^{(k^*+1)}(p)$, and $\cR_\alpha^{(k^*)}(p)$ is  $\mathbb{D}$-adapted.
    \end{enumerate}
\end{theorem}

\begin{proof}
    \emph{\ref{item:gapchaseD}.} It suffices to prove the result for $k=1$. Let $\cR_\alpha := \cR^{(1)}_\alpha$ and $\cR^+_\alpha := \cR^{(2)}_\alpha$. 
        First, observe that
        because $\cR_\alpha$ is monotone, $\cR^+_\alpha$ is also monotone, and 
        also $\cR^+_\alpha$ is neighborblind by construction. Next, note that, by neighbor-blindness and self-consistency of $\cR_\alpha$, we 
        have
        \begin{align*}
            i \in \cR_\alpha(p) 
                &\Leftrightarrow i \in \mathcal{R}_\alpha(\bone^{\Nio} p) \\
                &\Rightarrow p_i \leq \alpha |\mathcal{R}_\alpha(\bone^{\Nio} p)|/m \\
                &\Rightarrow p_i \leq \alpha |\{i\} \cup\mathcal{R}_\alpha(\bone^{\Nio} p)|/m,
        \end{align*}
        implying that
        \begin{equation}\label{eq:gapchase-inclusion}
            \cR_\alpha(p) \subset \cR^+_\alpha(p).
        \end{equation}
        We can use this to show the self-consistency:
        \begin{align*}
            i \in \cR^+_\alpha(p) 
                &\Leftrightarrow p_i \leq \alpha |\{i\} \cup\mathcal{R}_\alpha(\bone^{\Nio} p)|/m, \  \text{and} \  i \in \cR^+_\alpha(p) \\
                &\Rightarrow p_i \leq \alpha |\{i\} \cup \mathcal{R}^+_\alpha(\bone^{\Nio} p)|/m, \  \text{and} \  i \in \cR^+_\alpha(p) \\
                &\Leftrightarrow p_i \leq \alpha | \mathcal{R}^+_\alpha(\bone^{\Nio} p)|/m, \ \text{and} \  i \in \cR^+_\alpha(p).
        \end{align*}
        This shows the first claim, and the second claim was established
        in \eqref{eq:gapchase-inclusion}.

    \emph{\ref{item:gapchaseBH}}. First, suppose for induction
    that $\cR^\mathbb{D}_\alpha(p) \subset  \cR^{(k)}_\alpha(p)$. Then we have
    \[
        \cR^\mathbb{D}_\alpha(p)  \subset \left\{i: p_i 
        \leq \frac{\alpha |\{i\} \cup \cR^\mathbb{D}_\alpha(\bone^{\Nio} p)|}{m}  \right\}
        \subset   \left\{i: p_i 
        \leq \frac{\alpha |\{i\} \cup \cR^{(k)}_\alpha(\bone^{\Nio} p)|}{m}  \right\} = \cR^{(k+1)}_\alpha(p). 
    \]
    By self-consistency of $\cR^\mathbb{D}$ and BH, the base case $\cR^\mathbb{D}_\alpha(p) \subset  \cR^{(1)}_\alpha(p)$ 
    holds, showing the first claim. 

    Next, suppose for induction that 
    $\cR^{(k-1)}_\alpha(p) \supset \cR^{(k)}_\alpha(p)$. Then
    \[
        \cR^{(k)}_\alpha(p) = \left\{i: p_i 
        \leq \frac{\alpha |\{i\} \cup \cR^{(k-1)}_\alpha(\bone^{\Nio} p) |}{m}  \right\}
        \supset   \left\{i: p_i 
        \leq \frac{\alpha |\{i\} \cup \cR^{(k)}_\alpha(\bone^{\Nio} p)|}{m}  \right\} 
        = \cR^{(k+1)}_\alpha(p).
    \]
    The base case holds by monotonicity of $\BH$, showing the second claim. 

    Finally, because 
    $\cR_\alpha^{(k)}(p) \supset \cR_\alpha^{(k+1)}(p)$ for all $k$ and 
    all sets are finite, 
    there must be a $k^*$ such that $\cR_\alpha^{(k^*)}(p) = \cR_\alpha^{(k^*+1)}(p)$. 
    For every $k \geq 2$, the iterates $\cR_\alpha^{(k)}$ are neighbor-blind by definition 
    and monotone by induction, and whenever $i \in \cR_\alpha^{(k^*)}(p)$,
    \[
        \cR_\alpha^{(k^*)}(p) = \left\{i: p_i 
        \leq \frac{\alpha |\{i\} \cup \cR^{(k^*)}_\alpha(\bone^{\Nio} p)|}{m}  \right\} \subset 
        \left\{i: p_i 
        \leq \frac{\alpha |\cR^{(k^*)}_\alpha(p)|}{m}  \right\}
    \]
    so $\cR_\alpha^{(k^*)}$ is $\mathbb{D}$-adapted, showing the last claim.
\end{proof}

Intuitively, the gap chasing update leading to $\cR_\alpha^{(k+1)}$ in 
\eqref{eq:gapchasingiterate} enforces neighborblindness by masking the $p$-value vector, 
while reducing the ``self-consistency gap'' of $\cR_\alpha^{(k)}$ described in Section 
\ref{sec:fdrctrlremarks}. In the case of full independence (empty dependency graph), the fixed point $\cR^{(k^*)}_\alpha$ of the iteration
satisfies
\[
    \cR_\alpha^{(k^*)}(p) =
            \left\{i : p_i \leq \frac{\alpha |\{i\} \cup \cR^{(k^*)}_\alpha(p)|}{m}  \right\} = 
            \left\{i : p_i \leq \frac{\alpha |\cR^{(k^*)}_\alpha(p)|}{m}  \right\}.
\]
Hence property \ref{item:SC} is satisfied with logical equivalence, and the self-consistency gap has been removed.
If $\cR^{(1)}_\alpha = \cR^{\Bonf}_\alpha$, the fixed point under full independence is
\begin{equation}\label{eq:stepdownBHdef}
    \cR^{\BH^-}_\alpha(p) := \{i : p_i \leq \alpha r^*/m \}, \quad \text{where} \quad r^* := \max\{r : p_{(k)} \leq \alpha k/m \;\;\forall\,k \leq r \},
\end{equation}
the ``step-down'' variant of the BH procedure, defined by e.g. \citet{finnerFalseDiscoveryRate2001}.

Finally, we briefly note that Theorem \ref{thm:gapchase}\ref{item:gapchaseD} is a consequence
of general closure properties for $\mathbb{D}$-adapted procedures. 
An investigation of these is relegated to Appendix \ref{app:closure}.



\section{Computation}\label{sec:computation}

In this section, we summarize our computational strategy: full details appear in Appendix \ref{app:computation}. We focus mostly on the computation of $\IndBH$, since $\IndBH^{(k)}$ may call $\IndBH$ recursively many times. 

Computing $\IndBH$ involves listing the maximal independent sets of graphs, which is NP-hard (because it solves the clique decision problem, which is NP complete). Unfortunately, we do need the \emph{exact} results---approximations do not suffice for FDR control. Though this may seem intractable, we have already mentioned that $\IndBH$ can ignore all hypotheses outside the BH rejection set. With the additional tweaks in this section, $\IndBH$, $\IndBH^{(2)}$, and even $\IndBH^{(3)}$ seem to be tractable on real data. They ultimately rely on off-the-shelf implementations of the independent set routines, but we call them only when needed. 

We now settle some notation for this section. Take $\mathtt{LargestInd}(\mathbb{K})$ to denote \emph{any} largest independent in a graph $\mathbb{K}$---for us, it does not matter which---and let us define the sublevel set $Q(r)$, and a certain independent set $I_i(r)$ that contains $i$:
\begin{gather}
    Q(r) = \{j : p_j \leq \alpha r/m\},  \text{ and } \\
    I_i(r) = \{i\} \cup \mathtt{LargestInd}(\mathbb{D}[Q_{-i}(r)]) \quad \text{where} \quad Q_{-i}(r) = \{j \notin N_i : p_j \leq \alpha r/m\}, \label{eq:indepsetwithi}
\end{gather}
suppressing dependence on the $p$-values and $\alpha$. This just says that $Q(r)$ is that index set with $p$-values smaller than $\alpha r/m$, and  $I_i(r)$ is the largest independent set containing $i$ among the graph nodes $Q(r)$. 

As an important preliminary, recall that any $\mathbb{D}$-adapted procedure has 
a local threshold representation as described in Section \ref{sec:fdrctrlremarks}. In the case of IndBH, 
we can write the following proposition, proved in Appendix \ref{app:proofs}. Recall that 
$\mathbb{D}[B]$ is the induced subgraph of $\mathbb{D}$ from nodes $B \subset [m].$ 

\begin{restatable}{proposition}{indbhrep}\label{prop:indbhrep}
    $\IndBH(\alpha)$ rejects $H_i$ if and only if $p_i \leq \alpha \beta_{\alpha,i}^\IndBH(p_{-i})/m$, where
    \begin{gather*}
        \beta_{\alpha,i}^\IndBH(p_{-i}) := \max \{r :  |I_i(r)| \geq r \}.
    \end{gather*}
    \end{restatable}
This representation
directly associates IndBH with a computational strategy: for every $i$, 
compute $I_i(r)$ for all $r$ using an algorithm that implements $\texttt{LargestInd}$ on a subgraph. By testing on each $H_i$ separately, this strategy is also parallelizable (though we did not use the parallelized versions to measure running time). We will refine this strategy
in the ensuing sections.

\subsection{Reducing the graph size}\label{sec:graphreduce}

Our first major computational saving is to discard all the $p$-values $p_i$ where $H_i$ was not rejected by BH, reducing to the 
graph $\mathbb{D}[\cR^\BH_\alpha(p)]$ and its $p$-values. Running IndBH with this graph and $p$-values, at the adjusted level $\alpha |\cR^\BH_\alpha(p)|/m$, gives the same rejection set. Essentially, this is possible because IndBH is a union of BH rejection sets, and BH ignores all $p$-values not in its rejection set.

In practice, it is common that most hypotheses are null---a sparse 
signal pattern. In this case, $\cR^\BH_\alpha(p)$ makes very few rejections relative to $m$,
allowing us to drastically shrink the size of the problem. We now state a formal and slightly more general version of this claim. 
\begin{restatable}{proposition}{ignorebh}\label{prop:ignorebh}
    Suppose we have $\bar r$ such that whenever $p_i > \alpha \bar r/m$, the
    $\IndBH^{(k)}_\mathbb{D}(\alpha)$ procedure fails to reject $H_i$. 
    Then $\IndBH^{(k)}_\mathbb{D}(\alpha)$ 
    run on the $p$-values $(p_i)_{i = 1}^m$ is equivalent to $\IndBH^{(k)}_{\mathbb{D}'}(\alpha')$
    run on $(p_i)_{i \in Q(\bar r)}$, using the subgraph $\mathbb{D}' = \mathbb{D}[Q(\bar r)]$ and adjusted level $\alpha' = \alpha |Q(\bar r)|/m$. 
    In particular, this holds for $\bar r = |\cR^\BH_\alpha(p)|$, which has 
    $Q(\bar r) = \cR^\BH_\alpha(p)$. 
\end{restatable}

Because of this reduction, the size of the BH rejection set $\cR^\BH_\alpha(p)$ is more important than the number of hypotheses $m$ for the computational efficiency of our methods. For example, let us perturb slightly the setting of Figure \ref{fig:powercomp}, which tests the hypotheses $H_i : \mu_i = 0$ under Gaussian block dependence, where the blocks have equal size 100 and $\mu_i = 3$ for every non-null $H_i$. We will check the effect of the size of $\cR^\BH_\alpha(p)$ on the runtime. 

Suppose there are 10\% randomly placed non-nulls among $m = 2 \cdot 10^5$ hypotheses. Then on an M1 Macbook Pro, our implementation of $\IndBH(0.1)$ runs in about a second on average, and $\IndBH^{(3)}(0.1)$ in about 7 minutes. On the other hand, if the non-null percentage is only 1\%, then even when $m = 10^6$, $\IndBH(0.1)$ can run in about a second, but $\IndBH^{(3)}(0.1)$ takes only about five seconds. In the latter case, the $\BH$ rejection set is small, and our methods see the accompanying speedup. 

By Proposition \ref{prop:ignorebh}, any time the input graph $\mathbb{D}$ is mentioned from this point, the reader can
safely substitute instead the subgraph $\mathbb{D}' = \mathbb{D}[\cR_\alpha^\BH(p)]$ and its corresponding $p$-values, as long as the subsitutions $m' = |\cR_\alpha^\BH(p)|$ and $\alpha' = \alpha m'/m$ are 
also made for $m$ and $\alpha$. 
\subsection{Speedups based on connected components}

Our second major computational saving comes from the fact that, especially after the reduction of Section \ref{sec:graphreduce}, the resulting graph has many small connected components. 

Let $\mathbb{D}_1, \dots \mathbb{D}_{\bar k}$ be all $\bar k$ connected components of $\mathbb{D}$. 
In this section, we show how to combine computations done separately on each of these connected components. 

\subsubsection{Caching independence numbers from each component}\label{sec:caching}

Let $\kappa[i]$ be the index for the component such that $\mathbb{D}_{\kappa[i]}$ contains $i$ as a node. Crucially, because $I_i(r)$ is a largest independent set, we can write it as a union of largest independent sets from each of its component graphs:
\[
I_i(r) = \{i\} \cup \mathtt{LargestInd}(\mathbb{D}_{\kappa[i]}[Q_{-i}(r)]) \cup   \bigcup_{k \neq \kappa[i]} \mathtt{LargestInd}(\mathbb{D}_k[Q(r)]),
\]
from which it follows that
\[
|I_i(r)| = 1 +  \mathtt{IndNum}(\mathbb{D}_{\kappa[i]}[Q_{-i}(r)]) + \sum_{k \neq \kappa[i]} \mathtt{IndNum}(\mathbb{D}_k[\cQ(r)]),
\]
where $\mathtt{IndNum}(\mathbb{K})$ denotes the size of the largest independent set of a graph $\mathbb{K}$, called its \emph{independence number}. Computing independence numbers is NP-hard with respect to graph size, though it is more efficient with small connected components. 

Importantly, observe that the terms $\mathtt{IndNum}(\mathbb{D}_k[\cQ(r)])$ in the summation do \emph{not} depend on the hypothesis $i \in Q(\bar r)$, so those values can be computed once, cached, and reused for every $i$. Details on the caching can be found in Appendix \ref{app:computation}. 

We still may need to compute $\mathtt{IndNum}(\mathbb{D}_{\kappa[i]}[Q_{-i}( \cdot )])$ for each $i$, but not always, 
as described next.

\subsubsection{Cheap checks from the components}\label{sec:cheapcheck}

We can avoid the computation of $\mathtt{IndNum}(\mathbb{D}_{\kappa[i]}[Q_{-i}(\cdot )])$
most of the time. Assuming we have the cached
values $\mathtt{IndNum}(\mathbb{D}_k[\cQ(r)]) := \mathbf{V}_{k, r}$, define
\begin{gather*}
    \beta^+ = \max\left\{r \leq \bar r: \sum_{k = 1}^{\bar k} \mathbf{V}_{k, r} \geq r\right\}, \text{ and }\\
    \beta^- = \max\left\{r \leq \bar r: \sum_{k = 1}^{\bar k} \mathbf{V}_{k, r} - \max_{k'} \mathbf{V}_{k, r} + 1 \geq r\right\},
\end{gather*}
which satisfy $\beta^+ \geq \betaai^\IndBH(p_{-i})$ and 
$\beta^- \leq \betaai^\IndBH(p_{-i})$. It follows that $\IndBH$ can reject $H_i$ if $p_i \leq \alpha \beta^-/m$ 
and fails to reject $H_i$ if $p_i > \alpha \beta^+/m$. These two checks can be done on all hypotheses at once, 
and takes care of most of them.

Further details on the checks we use can be found in Appendix \ref{app:computation}. 

\subsubsection{When the components are fully connected}\label{sec:fullyconnectedshortcut}

If $\mathbb{D}_k$ is fully connected, then $\mathtt{IndNum}(\mathbb{D}_k[Q(r)])$ is easy to compute; it is equal to 
$1$ if the minimum $p$-value in $\mathbb{D}_k$ is below $\alpha r / m$. Additionally, $\mathtt{IndNum}(\mathbb{D}_{\kappa[i]}[Q_{-i}(r)])$ becomes identically zero, since the graph $\mathbb{D}_{\kappa[i]}[Q_{-i}(r)]$ has no nodes. 

When the graph $\mathbb{D}$ encodes block dependence, every $\mathbb{D}_k$ is fully connected, and we have the following proposition, proven in Appendix \ref{app:proofs}. 
\begin{restatable}{proposition}{indbhclique}\label{prop:indbhclique}
Suppose blocks $B_1, \dots, B_{\bar k} \subset [m]$ 
form a disjoint partition of
$[m]$, and that 
the neighborhoods $(N_i)_{i=1}^m$ of a graph $\mathbb{D}$ satisfy $N_i = B_k$ whenever
$i \in B_k$. 
Defining
\[
    j^*_k = \operatorname*{argmin}_{j \in B_k} p_j  \quad \text{ and } \quad  M = \bigcup_{k = 1}^{\bar k} \{j^*_k\},
\]
we have
\[
    \cR^{\textup{\IndBH}_\mathbb{D}}_\alpha(p) = \{i : p_i \leq \alpha \big| \cR^\BH_\alpha(\bone^{M^\mathsf{c}} p) \big| / m\}.
\]
\end{restatable}
Under block dependence, this says that we can compute IndBH almost as efficiently as BH itself. This form of IndBH was first discovered by \citet{guoAdaptiveControlsFWER2020}, who proved its FDR control in this special case. For us, we note that block dependence is not strictly required to apply this computational shortcut; as long as the components of the graph are cliques after the reduction of Section \ref{sec:graphreduce}, we can apply Proposition \ref{prop:indbhclique}.

When the components of the graph are not cliques, we could treat them as cliques to cheaply compute a large subset of the IndBH rejection set. Though the checks in Section \ref{sec:cheapcheck} are better for this purpose, they are more expensive, so
we use this shortcut to help compute $\IndBH^{(k)}$ as explained in Appendix \ref{app:computation}. 

\subsection{Considerations for $\mathbf{IndBH}^{(2)}$}\label{sec:considerindbh2}

Now we briefly consider computation for $\IndBH^{(2)}$. Similar commentary applies to $\IndBH^{(k)}$ where $k > 2$.

Recall that $\IndBH^{(2)}(\alpha)$ rejects $H_i$ if and only if $p_i \leq \alpha |\{i\} \cup \cR^{\IndBH}_\alpha(\bone^{\Nio} p)|/m$. This naively requires us to compute $\cR^{\IndBH}_\alpha(\bone^{\Nio} p)$ for every $i$, but
we can skip this calculation for most $i$ by making the following observations. 

First, we can compute $|\cR^{\IndBH}_\alpha(p)|$ once and check whether $i \in \cR^{\IndBH}_\alpha(p)$. If it is, then we immediately know that $i \in \cR_\alpha^{\IndBH^{(2)}}(p)$ by Proposition \ref{prop:IndBHk}. Also, if $p_i > \alpha |\{i\} \cup \cR^{\IndBH}_\alpha(p)| / m$, we immediately know that $i \notin \cR_\alpha^{\IndBH^{(2)}}(p)$, by monotonicity of $\cR_\alpha^{\IndBH}$. 

For the remaining hypotheses where we must compute $\cR^{\IndBH}_\alpha(\bone^{\Nio} p)$, we do not need to redo all the independent set calculations from $\cR^{\IndBH}_\alpha(p)$---most of the intermediate computations cached in Section \ref{sec:caching} can be reused. Again, full details are deferred to Appendix \ref{app:computation}. 




\section{Experiments}\label{sec:experiments}

In this section, we have selected four experimental settings to illustrate one of our main points: even though IndBH and its improvements are always less powerful than BH, they tend to be about as powerful when the dependency graph is sparse, while still provably controlling FDR. Each of the first three simulation settings uses a family of dependency structures which gets sparser as the number of hypotheses $m$ grows. 

Let $\cH_1 = [m] \setminus \cH_0$. Aside from the FDR, for a method MT we compute the true positive ratio and rejection ratio as:
\begin{gather*}
    \textup{TP Ratio} = \mathbb{E}\left[\frac{|\cH_1 \cap \mathcal{R}^\textup{MT}_\alpha(p)|}{|\cH_1 \cap \mathcal{R}^\BH_\alpha(p)|} \mid \cH_1 \cap \mathcal{R}^\BH_\alpha(p) \neq \varnothing \right] \\
    \textup{Rej. Ratio} = \mathbb{E}\left[\frac{|\mathcal{R}^\textup{MT}_\alpha(p)|}{| \mathcal{R}^\BH_\alpha(p)|} \mid \mathcal{R}^\BH_\alpha(p) \neq \varnothing \right]
\end{gather*}
Using these metrics, we compare the IndBH and $\IndBH^{(3)}$ procedures to the BH and BY procedures in each of these four settings. Further experiment runs, varying the simulation parameters and comparing additional methods, can be found in Appendix \ref{app:experiments}. 

\subsection{Block dependence with scattered signals.}\label{sec:blockscat}

In Figure \ref{fig:block_uniformly}, we consider a distribution of $p$-values which is block dependent. To specify this distribution, first we set the following simulation parameters: null proportion $\pi_0 \in [0,1]$, target power $\texttt{tarpow} \in [0,1]$, equicorrelation $\rho \in [0,1]$, and block size $b$ which evenly divides every $m$, the number of hypotheses. 

Then for $m \in \{100,500,2500,10000,50000\}$, the two-sided $p$-values are generated based on $X \sim \mathcal{N}_m(\mu, \Sigma)$, with distribution parameters determined as follows. We first choose a subset of hypotheses $\cH_1 \subset [m]$ of size $\lfloor (1 - \pi_0)m\rfloor$ uniformly at random to be non-null. Then we set the mean $\mu \in \mathbb{R}^m$ as $\mu_i = 0$ if $i \notin \cH_i$. For $i \in \cH_i$, if signal strength is \texttt{fixed}, we set $\mu_i = \mu^*$, and if signal strength is \texttt{random}, then $\mu_i \overset{iid}{\sim} \operatorname{Exp}(\mu^*)$ with the mean parameterization. In each case $\mu^* \in \mathbb{R}$ is tuned so that the BH procedure has power $\texttt{tarpow}$ on the ensuing $p$-values. 

We set the covariance matrix $\Sigma$ to have equicorrelated blocks with 
correlation $\rho$:

\[
\Sigma = 
\operatorname{diag}(\underbrace{\Sigma_\rho, \dots, \Sigma_\rho}_{m/b \text{ times}}) \in \mathbb{R}^{m \times m}
\quad \text{where} \quad
\Sigma_\rho = \begin{pmatrix}
1      & \rho  & \dots  & \rho \\
\rho   & 1     & \ddots  & \vdots \\
\vdots   & \ddots  & \ddots    & \rho \\
\rho   & \dots  & \rho  & 1
\end{pmatrix} \in \mathbb{R}^{b \times b}.
\]
Because $b$ is fixed as $m$ increases, this represents a dependency graph that gets sparser as $m$ increases. Finally, we generated two-sided $p$-values for the null hypotheses $H_i = 0$ as $p_i = 2(1 - \Phi(|X_i|))$.

\begin{figure}[htbp]
    \centering
    \includegraphics{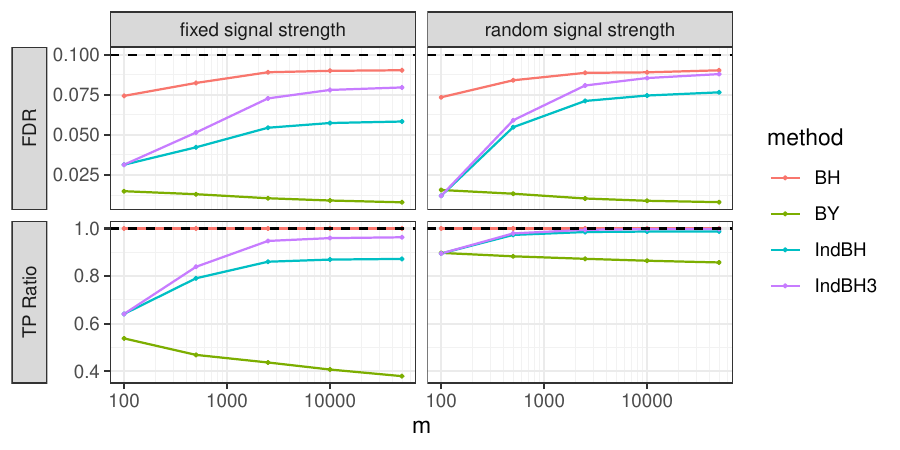}
    \caption{Results for block dependence with scattered signals. FDR control level set to $\alpha = 0.1$. Simulation parameters set to $\pi_0 = 0.9, \texttt{tarpow} = 0.6, \rho = 0.5, b = 100$.}\label{fig:block_uniformly}
\end{figure}

Figure \ref{fig:block_uniformly} shows that in the case of fixed signal strengths, 
the performance of $\IndBH$ improves  as the dependency graph becomes sparser, but does not nearly match the 
performance of BH, as $\IndBH^{(3)}$ does. But fixed signal strengths tend to exaggerate the
difference between multiple testing procedures---when the threshold is near 
$(\alpha/m) \cdot 2(1 - \Phi(\mu^*))$, slightly increasing it can lead to many more rejections. On the other hand, 
random signal strengths lead to all procedures performing more favorably compared to BH, 
and $\IndBH$ performs about as well as BH. 

\subsection{Banded dependence with clustered signals.}\label{sec:bandclust}

In Figure \ref{fig:banded_clustered}, we consider a distribution of $p$-values where the non-null signal pattern is somewhat adversarial to our method. 
Specifically, by first generating cluster centers on $\{1, \dots, m\}$ and drawing non-nulls indices randomly near the cluster centers, we make it so that many non-nulls lie within each other's neighborhoods. Neighborblindness then prevents them from ``assisting'' each other in being rejected using $\mathbb{D}$-adapted procedures. 

The simulation parameters we set are null proportion $\pi_0 \in [0,1]$, target power $\texttt{tarpow} \in [0,1]$, equicorrelation $\rho \in [0,1]$, band size $b'$, average number of points per cluster $\lambda_0 \in [0,1]$, and cluster tightness $\tau > 0$. 

Again, for $m \in \{100,500,2500,10000,50000\}$, the $p$-values are generated based on $X \sim \mathcal{N}_m(\mu, \Sigma)$, with their distribution determined as follows. First, given $\cH_1$, we specify $\mu$ in the same way as in Section \ref{sec:blockscat}, tuning an additional parameter $\mu^*$ so that BH has power $\texttt{tarpow}$. Next, we take $\Sigma$ to be a banded Toeplitz covariance matrix: specifically, $\Sigma_{ij} = \rho^{|i - j|}$ if $|i - j| \leq \lfloor (b' - 1)/2 \rfloor$, and zero otherwise.

Finally, the choice of $\cH_1 \subset [m]$ is now more complicated. We draw this set from the following point process on $\mathbb{Z}$. Let $\operatorname{PP}(\lambda)$ denote a Poisson point process on $[m]$, that takes values as multi-sets of points. Then draw the realized cluster centers $\mathbf{j}$ and daughters $\mathbf{h}$ from
\begin{gather}
    \mathbf{j} \sim \operatorname{PP}(\eta), \quad \eta(i) := \eta_0 1\{i \in [m]\} 
     \\
     \mathbf{h} \sim \operatorname{PP}(\lambda), \quad \lambda(i) := \lambda_0 \sum_{j \in \mathbf{j}} f_\sigma(i, j), \text{ where}\\
     f_\sigma(i, j) := \exp\left(-\frac{1}{2\sigma^2}(i-j)^2 \right) \Big/ \sum_{i \in \mathbb{Z}}\exp\left(-\frac{1}{2\tau^2}(i-j)^2 \right).\label{eq:discretegauss}
\end{gather}
Finally, we get $\cH_1$ from $\mathbf{h}$ by retaining the unique points, while discarding any points not in $\{1, \dots, m\}$. The daughters $\mathbf{h}$ are made up of $n_j \sim \operatorname{Pois}(\lambda_0)$ points in $\mathbb{Z}$ assigned to each point $j$, with locations distributed according to $f_\sigma(\cdot, j)$. Therefore, the parameter $\eta_0$ controls the average number of clusters, $\tau$ the tightness of the cluster, and $\lambda_0$ the average number of daughter points per cluster. Given $\lambda_0$, we set $\eta_0 = (1 - \pi_0) m/\lambda_0$ so that there are roughly $(1 - \pi_0)m$ non-nulls on average. 

Figure \ref{fig:thomasrug} illustrates a few realizations of $\cH_1$. We remark that the similar Thomas point process is used in the ecological sciences to model the locations of trees \citep{wiegandHandbookSpatialPointPattern2013}.

\begin{figure}[htbp]
    \centering
    \includegraphics{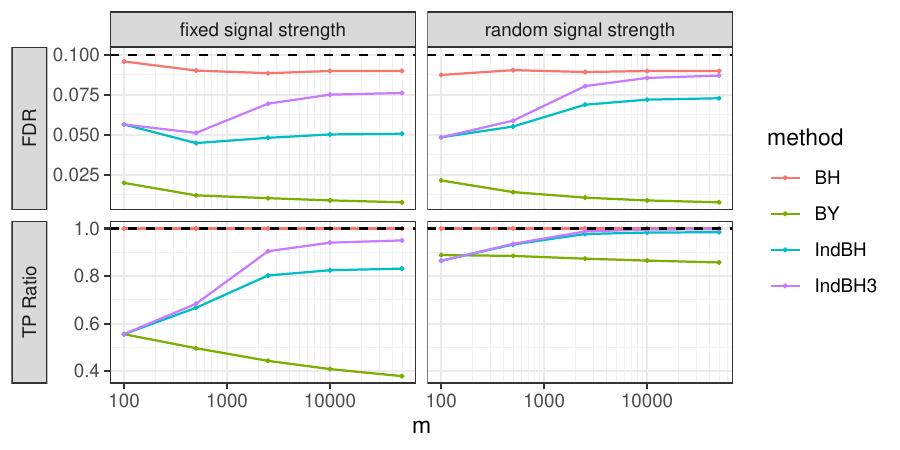}
    \caption{Results for banded dependence with clustered signals. FDR control level set to $\alpha = 0.1$. Simulation parameters set to $\pi_0 = 0.9, \texttt{tarpow} = 0.6, \rho = 0.5, b' = 100, \lambda_0 = 20, \tau = 6$.}\label{fig:banded_clustered}
\end{figure}

\begin{figure}[ht]
    \centering
    \begin{minipage}{0.98\textwidth}
        \includegraphics[width=\linewidth]{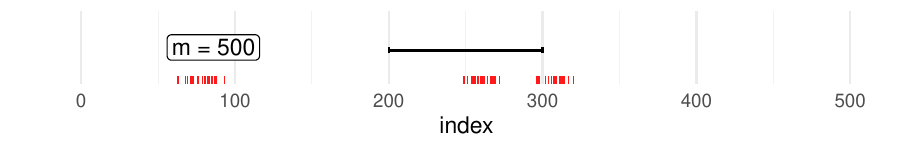}
    \end{minipage}
    \hfill
    \begin{minipage}{0.98\textwidth}
        \includegraphics[width=\linewidth]{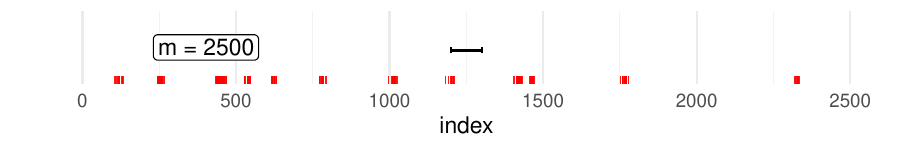}
    \end{minipage}
    \hfill
    \begin{minipage}{0.98\textwidth}
        \includegraphics[width=\linewidth]{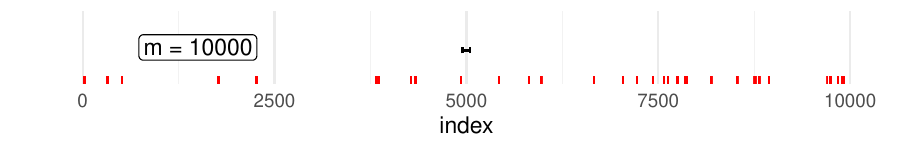}
    \end{minipage}
    \caption{A rug plot of one generation of the non-nulls $\cH_1$ in the simulation setting of Figure \ref{fig:banded_clustered}. A ruler of length $b' = 100$ is also shown, representing the size of the typical 
    neighborhood. 
    }
    \label{fig:thomasrug}
\end{figure}

Due the somewhat adversarial simulation setting, when the signal strength is fixed, both IndBH and $\IndBH^{(3)}$ perform more poorly compared to BH up to $m = 500$, but this graph is quite dense when $b' = 100$ is the size of an average neighborhood. The performance recovers starting at $m = 2500$.

\subsection{Ensuring FDR control.}\label{sec:fdrexpr}

In Figure \ref{fig:inflating_blocks}, we return to block dependence, but now with distributions that significantly inflate the FDR of the BH procedure. The distributions we chose are rather contrived: an ``adversarial'' distribution of $p$-values, defined in Appendix \ref{app:bydlower}, as well as one-sided $p$-values from a negatively correlated Gaussian. In this section, we take the null proportion $\pi_0 = 1$, so the remaining relevant parameters are the block size $b$ and, for the Gaussian, the equicorrelation $\rho$, which must be negative to inflate the FDR. 

We checked the FDR and rejection ratio for $m \in \{3, 6, 9, 18, 27\}$. In all these cases, IndBH and $\IndBH^{(3)}$ seem to behave the same.

\begin{figure}[htbp]
    \centering
    \includegraphics{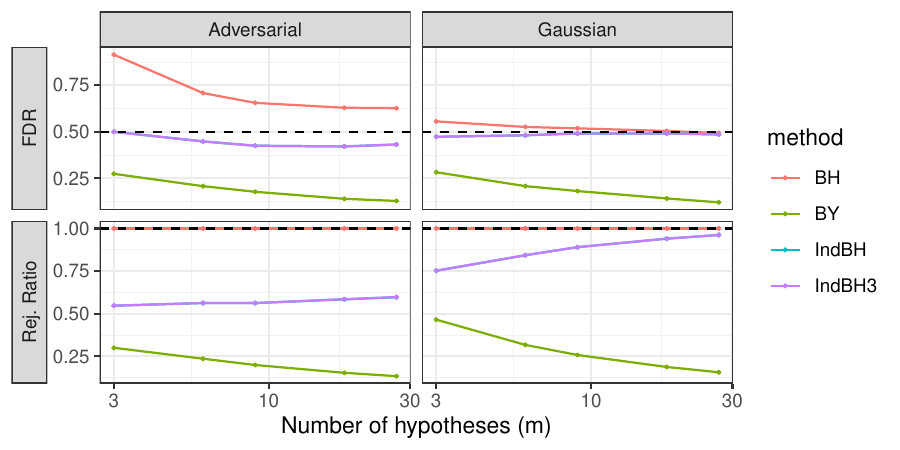}
    \caption{Results with distributions designed to inflate FDR. FDR control level set to $\alpha = 0.5$. Simulation parameters set to $b = 3$ and $\rho = -0.354$.}\label{fig:inflating_blocks}
\end{figure}

\subsection{A real data example.}\label{sec:realexpr}

In this section, we work through an analysis using our methods on genome-wide association study (GWAS) data. Specifically, we use $p$-values from the schizophrenia GWAS of \citet{ripkeBiologicalInsights1082014}. 
Each individual $p$-value encodes associations between an individual single-nucleotide polymorphism (SNP) with schizophrenia. 
In our analysis, we mimicked the approach of 
\citet{yurkoSelectiveInferenceApproach2020a} to subset to only those SNPs which 
appeared as eQTLs in the BrainVar study of \citet{werlingWholeGenomeRNASequencing2020}, 
i.e. those SNPs which were associated with gene expression in the developing brain. The raw dataset contained over 8 million $p$-values, each representing a SNP, but after 
filtering for $\text{INFO score} > 0.6$ and applying the BrainVar filter, we were left 
with about 140,000. 

GWAS $p$-values are often not independent due to linkage disequilibrium (LD), which is 
the tendency for some SNPs to be inherited together. For a given population, this is often measured by $r^2$: the correlation between number of alleles at two SNP locations 
\citep{slatkinLinkageDisequilibriumUnderstanding2008}. Such dependence formally 
invalidates the guarantees of the BH procedure.

\citet{purcellPLINKToolSet2007} propose using the PLINK software to find SNPs which 
are in low LD with each other and
``prune to a reduced subset of approximately independent SNPs'', a popular technique. 
But instead of performing this LD pruning, we instead used LD to generate a dependency graph between the SNP $p$-values---we drew an edge between two SNPs if and only if they are on the same chromosome 
\textit{and} their LD as measured by $r^2$ is greater than $0.2$, which is the PLINK default. As a reference to compute $r^2$, we used the European Ancestry data from 1000 Genomes, Phase 3 \citep{fairleyInternationalGenomeSample2020}. Results are shown in Figure \ref{fig:real_data}.

\begin{figure}[htbp]
    \centering
    \hspace*{-1.5cm} 
    \includegraphics[width=1.2\textwidth]{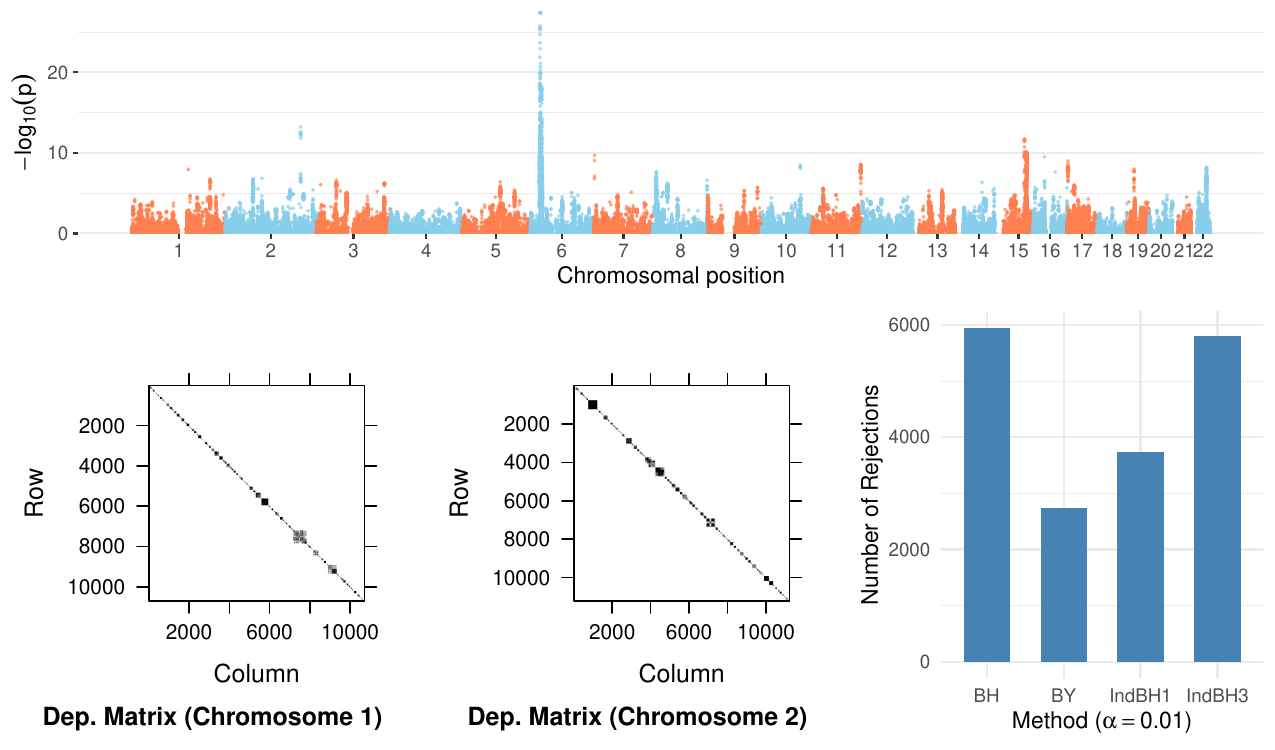}
    \caption{\emph{Top}: A Manhattan plot of the SNP $p$-values from the schizophrenia GWAS of \citet{ripkeBiologicalInsights1082014}, after filtering to about 140,000 $p$-values. \emph{Bottom left, middle}: Dependency matrices for SNP $p$-values within a single chromosome. (SNPs on different chromosomes are assumed independent). \emph{Bottom right}: Results for each method.}\label{fig:real_data}
\end{figure}

As will often be the case, the dependency matrix in this application is sparse, so that $\IndBH^{(3)}$ retains most of the power of BH. In fact, there seems to be a block dependence structure, not only between chromosomes but even within a \emph{single} chromosome. This is likely a manifestation of haplotype blocks in the human genome \citep{wallHaplotypeBlocksLinkage2003a}.



\section{Discussion}\label{sec:discussion}

In this work, we have developed procedures for FDR control in multiple testing 
under a natural dependence constraint: a dependency graph on $p$-values. We now 
discuss connections to the literature and compare to the
the BY and the BH procedure. 



\subsection{Connections to existing frameworks}\label{sec:connections}

Let us now understand existing theoretical frameworks for achieving
finite-sample FDR control. We mostly focus on those of
\citet{blanchardTwoSimpleSufficient2008b} and 
\citet{fithianConditionalCalibrationFalse2020b}, though related results were also derived by 
\citet{tamhaneGeneralizedStepupdownMultiple1998}, \citet{sarkarResultsFalseDiscovery2002a}, and 
\citet{finnerDependencyFalseDiscovery2007}. 


\subsubsection{Two simple sufficient conditions for FDR control}

Consider the following notions originally proposed by \citet{blanchardTwoSimpleSufficient2008b}. 
Given a ``shape function'' $s: \mathbb{R} \to \mathbb{R}$, a procedure $\cR_\alpha$ is 
$s$-self-consistent, or $s$-SC, if $i \in \cR_\alpha(p) \Rightarrow p_i \leq \alpha s\big(|\cR_\alpha(p)|\big)/m$. 
Also, the pair of random variables $(p_i, r^*)$ satisfies $s$-dependency control, or $s$-DC, if
\[
    \mathbb{E}\left[\frac{1\{p_i \leq c s(r^*)\}}{r^*}\right] \leq c.
\]
Henceforth, we assume $s(r) := r$, and no longer mention the shape function $s$. Property 
SC then becomes exactly the condition \ref{item:SC}. 

\citet{blanchardTwoSimpleSufficient2008b} show that whenever $\cR_\alpha$ is SC 
and $(p_i, |\cR(p)|)$ is DC for all $i \in \cH_0$, then $\cR_\alpha$ controls the FDR. 
Here, our work departs from theirs, as 
our $\mathbb{D}$-adaptivity condition does not necessarily imply that $(p_i, |\cR_\alpha(p)|)$ satisfies DC. 
It does, however, imply that $(p_i, |\cR_\alpha(\bone^{\Nio} p)|)$ satisfies DC, as seen 
in the proof of Theorem \ref{thm:pprdcontrol}. 

The self-consistency gap \eqref{eq:rejlowerbound}, which makes procedures conservative,
is not relevant under independence or PRD, because BH has no self-consistency gap 
and is easy to compute, while being the most liberal self-consistent procedure. 
In our setting, the most liberal $\mathbb{D}$-adapted procedure is $\SU_\mathbb{D}$, 
but it did not seem easy to compute. We instead 
started from $\IndBH_\mathbb{D}$, an efficiently computable procedure with 
reasonable performance, and applied the gap chasing iterates to improve it. 

\subsubsection{Conditional calibration}

The gap chasing iteration \eqref{eq:gapchasingiterate} can be viewed 
as an instantiation of the conditional
calibration framework of \citet{fithianConditionalCalibrationFalse2020b}, 
which we now describe. 

We begin with general data $X \in \cX$ from which $p$-values $p_1, \dots, 
p_m$ are computed. For each $i$, write down three ingredients: a calibratable, data-dependent threshold
$\tau_i(c; X)$ which is non-decreasing in $c$ for all $X$, satisfying $\tau_i(0; X) = 0$; 
a conditioning statistic $S_i$ on which $p_i$ is superuniform conditional on $S_i$; 
and any function $\beta_i: \cX \to \mathbb{R}$ at all. Then, suppose $\hat c(S_i)$ is chosen such that
\begin{equation}\label{eq:calibeq}
    \mathbb{E} \left[ \frac{1\{p_i \leq \tau_i(\hat c(S_i); X)\}}{\beta_i(X)} \mid S_i \right] \leq \alpha/m \quad \text{for all $i$.}
\end{equation}
Consider the procedure $\cR^\textup{CC}_\alpha(X) := \big\{i : p_i \leq \tau_i(\hat c(S_i); X) \big\}$. 
\citet{fithianConditionalCalibrationFalse2020b} show that $\cR^\textup{CC}_\alpha(X)$
``almost'' controls FDR in finite samples---an auxiliary randomization step is also required. 
But $\cR^\textup{CC}_\alpha(X)$ does control FDR in finite samples if $\beta_i$ 
was chosen such that
\begin{equation}\label{eq:safecondcal}
    i \in \cR^\textup{CC}_\alpha(X) \Rightarrow \beta_i(X) \leq |\cR^\textup{CC}_\alpha(X)| \quad \text{almost surely.}
\end{equation}
This says that each $\beta_i$ corresponds to a rejection lower bound, but
not necessarily that $\cR^\textup{CC}_\alpha(X)$ is self-consistent in general.

In our setting, let $X = p$ be our dataset, and $\cR_\alpha^\textup{init}(X)$ be any initial procedure. Consider the threshold $\tau_i(c; X) = c$ and conditioning statistic $S_i = p_{N_i^\mathsf{c}}$. 
It can be shown that calibrating based on
\[
\beta_i(X) = \inf_{\{p_i: p_i \leq \alpha |\{i\} \cup \cR^\textup{init}_\alpha(\bone^{\Nio} p)|/m\}} \big|\{i\} \cup \cR^\textup{init}_\alpha(\bone^{\Nio} p)\big|
\]
is equivalent to applying the gap-chasing iterations \eqref{eq:gapchasingiterate} on $\cR^\textup{init}_\alpha$.
The resulting $\cR^\textup{CC}_\alpha(X)$ satisfies \eqref{eq:safecondcal} 
when $\cR_\alpha^\textup{init}$ is $\mathbb{D}$-adapted,
again showing the FDR control of $\IndBH^{(k)}$ for $k > 1$. 

A different choice, namely 
$\beta_i(X) =  |\{i\} \cup \cR_\alpha^\textup{init}(\bone^{\Nio} p)|$, leads
to a larger $\hat c(S_i)$, and hence a more powerful procedure than ours. 
But this comes at much greater computational cost. The condition
in \eqref{eq:calibeq} becomes
\begin{equation}\label{eq:bettercondcal}
    \int_0^1 \frac{1\{p_i \leq \hat c (S_i) \}}
    {|\{i\} \cup \cR_\alpha^\textup{init}(\bone^{\Nio} p^{i \gets t})|} \, dt
    \leq \alpha/m \quad \text{for all $i$,}
\end{equation}
where $p^{i \gets t})$ is equal to $p$ except that $p_i$ is set to $t$. We apparently must compute $\cR_\alpha^\textup{init}(\bone^{\Nio} p^{i \gets t})$ for all $t \in [0, 1]$, 
whereas gap chasing only 
requires a single call to $\cR_\alpha^\textup{init}(\bone^{\Nio} p)$ for each $i$.

We also remark that this choice does \emph{not} lead to a self-consistent 
procedure as defined in \ref{item:SC}, demonstrating the suboptimality of self-consistency 
as an algorithm design principle. 
(\citet{solariMinimallyAdaptiveBH2017} also discovered a procedure
which is not self-consistent, 
controls FDR under PRDS, and dominates the BH procedure.)

\subsection{Comparisons and Conclusion}

To conclude, we will make a few practical observations
on the two most well-known 
FDR controlling methods in the literature, the BY and the BH procedure, 
by comparing them to our methods. 

\subsubsection{Versus BY}

As discussed in the Introduction, the correction of the $\BY(\alpha)$ procedure
might seem unnecessarily severe, but in practice it is often better than running $\Bonf(\alpha)$.
In our simulations, however, $\BY(\alpha)$ is often dominated by $\IndBH(\alpha)$, including
in the case of the complete graph, when $\IndBH(\alpha)$ is equivalent to $\Bonf(\alpha)$. 
Let us now explain this briefly. 

By monotonicity, both $\BY(\alpha)$ and $\IndBH(\alpha)$ reject hypothesis $H_i$ if 
$p_i \leq \alpha \beta_{i, \alpha}(p_{-i})/m$, for some function 
$\beta_{i,\alpha} : [0,1]^{m - 1} \to [0,m]$. But where $\BY$ lowers 
FDR to its nominal level by decreasing $\alpha$, $\IndBH$ directly 
limits the sensitivity of $\beta_{i, \alpha}(p_{-i})$ to dependent $p$-values in its argument, 
as shown explicitly in the proof of Theorem \ref{thm:control}.

If the dependent neighborhood $N_i$ is small and the non-null signal pattern is 
reasonably sparse---both true in typical problems---then it is unlikely 
many small $p$-values are in $N_i$, so they 
can still assist in the rejection of $H_i$, and IndBH makes more rejections. But when the $N_i$ become very large, or there are many non-nulls, or both, then true signal $p$-values could get masked and $\BY$ can
make more rejections. 

In the case of the 
complete graph, $\IndBH$ reduces to $\Bonf(\alpha)$, which wins if $\BY(\alpha)$ makes less than $\approx \log(m)$ rejections. This only occurs when the signals are very sparse, but when $m$ is small, even 
a non-null proportion of $1 - \pi_0 = 0.1$ can be sparse enough, which is what occurs in our simulations.

Ultimately, in typical multiple testing problems, 
both $\mathbb{D}$ and the signal pattern are at least moderately sparse, 
which seems to favor the $\IndBH$ procedure over $\BY$. 
But in these problems $\IndBH$  
performs about as well as BH, so why not just run BH? 

\subsubsection{Versus BH}

Note that IndBH, and all other $\mathbb{D}$-adapted procedures, are always less powerful than BH. But a belief among many researchers---both theoretical and applied---is that the uncorrected BH procedure can generally be used in practice without inflating the FDR much, see e.g. \citet{goemanMultipleHypothesisTesting2014}. If this is true, then these $\mathbb{D}$-adapted procedures may be less appealing methodologically. 

However, our methods perform similarly to BH in practice with a provable guarantee, while 
the theoretical literature has not
fully confirmed the robustness 
of BH in finite samples. Some progress is being made---see prior work, in particular \citet{chi2022multiple}---but only sparsely. A particularly important unproved conjecture in multiple testing is the \emph{exact} FDR control for two-sided testing of multivariate Gaussian means \citep{reiner-benaimFDRControlBH2007b, 
rouxInferenceGraphesPar2018, sarkarControllingFalseDiscovery2023a}, which is a basic 
setting often assumed in practice. 

For now, one practical role for our methods could be as a robustness check. Given a dependency graph $\mathbb{D}$, if the output of 
$\IndBH^{(3)}$ and BH are very different, then the user could pause and ask 
whether there might be an unfavorable dependence structure that 
$\IndBH^{(3)}$ is protecting against (that is, besides the multivariate Gaussian).

Finally, it bears reminding that even when BH controls FDR, 
the FDP may have high variance under dependence, making the FDR control misleading: 
see for example \citet{klugerCentralLimitTheorem2024a}. 
We look forward to future work that will further illuminate the safety of the BH procedure under dependence.

\section*{Software and Reproducibility}

Our \texttt{R} package \texttt{depgraphFDR} is 
available at 
\begin{center}
  \url{https://github.com/drewtnguyen/depgraphFDR}
\end{center}
Code and instructions to reproduce the plots in this paper is available at 
\begin{center}
  \url{https://github.com/drewtnguyen/depgraphFDRpaper}
\end{center}

\section*{Acknowledgments}

D.T.N. thanks Etienne Roquain and Rina Barber for helpful conversations.

\bibliographystyle{plainnat}
\bibliography{graphbh.bib}

\appendix


\section{Proofs}\label{app:proofs}

For the proofs here and elsewhere in the appendix, we define
\[
    p^{i \gets t} := (p_1, \dots, p_{i-1}, t, p_{i+1}, \dots, p_m).
\]

\pprd*

For the proof, we will use the following ``superuniformity lemma''. Originally shown by \citet{blanchardTwoSimpleSufficient2008b}, the following statement is modified from 
Lemma 1(b) of \citet{ramdasUnifiedTreatmentMultiple2019a}. 

\begin{lemma}\label{lem:supuni}
    For any coordinate-wise non-increasing function $f: [0,1]^m \to [0, \infty)$, if $p$ is PRD on a set $A$, then whenever $i \in A$ and $p_i$ is superuniform,
    \[
        \mathbb{E}\left[\frac{1\{p_i \leq f(p) \}}{f(p)} \right] \leq 1
    \]
    where we use the convention that $0/0 = 1$.
\end{lemma}

\begin{proof}[Proof of Theorem \ref{thm:pprdcontrol}]
    Following the proof of Theorem \ref{thm:control}, it is again sufficient to show that 
    \begin{equation*}
            \EE\left[
        \frac{1\{i \in \cR(p)\}}{|\cR(p)|}
            \right]
            \leq \alpha/m
    \end{equation*}
    for $i \in \cH_0$. Let $q = \oneNio p$, the $p$-value vector with node $i$'s neighborhood masked (excepting $i$ itself). Since $p$ is PPRD on $\cH_0$ w.r.t to $\mathbb{D}$, we have that  $q$ is PRD on $\cH_0$, and so
\begin{align*}
     \EE\left[
    \frac{1\{i \in \cR(p)\}}{|\cR(p)|}
        \right]
    &\leq \EE\left[
    \frac{1 \{i \in \cR(q)\}}{|\cR(q)| } 
            \right] &&\text{\ref{item:Mon} and \ref{item:NB}}\\ 
    &= \EE\left[
    \frac{1 \{q_i \leq \alpha |\cR(q)| \big/ m\}}{|\cR(q)| } 
        \right] &&\text{\ref{item:SC}} \\
    &\leq \alpha/m &&\text{by \ref{item:Mon} and Lemma \ref{lem:supuni}.} 
\end{align*}
In this, we only used that $p \mapsto |\cR_\alpha(p)|$ is 
coordinate-wise nonincreasing, 
which is weaker than
Assumption \ref{item:Mon}. 
\end{proof}

\pprdtstat*

\begin{proof}[Proof of Proposition \ref{prop:pprdtstat}]
    \citet{fithianConditionalCalibrationFalse2020b} demonstrate the formula
    \begin{equation}\label{eq:tstatrelation}
        T_j^2 = U_{i,j}^2 \left(
            \frac{\nu + T_i^2}{\Psi_{j,j} V_i} 
        \right) + 
        \frac{\Psi_{i,j} T_i U_{i,j}}{\Psi_{i,i}}  \sqrt{
            \frac{\nu + T_j^2}{\Psi_{j,j} V_i} 
        }
        + \left(\frac{\Psi_{i,j}}{\Psi_{i,i}} T_i\right)^2
    \end{equation}
    for any pair $i,j$, where $U_i = Z_{-i} - \Psi_{-i,i} \Psi^{-1}_{i,i} Z_i$ and 
    $V_i = \nu \hat \sigma^2 + Z_i^2/\Psi_{i,i}$.

    Now for any $i \in [m]$ and increasing 
    set $K_i \subset \mathbb{R}^{|N_i^\mathsf{c}|}$, define
    \begin{align*}
        f(s) := \mathbb{P}(T \in K_i \mid T_i^2 = s) 
        &= \mathbb{E} \left[\mathbb{E}\Big[ 1\{T_{N_i^\mathsf{c}} \in K_i \mid U_i, V_i, T_i^2 = s\}\Big]\right],
    \end{align*}
    Because $T_j^2$ is increasing in $T_i^2$ whenever $\Psi_{i,j} = 0$ in 
    \eqref{eq:tstatrelation}, it follows that the indicator is $1$ whenever $s$ is greater
    than some threshold depending on $U_i, V_i$, and $K_i$, so that $f(s)$ is
    increasing in $s$ and $T^2$ is PPRD with respect to $\mathbb{D}$.

\end{proof}

\indbhrep*

\begin{proof}[Proof of Proposition \ref{prop:indbhrep}]
    First, if 
    $i \in \cR^\IndBH(p)$, there exists an independent set $J$ containing $i$ 
    such that $p_j \leq \alpha |J|/m$ for all $j \in J$. Then
    $|J| \leq  \max \{r :  |I_i(r)| \geq r \}$, since 
    $J' = I_i(|J|)$ has size at least $|J|$, because
    \[
        J \in \Big\{\{i\} \cup B : B \in \Ind(\mathbb{D}[Q_{-i}(|J|)])\Big\},
    \]
    and $J'$ is the largest member of this set, by definition. Then $|J| \leq \beta_{\alpha,i}^\IndBH(p_{-i})$, so that $p_i \leq \alpha \beta_{\alpha,i}^\IndBH(p_{-i})/m$. 

    On the other hand, let $\cB_i =  I_i(\beta_{\alpha,i}^\IndBH(p_{-i}))$, which satisfies 
    $|\cB_i| = \beta_{\alpha,i}^\IndBH(p_{-i})$. We can hence write 
    $\{p_i \leq \alpha \betaai^\IndBH(p_{-i})/m\} \Leftrightarrow \{p_i \leq \alpha |\cB_i|/m\}$. 
    The set $\cB_i$ is an independent set which contains $i$
    and where $p_j \leq \alpha |\cB_i|/m$ for all $j \in \cB_i$. But this is 
    a certificate set as defined in Section \ref{sec:indbh}, so we must have $i \in \cR^\IndBH_\alpha(p)$. 
\end{proof}

\ignorebh*

To justify this proposition, we use a lemma from which it quickly follows. To 
avoid excessive notation, we will simply define 
$\cR^{\IndBH^{(k)}}_{\alpha, -i}(p_{-i}) \subset [m] \setminus \{i\} $ to be the result of running $\IndBH$ on the 
the subvector $p_{-i}$ and subgraph $\mathbb{D}[[m] \setminus \{i\}]$, but returning an index set corresponding
to that of the original vector $p$. We define the quantity $\cR^\BH_{\alpha, -i}(p_{-i})$ 
similarly, but without reference to the graph.
\begin{lemma}\label{lem:settoone}
    For any $k \geq 1$, on the event that $i \notin \cR_\alpha^{\IndBH^{(k)}}(p)$, we have the identity
    \[
        \cR_\alpha^{\IndBH^{(k)}}(p) = \cR_\alpha^{\IndBH^{(k)}}(\bone^{\{i\}} p) = 
        \cR_{\alpha',-i}^{\IndBH^{(k)}}(p_{-i}),
    \]
    where $\alpha' = \alpha (m - 1)/m.$
\end{lemma}


\begin{proof}[Proof of Lemma \ref{lem:settoone}.]
    We will prove this by induction. For $k = 1$, recall the definition of IndBH from Equation \eqref{eq:indbhdef}:
    \[
        \cR_\alpha^{\IndBH^{(1)}}(p) = \bigcup_{I \in \Ind(\mathbb{D})} \cR_\alpha^{\BH}(\bone^{I^\mathsf{c}} p)
    \] which is just a 
    union of BH rejection sets. It is a well known fact from the theory of multiple testing that whenever $i \notin \cR_\alpha^{\BH}(p)$, we have $\cR_\alpha^{\BH}(p) = \cR_\alpha^{\BH}(\bone^{\{i\}}p) = \cR_{\alpha', -i}^{\BH}(p_{-i})$. Now suppose $i \notin \cR_\alpha^{\IndBH^{(1)}}(p)$. Then it is in none of BH rejection sets that make up the union of Equation \eqref{eq:indbhdef}. So using the well-known fact, we have
    \[
        \bigcup_{I \in \Ind(\mathbb{D})} \cR_\alpha^{\BH}(\bone^{I^\mathsf{c}} p)
        = \bigcup_{I \in \Ind(\mathbb{D})} \cR_\alpha^{\BH}(\bone^{\{i\}} \bone^{I^\mathsf{c}} p)
        = \bigcup_{I \in \Ind(\mathbb{D})} \cR_{\alpha',-i}^{\BH}((\bone^{I^\mathsf{c}} p)_{-i})
    \]
    which shows the claim for $k = 1$, using the definition of IndBH. Now suppose it holds for some $k \geq 1$. Suppose $j \notin \cR_\alpha^{\IndBH^{(k+1)}}(p)$. Then by Proposition \ref{prop:IndBHk} and monotonicity, we also have $j \notin \cR_\alpha^{\IndBH^{(k)}}(\bone^{\Nio} p)$ for any $i$. We then have
    \begin{align*}
        \cR_\alpha^{\IndBH^{(k+1)}}(p) 
            &= \left\{i : p_i \leq \frac{\alpha |\{i\} \cup  \cR^{\IndBH^{(k)}}_\alpha(\bone^{\Nio} p)|}{m} \right\} \\
            &= \left\{i \neq j : p_i \leq \frac{\alpha |\{i\} \cup  \cR^{\IndBH^{(k)}}_\alpha(\bone^{\Nio} p)|}{m} \right\} \\
            &= \left\{i \neq j : p_i  \leq \frac{\alpha |\{i\} \cup  \cR^{\IndBH^{(k)}}_\alpha(\bone^{\{j\}} \bone^{\Nio} p)|}{m} \right\} \\
            &= \left\{i \neq j : p_i  \leq \frac{\alpha |\{i\} \cup  \cR^{\IndBH^{(k)}}_{\alpha', -j}((\bone^{\Nio} p)_{-j})|}{m} \right\} \\
            &= \left\{i \neq j : p_i  \leq \frac{\alpha' |\{i\} \cup  \cR^{\IndBH^{(k)}}_{\alpha', -j}((\bone^{\Nio} p)_{-j})|}{m - 1} \right\}.
    \end{align*}
    Here we used the induction hypothesis in the third and fourth equalities. The third expression on the right-hand side
    is precisely $\cR_\alpha^{\IndBH^{(k+1)}}(\bone^{\{i\}} p)$, and the last is $\cR_{\alpha',-i}^{\IndBH^{(k+1)}}(p_{-i})$, so we are done.
\end{proof}

\begin{proof}[Proof of Proposition \ref{prop:ignorebh}]
    Use Lemma $\ref{lem:settoone}$ repeatedly for every $i \notin Q(\bar r)$ to show the 
    claim. The reason why it holds for $\bar r = |\cR^\BH_\alpha(p)|$ is because 
    if $p_i > \alpha \bar r/m$, then $H_i$ cannot be rejected by BH. Because BH
    contains every $\mathbb{D}$-adapted procedure, then $H_i$ cannot be rejected by $\IndBH^{(k)}$ 
    for any $k$. 
\end{proof}

\indbhclique*

\begin{proof}[Proof of Proposition \ref{prop:indbhclique}]
    Let $\cR = \cR^{\IndBH_\mathbb{D}}_\alpha(p)$ and 
    \[
    \cR' = \{i : p_i \leq \alpha |\cR^\BH_\alpha(\bone^{M^\mathsf{c}}p)|/m\} = \{i : p_i \leq \alpha |C^*| /m\}
    \]
    where $C^* = |\cR^\BH_\alpha(\bone^{M^\mathsf{c}}p)|$. We will need to show that $\cR = \cR'$. 

    For the forward inclusion, recall from the discussion in Section \ref{sec:indbh} that $i \in \cR$ if there 
    exists a certificate set $C$ such that $i \in C$, and hence $p_i \leq \alpha |C|/m$ 
    by definition of certificate sets. In this setting, the certificate set with the largest cardinality
    is precisely $C^*$, so it follows that $p_i \leq \alpha |C^*|/m$ and $\cR \subset \cR'$. 

    For the backward inclusion, suppose that $p_i \leq \alpha |C^*| /m$. 
    Let $\kappa[i]$ be the index for the block that contains $i$. Then
    $C_i = (C^* \setminus \{j^*_{\kappa[i]} \}) \cup \{i\}$ has the same cardinality 
    as $C^*$, but remains a certificate set. Because it contains $i$, we know $i \in \cR$, 
    and because $i$ was arbitrary, $\cR' \subset \cR$.
\end{proof}


\section{Computational details}\label{app:computation}

Here we present our computational strategy for $\IndBH^{(k)}$ in full detail. The main algorithm is presented as Algorithm \ref{alg:IndbhK}, but we describe it here at a high level. 

After subsetting to the BH rejection set, we precompute a certain table $\mathbf{V}$, whose entries record the independence numbers of various different subgraphs, to be used as a reference when $\IndBH$ is eventually called. Then each time $\IndBH^{(\ell)}$ is called to determine the rejection status of each hypothesis $H_i$, we first run $\IndBH^{(\ell-1)}$ and then perform inexpensive inclusion/exclusion checks based on its rejection set. When these are inconclusive for a hypothesis $i$, we call $\IndBH^{(k-1)}$, but on the masked vector $\bone^{\Nio} p$. 

This continues until finally $\IndBH$ is called at the bottom of the recursion. In this case, we cheaply perturb the table $\mathbf{V}$ into the table $\widetilde{\mathbf{V}}$ to account for the masking, and again perform inclusion/exclusion checks for the rejection status of each $H_i$. These checks depend only on $\widetilde{\mathbf{V}}$. When they are inconclusive, we must exactly compute the IndBH rejection set using an expensive fallback. 

\begin{algorithm}[htbp]
    \DontPrintSemicolon
    \caption{Pseudocode for $\cR^{\IndBH^{(k)}_\mathbb{D}}_\alpha(p)$.}\label{alg:IndbhK}
    
    \SetKwProg{Def}{def}{:}{}
    
    \tcc{SETUP}
    \tcc{Initial filtering: Section \ref{app:filtering}}
    $p \gets p_{[\cR^\BH_\alpha(p)]}$\;
    $\mathbb{D} \gets \mathbb{D}[\cR^\BH_\alpha(p)]$\;
    $\alpha \gets \alpha |\cR^\BH_\alpha(p)|/m$\;
    \tcc{Precomputation: Section \ref{app:precomputation}}
    $\mathbf{V} \gets \texttt{precompute}(p, \mathbb{D})$
    
    \tcc{DEFINE MEMOIZED PROCEDURES}
    \Def{$\cR^{\textup{base}}(q)$}{ \tcp{Base case: an implementation of IndBH}
        $\widetilde{\mathbf{V}} \gets \texttt{update}(\mathbf{V})$ \tcp{Updating V: Section \ref{app:updateV}}
        \tcc{Inclusion exclusion checks: Section \ref{app:incluexcluIndBH}}
        $\cC^{(1)} \gets \texttt{check}^{(1)-}(q,\widetilde{\mathbf{V}})$\;
        $\cD^{(1)} \gets \texttt{check}^{(1)+}(q, \widetilde{\mathbf{V}})$\;
        $\cR \gets \varnothing$\;
        \For{$i \gets 1$ \KwTo $m$}{
            \If{$i \in \cC^{(1)}$}{
                $\cR = \cR \cup \{i\}$
            }
            \ElseIf{$i \notin \cD^{(1)}$}{
                \tcp{do nothing}
            }
            \Else{
                $\beta_i^{(1)} \gets \texttt{beta}(q, \widetilde{\mathbf{V}})$ \tcp{Expensive fallback: Section \ref{app:IndBHbeta}}
                \If{$p_i \leq \alpha \beta^{(1)}_i/m$}{
                    $\cR = \cR \cup \{i\}$
                }    
            }
        }
        \Return $\cR$\;
    }
    
    \Def{$\cR^{{(\ell)}}(q)$}{
        \If{$\ell =1$}{
            \Return $\cR^{\textup{base}}(q)$ \tcp{Run base case}
        }
        \tcc{Inclusion exclusion checks: Section \ref{app:incluexcluIndBHk}}
        $\cC^{(\ell)} \gets \texttt{check}^{(\ell)-}(q)$\;
        $\cD^{(\ell)} \gets \texttt{check}^{(\ell)+}(q)$\;
        $\cR \gets \varnothing$\;
        \For{$i \gets 1$ \KwTo $m$}{
            \If{$i \in \cC^{(\ell)}$}{
                $\cR = \cR \cup \{i\}$
            }
            \ElseIf{$i \notin \cD^{(\ell)}$}{
                \tcp{do nothing}
            }
            \Else{
                $\beta_i^{(\ell)} \gets \cR^{^{(\ell-1)}}(\bone^{N_i} q)$ \tcp{We mask $N_i$, not $\Nio$: Section \ref{app:maskingdetail}}
                \If{$p_i \leq \alpha \beta^{(\ell)}_i/m$}{
                    $\cR = \cR \cup \{i\}$
                }
            }    
        }
        \Return $\cR$\;
    }    
    
    \tcc{RETURN REJECTION SET}
    \Return $\cR^{{(k)}}(p)$
    
    \end{algorithm}

Detailed explanations of these steps now ensue. We adopt the notation from Section \ref{sec:computation}. 

\subsection{Initial filtering}\label{app:filtering}

Section \ref{sec:graphreduce} and Proposition \ref{prop:ignorebh} in the main text already explain 
why this filtering step does not hamper the correctness
of the algorithm. Though it is the most important step, the sections that follow 
do not technically depend on it for their correctness. 

This step---discarding all $p$-values which exceed $\alpha \bar r / m$---
could potentially be improved by applying Proposition \ref{prop:ignorebh}
with a smaller choice of $\bar r$. For example, any upper bound to the size of $\SU_\mathbb{D}(\alpha)$
would work---such bounds were demonstrated in the proof of Theorem \ref{thm:gapchase}. 
We leave any possible further improvements
to be addressed in the package documentation of \texttt{depgraphFDR}.

\subsection{Precomputing the table $\mathbf{V}$}\label{app:precomputation}

Recall from Proposition \ref{prop:indbhrep} that
\[
    \beta_{\alpha,i}^\IndBH(p_{-i}) := \max \{r :  |I_i(r)| \geq r \},
\]
where we have from Section \ref{sec:caching} that
\begin{align*}
    |I_i(r)| 
        &= 1 +  \mathtt{IndNum}(\mathbb{D}_{\kappa[i]}[Q_{-i}(r)]) + \sum_{k \neq \kappa[i]} \mathtt{IndNum}(\mathbb{D}_k[Q(r)]) \\
        &= 1 +  \mathtt{IndNum}(\mathbb{D}_{\kappa[i]}[Q_{-i}(r)]) + \sum_{k \neq \kappa[i]} \mathbf{V}_{k, r}, \label{eq:appdixindbhexact}
\end{align*}
where we define the table $\mathbf{V} \in \mathbb{R}^{\bar k \times |\cR^\BH_\alpha(p)|}$ by $\mathbf{V}_{k, r} := \mathtt{IndNum}\big(\mathbb{D}_k[Q(r)]\big)$. 
This table is useful as its elements get reused for every $i$. To compute $\mathbf{V}$, our strategy is as follows:
\begin{enumerate}[label=Step~\arabic*., ref=Step~\arabic*]
    \item \label{Vstep1} Compute and save the connected components $\mathbb{D}_1, \dots, \mathbb{D}_{\bar k}$. Initialize $\mathbf{V}$ as the zero matrix in $\mathbb{R}^{\bar k \times |\cR^\BH_\alpha(p)|}$. 
    \item \label{Vstep2} For components $k = 1, \dots, \bar k$, 
    \begin{itemize}
        \item Compute and save the maximal independent sets (MIS) $I_{k,1}, \dots, I_{k, \bar \ell}$ of $\mathbb{D}_k$. 
        \item For sets $\ell = 1, \dots, \bar \ell$ and $r = 1, \dots, |\mathcal{R}_\alpha^\BH(p)|$,
        \begin{itemize}
            \item Compute $u_{k, r, \ell} := \# \{i \in I_{k, \ell} : p_i \leq \alpha r/m \}$. (For this we use \texttt{R}'s \texttt{table} function on $\lceil m p_{I_{k, \ell}}/\alpha \rceil$.)
        \end{itemize}
        \item Take $\mathbf{V}_{k, r} = \max_\ell u_{k, r, \ell}$. 
    \end{itemize}
\end{enumerate}

To perform the graph-theoretic calculations, we used the \texttt{igraph} package
in \texttt{R} \citep{csardiIgraphInterfaceIgraph2025}. In particular, the function \texttt{igraph::maximal\_ivs} is an off-the-shelf implementation of the algorithm of \citet{tsukiyamaNewAlgorithmGenerating1977} for finding maximal independent sets. More modern algorithms for the maximal independent set problem---in its equivalent form of the maximal clique problem---are reviewed in \citet{wuReviewAlgorithmsMaximum2015} and \citet{marinoShortReviewNovel2024}.

We can also optimize further by reducing the number of columns of 
$\mathbf{V}$. Let 
\begin{equation}\label{eq:indbhbetaUB}
    \bar r^{(1)} := \max \left\{r \leq |\cR_\alpha^\BH(p)| : \sum_{k = 1}^{\bar{k}}\mathbf{V}_{k,r} \geq r \right\}.
\end{equation}
Because $\sum_{k = 1}^{\bar k} \mathbf{V}_{k,r} \geq |I_i(r) |$, comparing to
Proposition \ref{prop:indbhrep} we have $\beta^\IndBH_{\alpha,i}(p_{-i}) \leq \bar r^{(1)}$. (We can also see this by recognizing $\bar r^{(1)}$ as the size of the largest certificate set). Computing IndBH
only requires columns $\mathbf{V}_{\cdot, r}$ such that $r \leq \bar r^{(1)}$, 
where $\bar r^{(1)}$ is \emph{any} upper bound to $\beta^\IndBH_{\alpha,i}(p_{-i})$, 
so the remaining columns can be safely dropped. 

\subsection{Updating the table $\mathbf{V}$}\label{app:updateV}

Working through the recursion of Section \ref{sec:indbhk}, it can be seen that the eventual calls to $\IndBH$ from $\IndBH^{(k)}$ take the form $\cR_\alpha^\IndBH(\bone^{N_{i_{k-1}}^\circ} \dots \bone^{N_{i_2}^\circ} \bone^{N_{i_1}^\circ} p)$ for some sequence $i_1, \dots, i_{k-1}$. Let us generically denote this modified vector as 
\[
\tilde p := \bone^{N_{i_{k-1}}^\circ} \dots \bone^{N_{i_2}^\circ} \bone^{N_{i_1}^\circ} p,
\]
leaving the sequence of masked indicies $N_{i_{k-1}}^\circ, \dots , {N_{i_2}^\circ}, {N_{i_1}^\circ}$ implicit. 

(Note that Algorithm \ref{alg:IndbhK} actually takes $\tilde p := \bone^{N_{i_{k-1}}} \dots \bone^{N_{i_2}} \bone^{N_{i_1}} p$. This is somewhat computationally convenient, and still correct. The explanation is deferred to Section \ref{app:maskingdetail}.)

Before we can use the table $\mathbf{V}$ to help us run $\IndBH$ on $\tilde p$, we need to adjust 
it to account for the masking, because $\mathbf{V}$ was computed using the unmodified $p$. 
We hence repeat \ref{Vstep2} from Appendix \ref{app:precomputation}, but only for those connected components overlapping with 
$ N_{i_1}^\circ \cup \dots \cup N_{i_{k-1}}^\circ$. Denote the resulting table as $\widetilde{\mathbf{V}}$. 
We call this subroutine \texttt{update} in Algorithm \ref{alg:IndbhK}.

\subsection{Inclusion and exclusion: IndBH}\label{app:incluexcluIndBH}

Write $\beta^*_i = \beta^\IndBH_{\alpha,i}(\tilde p_{-i})$. 
We have
\[
i \in \cR^\BH_\alpha(\tilde p) \Leftrightarrow \tilde p_i \leq \alpha \beta^*_i/m \Leftrightarrow r_i \leq \beta^*_i,
\]
where $r_i = \lceil m\tilde p_i/\alpha \rceil$. Let
\begin{gather*}
    \beta^+ = \max\left\{r \leq \bar r^{(1)}: \sum_{k = 1}^{\bar{k}} \widetilde{\mathbf{V}}_{k,r} \geq r\right\}, \\
    \beta^- = \max\left\{r \leq \bar r^{(1)}: \sum_{k = 1}^{\bar k}\widetilde{\mathbf{V}}_{k,r} - \max_{k'}\widetilde{\mathbf{V}}_{k',r} + 1 \geq r\right\}, \\
    \beta^-_i = \max\left\{r \leq \bar r^{(1)}: \sum_{k = 1}^{\bar k}\widetilde{\mathbf{V}}_{k,r} - \widetilde{\mathbf{V}}_{\kappa[i],r} + 1 \geq r\right\}.
\end{gather*}
These quantities can be computed entirely from $\widetilde{\mathbf{V}}$ and satisfy
\[
    \beta^- \leq \beta^-_i \leq \beta^*_i \leq \beta^+
\]
for all $i$, justifying the following shortcuts. 
\begin{enumerate}
    \item $\texttt{check}^{(1)-}$. Whenever $r_i \leq \beta^-$ or $r_i \leq \beta^-_i$, declare $i \in \cR^{\IndBH}_\alpha(\tilde p)$. We represent $\cC^{(1)} = \{i: r_i \leq \beta^-_i \}$ in Algorithm \ref{alg:IndbhK}.
    \item $\texttt{check}^{(1)+}$. Whenever $r_i > \beta^+$, declare $i \notin \cR^{\IndBH}_\alpha(\tilde p)$. We represent $\cD^{(1)} = \{i: r_i \leq \beta^+ \}$ in Algorithm \ref{alg:IndbhK}.
\end{enumerate}
These checks typically take care of most of the hypotheses. 

The checks based on $\beta^+$ and $\beta^-$ can be performed for all hypotheses at once, and 
the one based on $\beta^-_i$ can be performed for all hypotheses in a connected component at once. In our implementation, we perform the $\beta^+$ and $\beta^-$ checks first. 

\subsection{Fallback: computing IndBH exactly}\label{app:IndBHbeta}

When all the cheap checks prove inconclusive, we must run IndBH exactly on the modified vector $\tilde p$. 
We do this by computing
\[
    \beta_{\alpha,i}^\IndBH(\tilde p_{-i}) := \max \{r :  |\widetilde{I}_i(r)| \geq r \},
\]
which we call \texttt{beta} in Algorithm \ref{alg:IndbhK}. We adapt 
the definitions from Section \ref{app:precomputation} to add a tilde symbol
to mean ``computed with $\tilde p$ instead of $p$''. That is, we let
\begin{align*}
    |\widetilde{I}_i(r)| 
        &= 1 +  \mathtt{IndNum}(\mathbb{D}_{\kappa[i]}[\widetilde{Q}_{-i}(r)]) + \sum_{k \neq \kappa[i]} \mathtt{IndNum}(\mathbb{D}_k[\widetilde{Q}(r)]) \\
        &= 1 +  \mathtt{IndNum}(\mathbb{D}_{\kappa[i]}[\widetilde{Q}_{-i}(r)]) + \sum_{k \neq \kappa[i]} \widetilde{\mathbf{V}}_{k, r},
\end{align*}
where $\widetilde{Q}(r) = \{i : \tilde p_i \leq \alpha r /m\}$. 

Now given $\widetilde{\mathbf{V}}$ from Section \ref{app:updateV}, it only 
remains to compute $\mathtt{IndNum}\Big(\mathbb{D}_{\kappa[i]}[\widetilde{Q}_{-i}(r)]\Big)$. We again use \ref{Vstep2} from Section \ref{app:precomputation}, but only for the connected component $\mathbb{D}_{\kappa[i]}$ which contains $i$, and restricting its independent sets to also contain $i$. 

\subsection{Inclusion and exclusion: $\text{IndBH}^{(\ell)}$}\label{app:incluexcluIndBHk}

For completeness, we translate the checks from Section \ref{sec:considerindbh2} into Algorithm \ref{alg:IndbhK}'s notation. Additionally, we use an auxiliary procedure $\cR^{(\ell)*}$ whose definition is deferred to Section \ref{app:blockindbhk}---it is a quickly computed subset of $\cR^{(\ell)}$ that is the whole set when the connected components are cliques.
\begin{enumerate}
    \item $\texttt{check}^{(\ell)-}$. Whenever $i \in \cR^{(\ell - 1)}(\tilde p)$ or $i \in \cR^{(\ell)*}(\tilde p)$, declare $i \in \cR^{(\ell)}(\tilde p)$. We represent $\cC^{(\ell)} = \cR^{(\ell - 1)}(\tilde p) \cup \cR^{(\ell)*}(\tilde p)$ in Algorithm \ref{alg:IndbhK}.
    \item $\texttt{check}^{(\ell)+}$. Whenever $p_i > \alpha |\{i\} \cup \cR^{(\ell - 1)}(\tilde p)|/m$, declare $i \notin \cR^{(\ell)}(\tilde p)$. We represent $\cD^{(\ell)} = \{i: r_i \leq |\{i\} \cup \cR^{(\ell - 1)}(\tilde p)|\}$ in Algorithm \ref{alg:IndbhK}.

\end{enumerate}

Both of these checks can be performed on all hypotheses at once.

Note that we always perform $\texttt{check}^{(\ell)-}$ first, 
and any $i$ that was not rejected by $\texttt{check}^{(\ell)-}$ satisfies $i \notin \cR^{(\ell - 1)}(\tilde p)$. Hence in our implementation of $\texttt{check}^{(\ell)+}$, we take $|\{i\} \cup \cR^{(\ell - 1)}(\tilde p)| = 1 + |\cR^{(\ell - 1)}(\tilde p)|$.

\subsection{A masking detail}\label{app:maskingdetail}

The main text defined $\IndBH^{(k)}$ in Equation \eqref{eq:indbhkdef}, reproduced here:
\begin{gather*}
     \cR_\alpha^{\IndBH^{(1)}}(p) := \cR_\alpha^{\IndBH}(p), \\
    \cR_\alpha^{\IndBH^{(k+1)}}(p) := \left\{i : p_i \leq \frac{\alpha |\{i\} \cup  \cR^{\IndBH^{(k)}}_\alpha(\bone^{\Nio} p)|}{m} \right\}.
\end{gather*}
The following proposition shows that we can use an alternative definition, where every element of $N_i$ is masked, not just
the punctured neighborhood $\Nio$. 
\begin{proposition}\label{prop:altindbhk}
    For $k \geq 1$, let
    \begin{gather*}
        \cA^{k+1} = \left\{i : p_i \leq \frac{\alpha |\{i\} \cup  \cR^{\IndBH^{(k)}}_\alpha(\bone^{\Nio} p)|}{m} \right\}, \\
        \cB^{k+1} = \left\{i : i \in \cR_\alpha^{\IndBH^{(k)}}(p), \text{ or } p_i \leq \frac{\alpha |\{i\} \cup  \cR^{\IndBH^{(k)}}_\alpha(\bone^{N_i} p)|}{m} \right\}. 
    \end{gather*}
    We have the equality $\cA^{k+1} = \cB^{k+1}$ for all $k \geq 1$. 
\end{proposition}
It also follows from this proposition that masking by $N_i$ instead of $\Nio$ in the call to $\cR^{(\ell-1)}$ in Algorithm \ref{alg:IndbhK} is correct, because it 
occurs \emph{after} $\texttt{check}^{(\ell)-}$, which checks whether $i \in \cR_\alpha^{\IndBH^{(k)}}(p)$, 
which is convenient for the implementation.

\begin{proof}[Proof of Proposition \ref{prop:altindbhk}.]
First, we know 
$\cA^{k+1} \supset \cR^{\IndBH^{(k)}}_\alpha(p)$ from Proposition \ref{prop:IndBHk}, and combined with the monotonicity of 
$\cR^{\IndBH^{(k)}}_\alpha$ we have $\cA^{k+1} \supset \cB^{k+1}$ for all $k \geq 1$. Next, suppose $i \in \cA^{k+1}$. If 
$i \in \cR^{\IndBH^{(k)}}_\alpha(p)$, then $i \in \cB^{k+1}$ by definition. If $i \notin \cR^{\IndBH^{(k)}}_\alpha(p)$, we can apply Lemma \ref{lem:settoone} to equate $\cR^{\IndBH^{(k)}}_\alpha(\bone^{N_i} p) = \cR^{\IndBH^{(k)}}_\alpha(\bone^{\Nio} p)$
for all $i$, so that again $i \in \cB^{k+1}$, showing that $\cA^{k+1} \subset \cB^{k+1}$. 
\end{proof}

\subsection{A fast subset of $\text{IndBH}^{(\ell)}$ based on fully connected components}\label{app:blockindbhk}

Given an input graph $\mathbb{D}$ with connected components $\mathbb{D}_1, \dots, \mathbb{D}_{\bar k}$, 
define a new graph $\mathbb{D}^*$ to be the graph 
where the nodes of each component 
form a clique. Since it is more connected than $\mathbb{D}$, 
the $\IndBH_{\mathbb{D}^*}(\alpha)$ procedure rejects a \emph{subset} of 
$\IndBH_\mathbb{D}(\alpha)$. Similarly, $\IndBH^{(\ell)}_\mathbb{D*}(\alpha)$ returns a subset of $\IndBH^{(\ell)}$. 
We use $\cR^{(\ell)*}$ to refer to the rejection set of $\IndBH^{(\ell)}_\mathbb{D*}(\alpha)$. 

Due to Proposition \ref{prop:indbhclique}, $\IndBH_{\mathbb{D}^*}$ is highly computationally 
efficient. $\IndBH^{(\ell)}_\mathbb{D*}(\alpha)$ is also more efficient 
because the same rejection threshold for $p_i$ (computed in Appendix \ref{app:maskingdetail}), 
can be used for every $p$-value
in its component. Other smaller optimizations also become possible, which we leave 
to the \texttt{depgraphFDR} documentation. 

Importantly, we base $\mathbb{D}^*$ on
the graph $\mathbb{D}$ obtained \emph{after} the Appendix \ref{app:filtering} filtering, discarding $p$-values that exceed $\alpha \bar r/m$. The rejection (sub)sets computed by
$\IndBH_{\mathbb{D}^*}(\alpha)$ and $\IndBH^{(\ell)}_{\mathbb{D}^*}(\alpha)$ 
hence depend on the choice $\bar r$---there
could be more connected components for a smaller $\bar r$. 

We simply use 
$\bar r = |\cR^\BH_\alpha(p)|$, but any $\bar r$ satisfying 
the conditions of Proposition \ref{prop:ignorebh} would work. Further details 
and possible improvements are again left to the \texttt{depgraphFDR} documentation.



\section{Randomized pruning for $\textup{SU}_\mathbb{D}$}\label{app:randomprune}

Let $u_1, \dots, u_m \sim \textup{Unif}[0,1]$ be iid external randomness. Then for any monotone procedure $\cR^\textup{init}_\alpha(p)$, define the adjusted procedure $\cR^\textup{adj}_\alpha$, and its pruned version $\cR_\alpha^\textup{rand}$:
\[
    \cR^\textup{adj}_\alpha(p) = \{i : p_i \leq \alpha |\{i\} \cup \cR^\textup{init}_\alpha(\bone^{\Nio} p)| / m  \}; \quad \tilde u_i = 
        \begin{cases}
            u_i |\{i\} \cup \cR^\textup{init}_\alpha(\bone^{\Nio} p)|/m &\text{if } i \in \cR^\textup{adj}_\alpha(p) \\
            \infty &\text{o.w.}

        \end{cases}
\]
and $\cR^\textup{rand}_\alpha(p; u) := \cR_1^\BH(\tilde u)$. 
\begin{proposition}
    $\cR^\textup{rand}_\alpha(p; u)$ controls FDR under a dependency graph $\mathbb{D}$. 
\end{proposition}

As noted by \citet{fithianConditionalCalibrationFalse2020b}, this can also be expressed as the $\BH(1)$ procedure 
applied to the $p$-values $\big(u_i |\{i\} \cup \cR^\textup{init}_\alpha(\bone^{\Nio} p)|/ |\cR^\textup{adj}_\alpha(p)|\big)_{i \in \mathcal{R}_\alpha^\textup{adj}(p)}$. 

In this form, it becomes transparent that the pruning step is a patch for a failure of the self-consistency property, after one round of gap chasing. If $\cR^\textup{adj}_\alpha$ were actually self-consistent, then no rejections are pruned, and pruning is unlikely if $\cR^\textup{adj}_\alpha$ is almost self-consistent.

\begin{proof}
    \allowdisplaybreaks
    Similarly to Theorem \ref{thm:control} it suffices to show that 
    \[
        \mathbb{E}\left[
            \frac{
                1\{i \in \cR^\textup{rand}_\alpha(p; u)\}
            }{
                |\cR^\textup{rand}_\alpha(p; u)|
            }                
         \right]  \leq \alpha/m.
    \]
    We proceed as
    \begin{align*}
        \mathbb{E}\left[
            \frac{
                1\{i \in \cR^\textup{rand}_\alpha(p; u)\}
            }{
                |\cR^\textup{rand}_\alpha(p; u)|
            }                
         \right] 
         &= \mathbb{E}\left[
            \frac{
                1\{i \in \cR_1^\BH(\tilde u)\}
            }{
                |\cR_1^\BH(\tilde u)|
            }                
         \right]  \\
         &= \mathbb{E}\left[
            \frac{
                1\{i \in \cR_\alpha(p)\} 1\{i \in \cR_1^\BH(\tilde u)\}
            }{
                |\cR_1^\BH(\tilde u)|
            }                
         \right]  \\
         &= \mathbb{E}\left[
            \frac{
                1\{i \in \cR_\alpha(p)\} 1\{\tilde u_i \leq |\cR_1^\BH(\tilde u)|/m\}
            }{
                |\cR_1^\BH(\tilde u)|
            }                
         \right]  \\
         &= \mathbb{E}\left[
            \frac{
                1\{i \in \cR_\alpha(p)\} 1\{\tilde u_i \leq  |\cR_1^\BH(\tilde u^{i \gets 0})|/m\}
            }{
                |\cR_1^\BH(\tilde u^{i \gets 0})|
            }                
         \right]  \\
         &= \mathbb{E}\left[
            \frac{
                1\{i \in \cR_\alpha(p)\} 1\{ u_i {|\{i\} \cup \cR^\textup{init}_\alpha(\bone^{\Nio} p)|}/{m} \leq  |\cR_1^\BH(\tilde u^{i \gets 0})|/m\}
            }{
                |\cR_1^\BH(\tilde u^{i \gets 0})|
            }                
         \right]  \\
         &= \mathbb{E}\left[\mathbb{E}\left[
            \frac{
                1\{i \in \cR_\alpha(p)\} 1\{ u_i {|\{i\} \cup \cR^\textup{init}_\alpha(\bone^{\Nio} p)|}/{m} \leq  |\cR_1^\BH(\tilde u^{i \gets 0})|/m\}
            }{
                |\cR_1^\BH(\tilde u^{i \gets 0})|
            }                
         \right] \mid p, u_{-i} \right]  \\
         &= \mathbb{E}\left[
            \frac{
                1\{i \in \cR_\alpha(p)\} 
            }{
                |\{i\} \cup \cR^\textup{init}_\alpha(\bone^{\Nio} p)|
            }                
         \right]  \\
         &\leq \mathbb{E}\left[
            \frac{
                1\{p_i \leq \alpha |\{i\} \cup \cR^\textup{init}_\alpha(\bone^{\Nio} p)|/m\} 
            }{
                |\{i\} \cup \cR^\textup{init}_\alpha(\bone^{\Nio} p)|
            } \right]\\
            &= \mathbb{E} \left[\mathbb{E}\left[
                \frac{
                    1\{p_i \leq \alpha |\{i\} \cup \cR^\textup{init}_\alpha(\bone^{\Nio} p)|/m\} 
                }{
                    |\{i\} \cup \cR^\textup{init}_\alpha(\bone^{\Nio} p)|
                } \mid p_{N_i^\mathsf{c}}\right] \right]\\    
        &\leq \alpha/m.
    \end{align*}

\end{proof}

\section{BH under a dependency graph}\label{app:BYD}

Section \ref{sec:fdrctrlremarks}, Equation \eqref{eq:worstcasefdr} defined
the quantity
\[
    \gamma_\mathbb{D}(\alpha) := \sup \big\{\FDR_P(\cR^\BH_\alpha) : P \in \mathcal{P}_\mathbb{D}\big\}.
\]
In the ensuing sections, we establish the bounds
\begin{equation}\label{eq:BYDbounds}
    1 - \prod_{k = 1}^{\bar k} \left( 
        1 - \frac{\alpha |B_k|}{m}  \sum_{s = 1}^{|B_k|} 1/s \right) \leq \gamma_\mathbb{D}(\alpha) \leq  \frac{\alpha}{m} \sum_{i = 1}^m \left(|N_i| - \big({|N_i| - 1}\big){m^\frac{-1}{|N_i| - 1}} \right)
\end{equation}
where $B_1, \dots, B_{\bar k}$ represents any clique cover 
of $\mathbb{D}$, and $N_1, \dots, N_m$ represents the neighborhoods of node $i$ in $\mathbb{D}$, and $\alpha$ is such that $1 - \frac{\alpha |B_k|}{m}  \sum_{s = 1}^{|B_k|} 1/s > 0$ for every $k$. 

To better interpret this result, first, suppose the graph is fully connected. Then it is a clique, and the whole graph can be used as clique cover, while $N_i = m$ for every $i$, obtaining
\begin{equation*}
    \alpha \sum_{s = 1}^m 1/s \leq \gamma_\mathbb{D}(\alpha) \leq  \alpha \left(m - \big({m - 1}\big){m^\frac{-1}{m - 1}} \right).
\end{equation*}
This upper bound does not exactly equal the lower bound, as we know it should from \citet{benjaminiControlFalseDiscovery2001}, but it does satisfy $(m - \big({m - 1}\big){m^\frac{-1}{m - 1}}) \big / \sum_{s = 1}^m 1/s \to 1$ as $m \to \infty$, so it almost recovers the \citet{benjaminiControlFalseDiscovery2001} result. 

Next, suppose $\mathbb{D}$ is a block dependent graph where all blocks have equal size $b$, the clique cover can be chosen to be precisely those blocks. Then there are $\bar k = m/b$ blocks and the left-hand side of \eqref{eq:BYDbounds} is
\[
    1 -  \left( 
        1 - \frac{\alpha b}{m}  \sum_{s = 1}^{b} 1/s \right)^{m/b} = \alpha \sum_{s = 1}^{b} 1/s + O(\alpha^2) \approx \alpha \log(b),
\]
so in this case we must at least pay approximately an inflation of factor $\log(b)$. Unfortunately, in this case, and for general graphs, the upper and lower bounds are quite different. We conjecture that $\gamma_\mathbb{D}(\alpha)$ is much closer to our lower bound 
than our upper bound. Indeed, the upper bound is especially loose: the right side of 
\eqref{eq:BYDbounds} becomes
\[
    \alpha \left(b - \big({b - 1}\big){m^\frac{-1}{b - 1}} \right),
\]
which is $O(b)$ and often becomes vacuous even for moderately large $b$, exceeding the always-valid bound of $\alpha \sum_{s = 1}^{m} 1/s$. 

Next, we show how to obtain these bounds. 
\subsection{An FDR upper bound}

\begin{proposition}\label{prop:BYDcontrol} 
Whenever $\mathbb{D}$ is a dependency graph for $p$, the procedure $\BH(\alpha)$ controls the FDR at level $ L_{\mathbb{D}} \alpha $ on $p$, where 
\[
    L_{\mathbb{D}} := (1/m) \sum_{i = 1}^m L_{\mathbb{D}, i}, \quad \text{and }\quad  L_{\mathbb{D}, i} = |N_i| - \big({|N_i| - 1}\big){m^\frac{-1}{|N_i| - 1}}.
\]
\end{proposition}

\begin{proof}[Proof of Proposition 
\ref{prop:BYDcontrol}]
    Label the elements of $\Nio = N_i \setminus \{i\}$ as $\Nio = \{j_1, \dots, j_{n_i}\}$, so that $|\Nio| = n_i$. Let $\piktozero{k}$ be defined such that
    \[
        \piktozero{k}_j = 
            \begin{cases}
                0 & \text{for $j \in \{j_1, \dots, j_{k}\} \cup \{i\}$} \\
                1 & \text{for $j \in \{j_{k+1}, \dots, j_{n^*}\}$} \\
                p_j & \text{for $j \notin N_i$.} 
            \end{cases}
    \]
    Then let $p_\textup{max}(S_i) := \sup \{p_i : \exists k \text{ with } p_i \leq \alpha \mathcal|{R}_\alpha^\BH(\piktozero{k})|/m \}$, and for $p_i \leq p_\textup{max}$, let $k^*(p_i, S_i) := \min \{k : p_i \leq \alpha \mathcal|{R}_\alpha^\BH(\piktozero{k})|/m \}$. As in the proof of Theorem \ref{thm:control}, we have 
    \begin{equation}\label{eq:appfdrdecomp}
        \FDR = \sum_{i \in \cH_0} \EE\left[
            \frac{1\{i \in \cR^\BH_\alpha(p)\}}{|\cR^\BH_\alpha(p)|}
                \right],        
    \end{equation}
    so it is sufficient to show that 
    \begin{equation}\label{eq:fdri-byd}
            \EE\left[
        \frac{1\{i \in \cR^\BH_\alpha(p)\}}{|\cR^\BH_\alpha(p)|}
            \right]
            \leq \alpha L_{\mathbb{D}, i} /m
    \end{equation}
    for an arbitrary $i \in \cH_0$. We upper bound the conditional expectation almost surely:
    \begin{align*}
        \mathbb{E}\left[ 
            \frac{1\{i \in \mathcal{R}^\BH_\alpha(p)\}}{
                |\mathcal{R}^\BH_\alpha(p)| 
            } \mid S_i
        \right]
        &=
            \mathbb{E}\left[ 
                \frac{1\{p_i \leq \alpha |\mathcal{R}^\BH_\alpha(p)|/m \}}{
                    |\mathcal{R}^\BH_\alpha(p)| 
                } \mid S_i
        \right] \\
        &=
            \mathbb{E}\left[ 
                \frac{1\{p_i \leq \alpha |\mathcal{R}^\BH_\alpha(p^{i \gets 0})|/m \}}{
                    |\mathcal{R}^\BH_\alpha(p^{i \gets 0})| 
                } \mid S_i
        \right] \\
        &\overset{(i)}{\leq}
            \mathbb{E}\left[ 
                \max_{0 \leq k \leq n_i}
                \frac{1\{p_i \leq \alpha |\mathcal{R}^\BH_\alpha(\piktozero{k})|/m \}}{
                    |\mathcal{R}^\BH_\alpha(\piktozero{k})| 
                } \mid S_i
            \right] \\
        &=
        \int_0^{p_\textup{max}} \frac{1}{|\mathcal{R}^\BH_\alpha(\piktozero{k^*(p_i, S_i)})| } \, dp_i \\
        &=:
        \int_0^1 g(p_i; S_i) \, dp_i,
    \end{align*}
where 
    \[
        g(p_i; S_i) := 
        \begin{cases}
            |\mathcal{R}^\BH_\alpha(\piktozero{k^*(p_i, S_i)})|^{-1} & \text{$p_i \in [0, p_\textup{max}(S_i)]$} \\
            0 & \text{$p_i \in [p_\textup{max}(S_i), 1]$}
        \end{cases}.
    \]
The function $g$ is necessarily a piecewise, 
decreasing function, with range $G \subset \{1, 1/2, \dots, 1/m, 0\}$ 
where $|G \setminus \{0\}| \leq n_i + 1$, and which satisfies $p_i \leq \alpha g(p_i; S_i)^{-1}/m$
whenever $g(p_i;S_i) > 0$.

For any $g$ satisfying these conditions, enumerate its range as $G = \{1/r_1, 1/r_2, \dots, 1/r_{n_i + 1}, 0\}$ for integers 
$1 \leq r_1 \leq \dots \leq r_{n_i + 1} \leq m$. Then there exists a function $\tilde g$ satisfying these conditions, with the same range, 
of the form
\[
    \tilde g(p_i) = \begin{cases}
        r_1^{-1} & p_i \in [0, \alpha r_1/m]\\ 
        r_2^{-1} & p_i \in [\alpha r_1/m, \alpha r_2/m] \\
        \ \vdots & \\
        r_{n_i + 1}^{-1} & p_i \in [\alpha r_{n_i}/m, \alpha r_{n_i + 1}/m] \\
        0 & p_i \in [\alpha r_{n_i} /m, 1]
    \end{cases}
\]
in which $\int_0^1 g(p_i; S_i) \, dp_i \leq \int_0^1 \tilde g(p_i; S_i) \, dp_i$, which evaluates to
\[
    \int_0^1 \tilde g(p_i) \, dp_i = \frac{\alpha}{m}\left(1 + 
    \left(1 - \frac{r_1}{r_2}\right) + 
    \left(1 - \frac{r_2}{r_3}\right) + 
    \dots  + 
    \left(1 - \frac{r_{n_i}}{r_{n_i + 1}}\right)
     \right).
\]
The function $g^*$ which solves $\max_g \int_0^1 g(p_i; S_i) \,dp_i$ is then specified by maximizing the above as an integer optimization problem. Relaxing to real valued $r_1, \dots, r_{n_i + 1} \in [1,m]$ instead gives the optimizer.
$r^*_k = m^\frac{k-1}{n_i}$. Plugging this in gives the upper bound 
\[
    \int_0^1 g(p_i; S_i) \, dp_i \overset{(ii)}{\leq} \int_0^1 g^*(p_i; S_i) \, dp_i \leq \frac{\alpha}{m}\left((n_i + 1) - n_i m^{-1/n_i}  \right).
\]
Recalling that $|N_i| = n_i + 1$, using the law of total expectation gives the result. 

\end{proof}

This approach is unoptimizable in the sense that there exists a distribution where the conditional expectation is equal to the marginal, and the a.s. inequalities
(i) and (ii) become equalities. But the overall bound is likely very loose, because to achieve it, a single distribution would have to saturate all of the inequalities in \eqref{eq:fdri-byd} at once. A more refined approach should probably analyze more than one term of \eqref{eq:appfdrdecomp} at a time. 


\subsection{An FDR lower bound}\label{app:bydlower}

Let $B_1, \dots, B_{\bar k} \subset [m]$ be a clique cover 
of $\mathbb{D}$: that is, a partition of the nodes of the graph such that
each $B_k$ is a clique in $\mathbb{D}$. Let $b_k := |B_k|$. Let 
$\mathbb{P}_k$ refer to the probability distribution supported on $\{0, 1, \dots, b_k\}$ such that
$\Pr(X = s) = (\alpha b_k/m) \cdot 1 /s$ for $s \geq 0$ and $\Pr(X = 0) = 1 - (\alpha b_k/m)\sum_{s = 1}^{b_k} (1/s)$. It follows that we must restrict $\alpha$ as well so that $\Pr(X = 0) \geq 0$. 

Given such a clique cover, 
consider the following hierarchical way to sample a $p$-value vector $p \in [0,1]^m$. 

\begin{itemize}
    \item Independently for each $k \in \{1, \dots, \bar k\}$, 
    \begin{itemize}
        \item Sample $s_k \sim \mathbb{P}_k$. 
        \item Sample a subset $S_k \subset B_k$ of size $s_k$ uniformly at random. 
    \end{itemize}
    \item Conditionally on $S_1, \dots, S_k$, independently for $i \in \{1, \dots, m\}$, 
    \begin{itemize}
        \item Sample 
        \[
        p_i \sim \begin{cases}
            \textup{Unif}[\alpha (s_{\kappa[i]} - 1)/m, \alpha s_{\kappa[i]}/m] & \text{if $i \in S_{\kappa[i]}$} \\
            \textup{Unif}[\alpha b_{\kappa[i]}/m,1] & \text{if $i \notin S_{\kappa[i]}$}
        \end{cases}
        \]
    \end{itemize}
\end{itemize}

\begin{proposition}
    When $p$ is sampled as above, then it has dependency graph $\mathbb{D}$, 
    satisfies $p_i \sim \textup{Unif}[0,1]$ marginally for every $i$, and under the global null, the BH procedure has
    \[
    \FDR(\cR^\BH_\alpha) \geq 1 - \prod_{k = 1}^{\bar k} \left( 
        1 - \frac{\alpha b_k}{m}  \sum_{s = 1}^{b_k} 1/s
    \right).
    \]
\end{proposition}

It follows that the worst-case FDR $\gamma_\mathbb{D}(\alpha)$ over all models with dependency graph $\mathbb{D}$ 
is lower bounded by this quantity. 

\paragraph{Block dependence.} Not every choice of clique cover gives a useful bound. For example, taking $B_1, \dots, B_m$ with 
$B_i = \{i\}$  lower bounds the FDR by $(1 - (1 - \alpha/m)^m)$, which is smaller than $\alpha$. But under block dependence, a good
choice of clique cover is the blocks themselves. This choice, with the 
above sampling distribution, was used in Section \ref{sec:fdrexpr} for exhibiting FDR inflation. 

\begin{proof}
    The random vector $p$ has dependency graph equal to the disjoint union of the clique covers. It follows that 
    it has dependency graph $\mathbb{D}$ because $\mathbb{D}$ is strictly denser than the union, in 
    the sense of having more edges. 

    Also, $p_i$ is marginally uniform, as we can compute its density $f(x)$. For $x \in [\alpha (s-1)/m, \alpha s/m]$ and any $s \in \{0, \dots, b_{\kappa[i]} \}$, we have
    \begin{align*}
        f(x) &= \lim_{\delta x\to 0} \frac{ \Pr(p_i \in [x, x + \delta x], i \in S_{\kappa[i]}, s_{\kappa[i]} = s)}{\delta x} \\
            &= \lim_{\delta x\to 0} \frac{\Pr(p_i \in [x, x + \delta x] \mid i \in S_{\kappa[i]}, s_{\kappa[i]} = s)}{\delta x} \Pr(i \in S_{\kappa[i]} \mid s_{\kappa[i]} = s) \Pr(s_{\kappa[i]} = s) \\
               &= \lim_{\delta x\to 0} (1/\delta x) \left( m \cdot \delta x/\alpha \right) \left(s/b_{\kappa[i]} \right) \left(\alpha b_{\kappa[i]}/m \cdot 1/s \right)\\
               &= \lim_{\delta x\to 0} 1 \\
               &= 1. 
    \end{align*}
    and similarly, for $x \in [\alpha b_{\kappa[i]}/m, 1]$, 
    \begin{align*}
        f(x) &= \lim_{\delta x\to 0} \frac{ \Pr(p_i \in [x, x + \delta x], i \notin S_{\kappa[i]})}{\delta x} \\
            &= \lim_{\delta x\to 0}  \frac{ \Pr(p_i \in [x, x + \delta x] \mid i \notin S_{\kappa[i]})}{\delta x} \Pr(i \notin S_{\kappa[i]}) \\
            &= \lim_{\delta x\to 0}  \frac{ \Pr(p_i \in [x, x + \delta x] \mid i \notin S_{\kappa[i]})}{\delta x} \left( 1 - \Pr(i \in S_{\kappa[i]}) \right) \\
            &= \lim_{\delta x\to 0}  \frac{ \Pr(p_i \in [x, x + \delta x] \mid i \notin S_{\kappa[i]})}{\delta x} \left(1 - \sum_{s = 0}^{b_{\kappa[i]}} \Pr(i \in S_{\kappa[i]} \mid s_{\kappa[i]} = s) \Pr(s_{\kappa[i]} = s) \right)\\
            &= \lim_{\delta x\to 0} (1/\delta x) \frac{\delta x}{(1 - \alpha b_{\kappa[i]}/m)} \left(1 - \sum_{s = 1}^{b_{\kappa[i]}} \left(s/b_{\kappa[i]} \right) \left(\alpha b_{\kappa[i]}/m \cdot (1/s) \right) \right)\\
               &= \lim_{\delta x\to 0} 1 \\
               &= 1.
    \end{align*}
    Finally, we lower bound the FDR of BH under the global null, which is equal to the probability of rejecting any hypothesis. Since BH rejects if any $s_k > 1$, and since the $s_1, \dots, s_{\bar{k}}$ are jointly independent, we can write
    \begin{align*}
        \FDR(\cR^\BH_\alpha) 
            &\geq \Pr(s_k > 0 \text{ for some $k \in \{1, \dots, \bar k\}$}) \\
            &= 1 - \Pr(s_k = 0 \text{ for all $k \in \{1, \dots, \bar k\}$}) \\
            &= 1 - \prod_{k = 1}^{\bar k} \left( 
                1 - \frac{\alpha b_k}{m}  \sum_{s = 1}^{b_k} 1/s
            \right).
    \end{align*}
    
\end{proof}


\section{Further examples}\label{app:examples}

\subsection{A naive adjustment of BH.}\label{app:naive}

    Consider the procedure 
    \[
        \cR_\alpha^{\mathrm{naiv}_\mathbb{D}}(p) = \{ i : p_i \leq \alpha |\cR^\BH_\alpha(\bone^{\Nio} p)|/m \},
    \]
    which is a simple modification of the BH procedure's thresholds which makes it neighbor-blind, and it is monotone as well. But it loses the self-consistency property and does not control FDR under dependency graph $\mathbb{D}$, which we can demonstrate with 
    $3$ $p$-values. 
    \naivefails*

    \begin{proof}

        Let $\mathbb{D}$ be the graph with nodes $\{1,2,3\}$ and, aside from self-edges, only a single edge between $\{2, 3\}$. Then apply the construction of Section \ref{app:bydlower} with the cover $\{\{1\}, \{2,3\}\}$, which has dependency graph $\mathbb{D}$. 

        When $p \sim P$ is sampled in this way, then under the global null, 
        \[
            \FDR_P(\cR_\alpha^{\mathrm{naiv}_\mathbb{D}}(p)) = \alpha \left(1 + \frac{2 \alpha^2 (1 - \alpha)}{3(3 - 2\alpha)^2} \right) > \alpha
        \]    
        by reading off of Table \ref{tab:NBcounterex}.
    \end{proof}

    To understand this failure more intuitively, it is revealing to inspect the FDR contribution terms 
    $\FDR_i(\cR_\alpha) := \mathbb{E}\big[ 1\{i \in \cR_\alpha(p)\}/|\cR_\alpha(p)| \big]$. In 
    this case, because node $1$ has no other nodes in its neighborhood, $N_1^\circ = \varnothing$ and 
    \begin{align}
        \FDR_1(\cR_\alpha) 
            &= \mathbb{E}\left[ \frac{1\{1 \in \cR^\BH_\alpha(p)\}}{|\cR_\alpha(p)|} \right] \nonumber \\
            &= \mathbb{E}\left[ \frac{1\{p_1 \leq \alpha|\cR^\BH_\alpha(p)|/m\}}{|\cR_\alpha(p)|} \right] \nonumber \\ 
            &> \alpha/m, \label{eq:fdr1counterex}
    \end{align}
    where in this case $m = 3$ $p$-values. Our main approach, used in Theorem \ref{thm:control} and \ref{thm:pprdcontrol}, would have been to show that $\FDR_i(\cR_\alpha) \leq \alpha/m$ for every $i$. 
    
    Roughly, the set $\cR^\BH_\alpha(\bone^{N_1^\circ} p) = \cR^\BH_\alpha(p)$ used in $p_1$'s local threshold does not account 
    for what the other agents can ``see'': that is, $p_2$ and $p_3$ cannot see each other's $p$-values in their thresholds, which can lead to fewer rejections such that $\cR^\BH_\alpha(\bone^{N_1^\circ} p)$ fails to be a rejection lower bound. Overall, $\cR^\BH_\alpha(\bone^{N_1^\circ} p) > |\cR_\alpha(p)|$ too often, and $H_1$ is rejected too frequently, so $\FDR_1(\cR_\alpha) > \alpha/m$. 

    In this specific example, we also have $\FDR_2(\cR_\alpha) = \FDR_3(\cR_\alpha) < \alpha/m$, but the total $\FDR(\cR_\alpha) = \sum_{i = 1}^3 \FDR_i(\cR_\alpha)$ still exceeds $\alpha$. 
        \begin{table}[tbp]
        \footnotesize
        \renewcommand{\arraystretch}{2}
        \begin{tabular}{c|cccc}
        \toprule
        {\small\textbf{Events}} & $\{p_1 \in A_1\}$ & $\{p_1 \in A_2\}$ & $\{p_1 \in A_3\}$ & $\{p_1 \in A_\times\}$ \\
        \midrule
        $\{p_2 \in A_2, p_3 \in A_2\}$ & $\dfrac{\alpha^2}{9}$, $\{1,2,3\}$ & $\dfrac{\alpha^2}{9}$, $\{1,2,3\}$ & $\dfrac{\alpha^2}{9}$, $\{1\}$ & $\dfrac{\alpha(1-\alpha)}{3}$, $\{\}$ \\
        $\{p_2 \in A_1, p_3 \in A_3\}$ & $\dfrac{\alpha^2 \theta}{9}$, $\{1,2\}$ & $\dfrac{\alpha^2 \theta}{9}$, $\{1,2\}$ & $\dfrac{\alpha^2 \theta}{9}$, $\{1,2\}$& $\dfrac{\alpha(1-\alpha)\theta}{3}$, $\{2\}$ \\
        $\{p_2 \in A_1, p_3 \in A_\times\}$ & $\dfrac{\alpha^2 (1-\theta)}{9}$, $\{1,2\}$ & $\dfrac{\alpha^2 (1-\theta)}{9}$, $\{1,2\}$ & $\dfrac{\alpha^2 (1-\theta)}{9}$, $\{2\}$ & $\dfrac{\alpha(1-\alpha)(1-\theta)}{3}$, $\{2\}$ \\
        $\{p_2 \in A_3, p_3 \in A_1\}$ & $\dfrac{\alpha^2 \theta}{9}$, $\{1,3\}$ & $\dfrac{\alpha^2 \theta}{9}$, $\{1,3\}$ & $\dfrac{\alpha^2 \theta}{9}$, $\{1,3\}$ & $\dfrac{\alpha(1-\alpha)\theta}{3}$, $\{3\}$ \\
        $\{p_2 \in A_\times, p_3 \in A_1\}$ & $\dfrac{\alpha^2 (1-\theta)}{9}$, $\{1,3\}$ & $\dfrac{\alpha^2 (1-\theta)}{9}$, $\{1,3\}$ & $\dfrac{\alpha^2 (1-\theta)}{9}$, $\{3\}$ & $\dfrac{\alpha(1-\alpha)(1-\theta)}{3}$, $\{3\}$ \\
        $\{p_2 \in A_3, p_3 \in A_3\}$ & $\dfrac{\alpha(1-\alpha)\theta^2}{3}$, $\{1\}$ & $\dfrac{\alpha(1-\alpha)\theta^2}{3}$, $\{1\}$ & $\dfrac{\alpha(1-\alpha)\theta^2}{3}$, $\{1\}$ & $(1-\alpha)^2 \theta^2$, $\{\}$ \\
        $\{p_2 \in A_3, p_3 \in A_\times\}$ & $\dfrac{\alpha(1-\alpha)\theta(1-\theta)}{3}$, $\{1\}$ & $\dfrac{\alpha(1-\alpha)\theta(1-\theta)}{3}$, $\{\}$ & $\dfrac{\alpha(1-\alpha)\theta(1-\theta)}{3}$, $\{\}$ & $(1-\alpha)^2 \theta(1-\theta)$, $\{\}$ \\
        $\{p_2 \in A_\times, p_3 \in A_3\}$ & $\dfrac{\alpha(1-\alpha)\theta(1-\theta)}{3}$, $\{1\}$ & $\dfrac{\alpha(1-\alpha)\theta(1-\theta)}{3}$, $\{\}$ & $\dfrac{\alpha(1-\alpha)\theta(1-\theta)}{3}$, $\{\}$ & $(1-\alpha)^2 \theta(1-\theta)$, $\{\}$ \\
        $\{p_2 \in A_\times, p_3 \in A_\times\}$ & $\dfrac{\alpha(1-\alpha)(1-\theta)^2}{3}$, $\{1\}$ & $\dfrac{\alpha(1-\alpha)(1-\theta)^2}{3}$, $\{\}$ & $\dfrac{\alpha(1-\alpha)(1-\theta)^2}{3}$, $\{\}$ & $(1-\alpha)^2 (1-\theta)^2$, $\{\}$ \\
        \bottomrule
        \end{tabular}
        \caption{Table of joint probabilities and rejection sets for the events in the row and column headers, given the joint distribution and multiple testing procedure of Example XX. We set $A_k = [\alpha(k-1)/3, \alpha k/3]$, $A_\times = [\alpha,1]$, and $\theta = \frac{\alpha/3}{1 - 2\alpha/3}$.}
        \label{tab:NBcounterex}
    \end{table}

\subsection{The sufficient conditions are not necessary.}\label{app:notnecessary}

\begin{example}[Neighbor-blindness is not necessary]
    The BY procedure is not neighbor-blind, but controls FDR under any dependency graph $\mathbb{D}$. 
\end{example}

\begin{example}[Monotonicity is not necessary]
    The procedure that rejects $H_i$ if and only if $p_i \in [\alpha/m, 2\alpha/m]$
    where $m$ is the number of hypotheses is not monotone, but controls FDR under any dependency graph $\mathbb{D}$. 
\end{example}

\begin{example}[Self-consistency is not necessary]
    Applying conditional calibration as described by the discussion 
    surrounding Equation \eqref{eq:bettercondcal} controls FDR under a dependency graph $\mathbb{D}$, 
    but is not self-consistent in general. 
\end{example}

\subsection{$\text{IndBH}^{(\infty)}$ is not optimal.}\label{app:indbh-notoptimal}

Figure \ref{fig:IndBH_not_optimal} demonstrates that $\IndBH^{(\infty)}(\alpha)$ is 
different from $\SU(\alpha)$ in general. 

\begin{figure}[ht]
  \centering

  \begin{minipage}[t]{0.96\textwidth}
  \centering
  \begin{tikzpicture}[scale = 1.5, every node/.style={circle, draw, minimum size=10mm, font = \Large}]
    \node (1) at (1.5,0.8) {1};
    \node (2) at (1.5,-0.8) {2};
    \node (3) at (3,0) {3};
    \node (4) at (4.5,0.8) {4};
    \node (5) at (4.5,-0.8) {5};

    \draw (3) -- (1) -- (2) -- (3);
    \draw (3) -- (4);
    \draw (3) -- (5);

    \node[above=-10pt of 1, draw=none, inner sep=0pt, font=\normalsize] {$p_1 = 0.00$};
    \node[below=-10pt of 2, draw=none, inner sep=0pt, font=\normalsize] {$p_2 = 0.00$};
    \node[below=-10pt of 3, draw=none, inner sep=0pt, font=\normalsize] {$p_3 = 0.00$};
    \node[above=-10pt of 4, draw=none, inner sep=0pt, font=\normalsize] {$p_4 = 0.04$};
    \node[below=-10pt of 5, draw=none, inner sep=0pt, font=\normalsize] {$p_5 = 0.04$};  
  \end{tikzpicture}
  \end{minipage}
  \caption{When $\alpha = 0.05$, the $\IndBH^{(\infty)}(\alpha)$ procedure gives the same rejection set as 
  $\IndBH^{(2)}(\alpha)$---the fixed point is reached after only one iteration. This rejection set is $H_3, H_4, H_5$, but $\SU(\alpha)$ rejects everything.}
  \label{fig:IndBH_not_optimal}
\end{figure}
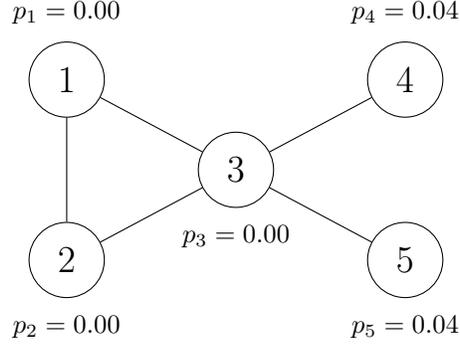

\section{Closure properties for $\mathbb{D}$-adapted procedures}\label{app:closure}

Given an initial procedure $\cR_\alpha$ and an independent set $I \in \Ind(\mathbb{D})$, 
define $N_I^\circ = \bigcup_{i \in I} N^\circ_i$, and
\[
    \cR^I_\alpha(p) = 
        \begin{cases}
            I \cup \cR_\alpha(p)  & \text{if } \forall i \in I,\ p_i \leq \alpha |I \cup \cR_\alpha(\bone^{N_I^\circ} p)|/m \\
            \cR_\alpha(p) & \text{otherwise}
        \end{cases}
\]
This can be related to the gap chasing update of Theorem \ref{thm:gapchase}. Let
\[
    \cR_\alpha^+(p) = \left\{i: p_i 
    \leq \frac{\alpha |\{i\} \cup \cR_\alpha(\bone^{\Nio} p)|}{m}  \right\},
\]
so $\cR_\alpha^+$ is the result of the update applied to $\cR_\alpha$.  We remark that if $\cR_\alpha$ is $\mathbb{D}$-adapted, then the formula
\[
    \cR_\alpha^+(p) = \bigcup_{i = 1}^m \cR^{\{i\}}_\alpha(p)
\]
holds---the inclusion 
$\cR_\alpha^+(p) \subset \bigcup_{i = 1}^m \cR^{\{i\}}_\alpha(p)$ can be 
seen directly, while the reverse inclusion 
$\cR_\alpha^+(p) \supset \bigcup_{i = 1}^m \cR^{\{i\}}_\alpha(p)$ can also 
be seen directly once we recall from the proof of Theorem \ref{thm:gapchase}\ref{item:gapchaseD}
that $\cR_\alpha^+(p) \supset \cR_\alpha(p)$ whenever $\cR_\alpha$ is $\mathbb{D}$-adapted. 

The next two propositions, combined with the above observation, 
give an alternate proof for Theorem \ref{thm:gapchase}\ref{item:gapchaseD}. However, we
frame them now as \emph{closure properties} of the family of $\mathbb{D}$-adapted procedures. 
Each defines an operation on $\mathbb{D}$-adapted procedures that returns a new 
such procedure. 

\begin{proposition}\label{prop:unionrule}
    Let $\cR^{(1)}_\alpha$ and $\cR^{(2)}_\alpha$ are $\mathbb{D}$-adapted, and define 
    $\cR^\textup{un}_\alpha(p) = \cR^{(1)}_\alpha(p) \cup \cR^{(2)}_\alpha(p)$. Then $\cR^\textup{un}_\alpha$ 
    is $\mathbb{D}$-adapted. 
\end{proposition}
\begin{proof}
    Proven as part of the proof of Proposition \ref{prop:IndBH1}. 
\end{proof}

\begin{proposition}\label{prop:extensionrule}
    If $\cR_\alpha$ is $\mathbb{D}$-adapted, then for any $I \in \Ind(\mathbb{D})$, $\cR^I_\alpha$ is $\mathbb{D}$-adapted. 
\end{proposition}
\begin{proof}
    Self-consistency holds because if $i \in \cR_\alpha(p)$, then 
    \[
        p_i \leq \alpha |\cR_\alpha(p)|/m \leq |\cR^I_\alpha(p)|/m
    \]
    by self-consistency of $\cR_\alpha$, and if $i \in \cR^I_\alpha(p)$ but
    not in $\cR_\alpha(p)$, we must have $I \cup \cR_\alpha(p) = \cR^I_\alpha(p)$ and
    \[
        p_i \leq \alpha |I \cup \cR_\alpha(\bone^{N_I^\circ} p)|/m \leq \alpha |I \cup \cR_\alpha(p)|/m \leq \alpha |\cR^I_\alpha(p)|/m.
    \]
    Monotonicity holds because $\cR_\alpha$ is monotone. Finally, neighbor-blindness holds because, if $i \notin I$, 
    then $i \in \cR^I_\alpha(\bone^{\Nio} p) \Rightarrow i \in \cR^I_\alpha(p)$ by the neighbor-blindness of $\cR_\alpha$, and if $i \notin I$, the event $\{\forall i \in I,\ p_i \leq \alpha |I \cup \cR_\alpha(\bone^{N_I^\circ} p)|/m\}$ does not depend on any $p$-values in $\Nio$, so again
    $i \in \cR^I_\alpha(\bone^{\Nio} p) \Rightarrow i \in \cR^I_\alpha(p)$. Monotonicity gives $i \in \cR^I_\alpha(\bone^{\Nio} p) \Leftarrow i \in \cR^I_\alpha(p)$. 
\end{proof}

Given this, a natural object would be the family of procedures \emph{generated} by these 
operations, given a starting set of procedures. But even if the initial procedure is trivial, 
we can derive non-trivial procedures. For example, 
if $\cR_\alpha$ always returns the empty set, applying Proposition \ref{prop:extensionrule}
with any independent set $I$ to it gives the set 
\[
    \cR^I_\alpha(p) = \begin{cases}
        I & \text{if } I = \cR^\BH_\alpha(\bone^{I^\mathsf{c}} p) \\
        \varnothing & \text{o.w.}
    \end{cases}. 
\]
We can repeat this for all $I' \subset I$ 
and use Proposition \ref{prop:unionrule} to union them all together, to get the set 
$\cR^\BH_\alpha(\bone^{I^\mathsf{c}} p) = \cup_{I' \subset I} \cR^I_\alpha(p)$. Then using
Proposition \ref{prop:unionrule} for all $I \in \Ind(\mathbb{D})$ 
gets $\cR^{\IndBH_\mathbb{D}}_\alpha(p)$. 

For a given $\alpha$, let $\mathbf{clo}_\mathbb{D}(\alpha)$ be the \emph{$\mathbb{D}$-adapted closure}---the family
of $\mathbb{D}$-adapted procedures that can be generated using Propositions \ref{prop:unionrule}
and \ref{prop:extensionrule}, starting from the set $\{\varnothing\}$, where
$\varnothing$ is the $\mathbb{D}$-adapted ``procedure'' that always returns the empty set. We have seen that for any $I \in \Ind(\mathbb{D})$, 
both $\cR^\BH_\alpha(\bone^{I^\mathsf{c}} (\cdot) )$ and $\cR^\IndBH_\alpha$ are in $\mathbf{clo}_\mathbb{D}(\alpha)$. However,           it contains more. 

\begin{proposition}
    For $\alpha \in [0,1)$, $\cR^{\SU_\mathbb{D}}_\alpha$ is contained inside $\mathbf{clo}_\mathbb{D}(\alpha)$. 
\end{proposition}

Because $\SU_\mathbb{D}$ is optimal, this suggests that the operations 
of Propositions \ref{prop:unionrule}
and \ref{prop:extensionrule}
in some sense maximally explore the space of $\mathbb{D}$-adapted procedures.


\begin{proof}
    Let $\cR_\alpha$ be $\mathbb{D}$-adaptive. Due to the optimality of $\SU_\mathbb{D}$, it is enough to show that 
    $\cR_\alpha(p) \subset \cR'_\alpha(p)$ for some $\cR'_\alpha \in \mathbf{clo}_\mathbb{D}(\alpha)$ for all $p$. 
    Our strategy will to be induct over subsets $K \subset [m]$ in decreasing cardinality, showing that for every $K$ there exists a procedure $\cR'_\alpha \in \mathbf{clo}_\mathbb{D}(\alpha)$ satisfying 
    \begin{equation}\label{eq:induct}
        \cR_\alpha(\bone^{\NKo} p) \subset \cR'_\alpha(\bone^{\NKo} p) \quad \text{for all $p$}. 
    \end{equation}
   The result follows by taking $K = \varnothing$. 
   
   For what follows, let $K^* \subset K$ denote the $\mathbb{D}$-\emph{singletons} in $K$---
    elements of $K$ which, seen as nodes in $\mathbb{D}$, are disconnected from all other nodes in $K$. (Note that $K^*$ is always an independent set.) Then for any $K \subset [m]$,
    \begin{equation}\label{eq:Nc2singleton}
        (\NKo)^\mathsf{c} \subset K^* \cup K^\mathsf{c}.
    \end{equation}
    First, let $A = [m]$, which will serve as an induction base case. Then
    \[
    \cR_\alpha(\bone^{\NAo} p) \subset \cR^\BH_\alpha(\bone^{\NAo} p) = \cR^\BH_\alpha(\bone^{{A^*}^\mathsf{c}}p),
    \]
    but we saw, in the discussion around the definition of $\mathbf{clo}_\mathbb{D}(\alpha)$, that if $\cR'(p) := \cR^\BH_\alpha(\bone^{{A^*}^\mathsf{c}}p) = \cR'(\bone^{{A^*}^\mathsf{c}}p)$, we know $\cR'$ is in $\mathbf{clo}_\mathbb{D}(\alpha)$, because $A^*$ is an independent set. 

    Next, assume that the induction hypothesis \eqref{eq:induct} holds for all $K$ of cardinality $k + 1$, and suppose that $A \subset [m]$ has cardinality $k$. Because $\alpha < 1$, self-consistency gives $i \in \cR_\alpha(p) \Rightarrow p_i < 1$, and we have the implications 
    \[
    i \in \cR_\alpha(\bone^{\NAo}p) \Rightarrow i \in (\NAo)^\mathsf{c} \Rightarrow i \in A^*  \text{ or } i \notin A,
    \]
    by \eqref{eq:Nc2singleton}. Then by neighbor-blindness of $\cR_\alpha$,
    \[
    \cR_\alpha(\bone^{\NAo}p) = \bigcup_{j \in \cR_\alpha(\bone^{\NAo}p)} \cR_\alpha(\bone^{\Njo}\bone^{\NAo}p) \subset A^* \cup \bigcup_{j \notin A} \cR_\alpha(\bone^{\Njo} \bone^{\NAo} p) \subset A^* \cup \cQ_\alpha(\bone^{\NAo} p)
    \]
    where $\cQ_\alpha(p) := \bigcup_{j \notin A} \cR^j_\alpha(p)$
    for some procedures $\cR_\alpha^j \in \mathbf{clo}_\mathbb{D}(\alpha)$ satisfying 
    $\cR^j_\alpha(\bone^{\Njo} \bone^{\NAo} p) \supset \cR_\alpha(\bone^{\Njo} \bone^{\NAo} p)$,
    which exist by induction. Because it takes unions, we have $\cQ_\alpha(p) \in \mathbf{clo}_\mathbb{D}(\alpha)$. 
    
    For a set $T \subset [m]$, let 
    \[
         \operatorname{SC}(T) = \bigcup \{S \subset T : \forall i \in S,\ p_i \leq \alpha |S|/m\}. 
    \]
    Whenever $T_1 \subset T_2$, it follows that $\operatorname{SC}(T_1) \subset \operatorname{SC}(T_2)$. Using this fact, 
    along with the self-consistency of $\cR_\alpha$, as well as of $\cQ_\alpha$ from Propositions \ref{prop:unionrule} and \ref{prop:extensionrule}, we have
    \[
        \cR_\alpha(\bone^{\NAo}p) = \operatorname{SC}(\cR_\alpha(\bone^{\NAo}p))  \subset 
        \operatorname{SC}(\cR_\alpha(\bone^{\NAo}p)) = B^* \cup \cQ_\alpha(\bone^{\NAo} p)
    \]
    for some subset $B^* \subset A^*$, such that for all $i \in B^*$, we have 
    $p_i \leq \alpha |\cQ_\alpha(\bone^{\NAo} p)|/m$.

    But then it follows that, because $\bone^{\NAo}\bone^{N_{B^*}^\circ} = \bone^{\NAo}$, and defining 
    \[
        \cQ_\alpha^{B^*}(p) := 
            \begin{cases}
                B^* \cup \cQ_\alpha(p) & \forall i \in B^*, p_i \leq \alpha|\cQ_\alpha(\bone^{N_{B^*}^\circ} p) |/m \\
                \cQ_\alpha(p) &\text{o.w.}
            \end{cases}
    \]
    we must have $\cQ_\alpha^{B^*}(\bone^{\NAo} p) = B^* \cup \cQ_\alpha(\bone^{\NAo} p)$, and hence $\cR_\alpha(\bone^{\NAo}p) \subset \cQ^{B^*}_\alpha(\bone^{\NAo}p)$.  This choice of $B^*$ works for a specific $p$, 
    but we can always take
    \[
        \cR'_\alpha(p) = \bigcup_{B^* \subset A^*} \cQ_\alpha^{B^*}(p)
    \]
    so that $\cR_\alpha(\bone^{\NAo} p) \subset \cR'_\alpha(\bone^{\NAo} p)$ for all $p$. Since we only used the operations 
    of Proposition \ref{prop:unionrule} and \ref{prop:extensionrule}, 
    $\cR'_\alpha \in \mathbf{clo}_\mathbb{D}(\alpha)$, we showed the inductive step as needed. 
\end{proof}

\section{More simulation results}\label{app:experiments}

Here we reproduce several more simulation settings by varying the parameters in Figures \ref{fig:block_uniformly} 
and \ref{fig:banded_clustered}. We also compare against two additional methods 
at level $\alpha = 0.1$, called  
\texttt{BYgraph} and \texttt{eBH} in the figures. 

\texttt{BYgraph} refers to the $\BH(\alpha')$ procedure where
\[
    \alpha' = \sup\left\{a : 1 - \left(1 - \frac{a b}{m} \sum_{s = 1}^b 1/s \right)^{m/b} \leq \alpha\right\},
\]
with $b = 100$ throughout to match the underlying simulation settings. This is based on the FDR lower bound derived in Appendix \ref{app:BYD}, 
specifically the version for block dependence where all blocks have size $b$. We do not show any formal FDR guarantees for this procedure. However, when the graph $\mathbb{D}$ does encode equally $b$-sized blocks, the power of \texttt{BYgraph} upper bounds the power of $\BH(\gamma^{-1}_\mathbb{D}(\alpha))$, where $\gamma_\mathbb{D}(\cdot)$ was defined in \eqref{eq:worstcasefdr}.

\texttt{eBH} refers to a way to use the eBH procedure of \citet{wangFalseDiscoveryRate2021}, modeled after their Example 2. We obtained e-values $e_i = \lambda p_i^{\lambda - 1}$ with $\lambda = 1/2$, 
and ran eBH with boosting factor $(2/\alpha)^{1/2}$. 
Effectively, this is just running the BH procedure with modified level 
$\alpha' = \sqrt{2 \alpha}$ and modified $p$-values $p'_i = 2 \sqrt{p_i}$. 

\begin{figure}[htbp]
    \centering
    \includegraphics[height = 0.3\textheight]{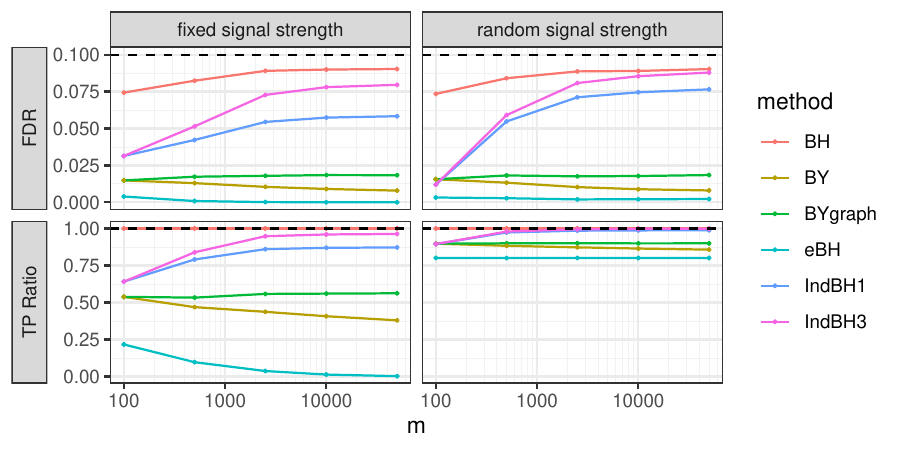}
    \caption{Results for block dependence with scattered signals. FDR control level set to $\alpha = 0.1$. Simulation parameters set to $\pi_0 = 0.9, \texttt{tarpow} = 0.6, \rho = 0.5, b = 100$. (This is the same setting as Figure \ref{fig:block_uniformly}). }\label{fig:block_uniformly_variant1}
\end{figure}

\begin{figure}[htbp]
    \centering
    \includegraphics[height = 0.3\textheight]{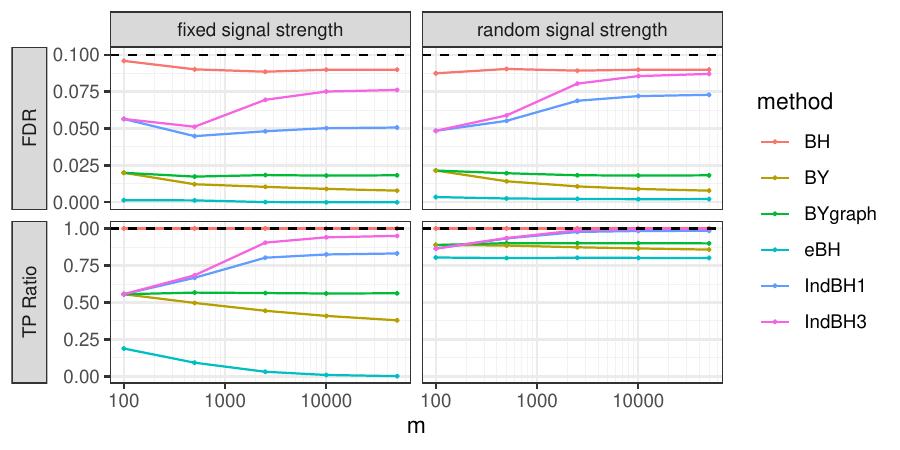}
    \caption{Results for banded dependence with clustered signals. FDR control level set to $\alpha = 0.1$. Simulation parameters set to $\pi_0 = 0.9, \texttt{tarpow} = 0.6, \rho = 0.5, b' = 100, \lambda_0 = 20, \tau = 6$. (This is the same setting as Figure \ref{fig:banded_clustered}.)}\label{fig:banded_clustered_variant1}
\end{figure}

\begin{figure}[htbp]
    \centering
    \includegraphics[height = 0.3\textheight]{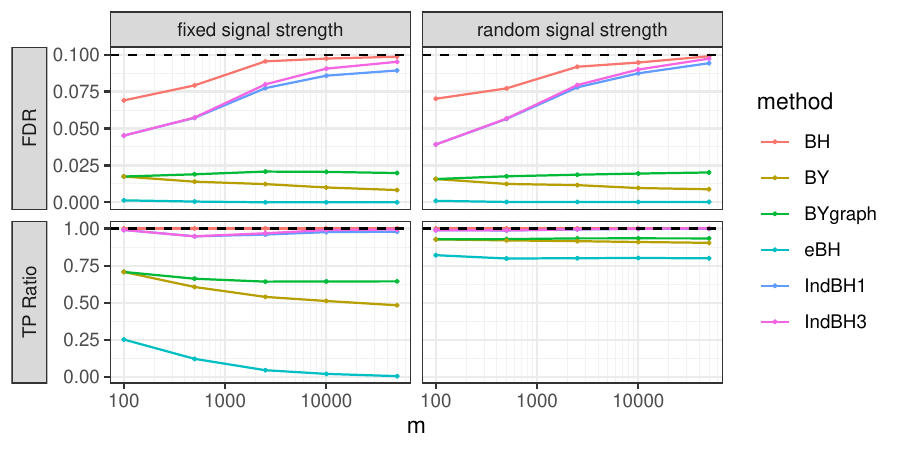}
    \caption{Results for block dependence with scattered signals. FDR control level set to $\alpha = 0.1$. Simulation parameters set to $\pi_0 = 0.99, \texttt{tarpow} = 0.6, \rho = 0.5, b = 100$. }\label{fig:block_uniformly_variant2}
\end{figure}

\begin{figure}[htbp]
    \centering
    \includegraphics[height = 0.3\textheight]{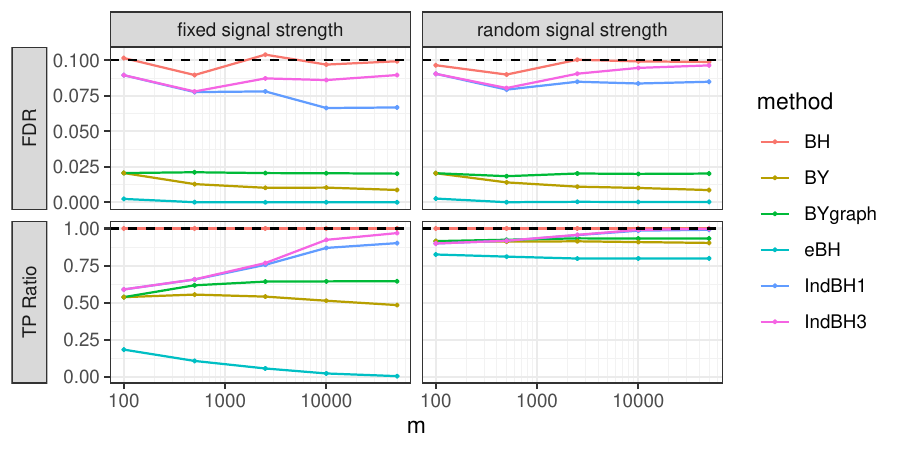}
    \caption{Results for banded dependence with clustered signals. FDR control level set to $\alpha = 0.1$. Simulation parameters set to $\pi_0 = 0.99, \texttt{tarpow} = 0.6, \rho = 0.5, b' = 100, \lambda_0 = 20, \tau = 6$. }\label{fig:banded_clustered_variant2}
\end{figure}

\begin{figure}[htbp]
    \centering
    \includegraphics[height = 0.3\textheight]{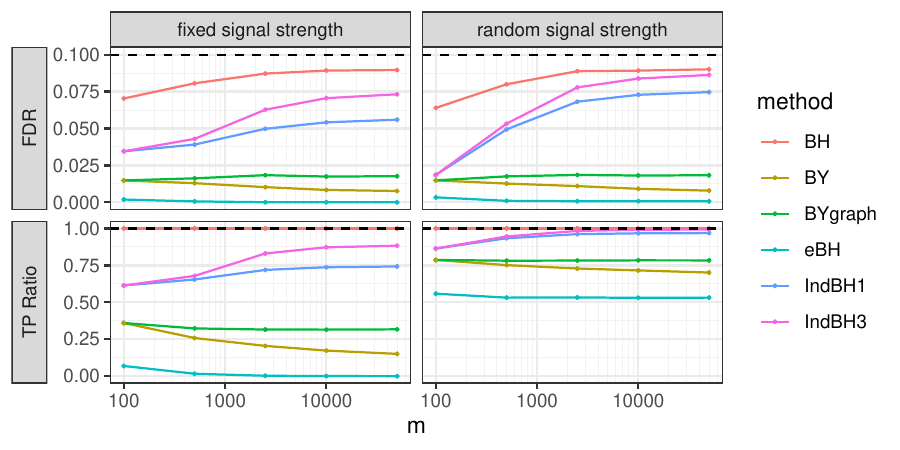}
    \caption{Results for block dependence with scattered signals. FDR control level set to $\alpha = 0.1$. Simulation parameters set to $\pi_0 = 0.9, \texttt{tarpow} = 0.3, \rho = 0.5, b = 100$.  }\label{fig:block_uniformly_variant3}
\end{figure}

\begin{figure}[htbp]
    \centering
    \includegraphics[height = 0.3\textheight]{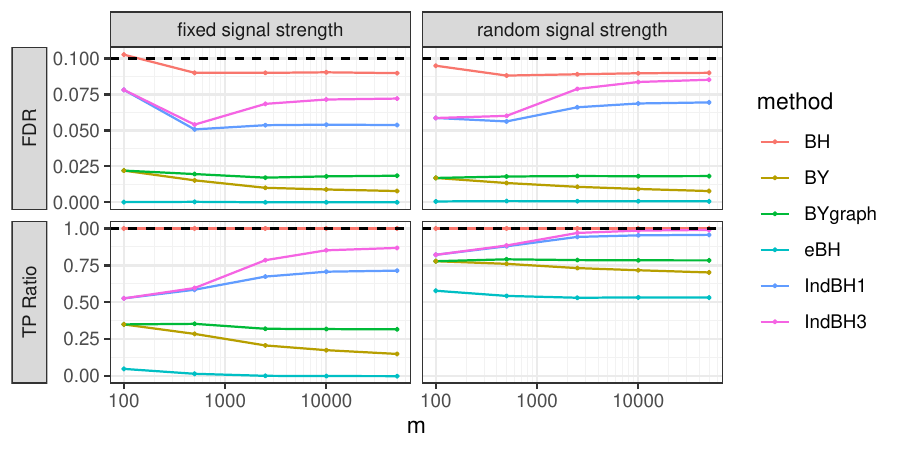}
    \caption{Results for banded dependence with clustered signals. FDR control level set to $\alpha = 0.1$. Simulation parameters set to $\pi_0 = 0.9, \texttt{tarpow} = 0.3, \rho = 0.5, b' = 100, \lambda_0 = 20, \tau = 6$. }\label{fig:banded_clustered_variant3}
\end{figure}

\begin{figure}[htbp]
    \centering
    \includegraphics[height = 0.3\textheight]{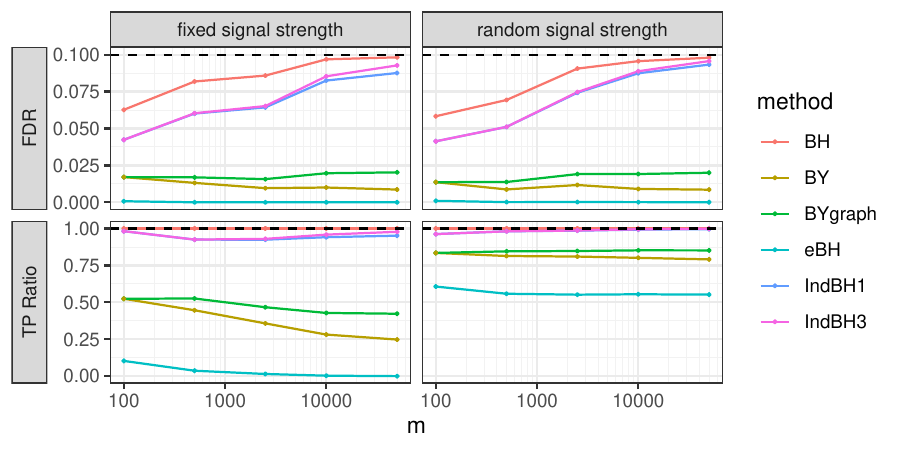}
    \caption{Results for block dependence with scattered signals. FDR control level set to $\alpha = 0.1$. Simulation parameters set to $\pi_0 = 0.99, \texttt{tarpow} = 0.3, \rho = 0.5, b = 100$. }\label{fig:block_uniformly_variant4}
\end{figure}

\begin{figure}[htbp]
    \centering
    \includegraphics[height = 0.3\textheight]{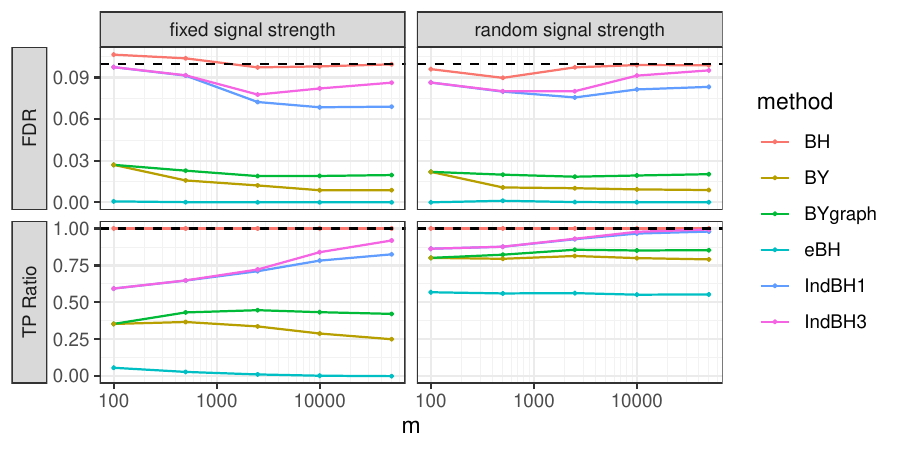}
    \caption{Results for banded dependence with clustered signals. FDR control level set to $\alpha = 0.1$. Simulation parameters set to $\pi_0 = 0.99, \texttt{tarpow} = 0.3, \rho = 0.5, b' = 100, \lambda_0 = 20, \tau = 6$.}\label{fig:banded_clustered_variant4}
\end{figure}

\end{document}